\documentclass[a4paper,cleveref, autoref, thm-restate,svgnames]{lipics-v2019}

\pdfoutput=1

\hideLIPIcs{}
\title{The Cost and Complexity of Minimizing Envy\\ in House Allocation} %TODO Please add

\titlerunning{The Complexity and Cost of Minimizing Envy in House Allocation} %TODO optional, please use if title is longer than one line

\author{Jayakrishnan Madathil}{University of Glasgow, UK}{jayakrishnan.madathil@glasgow.ac.uk}{}{Supported by the Engineering and Physical Sciences Research Council [EP/V032305/1].}%TODO mandatory, please use full name; only 1 author per \author macro; first two parameters are mandatory, other parameters can be empty. Please provide at least the name of the affiliation and the country. The full address is optional. Use additional curly braces to indicate the correct name splitting when the last name consists of multiple name parts.

\author{Neeldhara Misra}{Indian Institute of Technology, Gandhinagar}{neeldhara.m@iitgn.ac.in}{}{}

\author{Aditi Sethia}{Indian Institute of Technology, Gandhinagar}{aditi.sethia@iitgn.ac.in}{}{}

\authorrunning{Madathil et al.} %TODO mandatory. First: Use abbreviated first/middle names. Second (only in severe cases): Use first author plus 'et al.'

\Copyright{John Q. Public and Joan R. Public} %TODO mandatory, please use full first names. LIPIcs license is "CC-BY";  http://creativecommons.org/licenses/by/3.0/


 %TODO mandatory; please add comma-separated list of keywords

\category{} %optional, e.g. invited paper

\relatedversion{} %optional, e.g. full version hosted on arXiv, HAL, or other respository/website
\supplement{}%optional, e.g. related research data, source code, ... hosted on a repository like zenodo, figshare, GitHub, ...

\acknowledgements{}%optional

\nolinenumbers %uncomment to disable line numbering

\EventEditors{John Q. Open and Joan R. Access}
\EventNoEds{2}
\EventLongTitle{42nd Conference on Very Important Topics (CVIT 2016)}
\EventShortTitle{CVIT 2016}
\EventAcronym{CVIT}
\EventYear{2016}
\EventDate{December 24--27, 2016}
\EventLocation{Little Whinging, United Kingdom}
\EventLogo{}
\SeriesVolume{42}
\ArticleNo{23}
\usepackage{natbib}
\usepackage{charter}
\usepackage[table,xcdraw]{xcolor}
\usepackage[backgroundcolor=SkyBlue,linecolor=SkyBlue,obeyFinal]{todonotes}

\makeatletter
\providecommand\@dotsep{5}
\def\listtodoname{}
\def\listoftodos{\@starttoc{tdo}\listtodoname}
\makeatother

\newcounter{nmcomment}

\usepackage{fancybox}
\usepackage{parskip}
\usepackage{multirow}
\usepackage{graphicx}
\usepackage[table,xcdraw]{xcolor}
\usepackage{comment}
\excludecomment{shortver}
\includecomment{longver}
\usepackage{lscape,booktabs,multirow}
\usepackage{wrapfig}
\usepackage{tikz}
\usetikzlibrary{positioning}
\usetikzlibrary{arrows}
\usetikzlibrary{snakes}
\usetikzlibrary{decorations.pathmorphing}
\usetikzlibrary{decorations.markings}
\usetikzlibrary{decorations.pathreplacing}
\usetikzlibrary{shapes.geometric}

\usepackage{latexsym,amssymb,amsmath,mathtools,eulervm,xspace,nicefrac, amsthm}

\usepackage{etoolbox}
\makeatletter
\patchcmd{\BR@backref}{\newblock}{\newblock($\uparrow$~}{}{}
\patchcmd{\BR@backref}{\par}{)\par}{}{}
\makeatother

\hypersetup{
    colorlinks=true,
    linkcolor=DarkSlateBlue,
    filecolor=magenta,      
    urlcolor=DodgerBlue,
    citecolor = SteelBlue
}
\usepackage{framed}

\newcommand{\defproblem}[3]{
  \vspace{1mm}
\begin{center}
\noindent\fbox{

  \begin{minipage}{\textwidth}
  \begin{tabular*}{\textwidth}{@{\extracolsep{\fill}}lr} \textsc{#1}   \\ \end{tabular*}
  {\bf{Input:}} #2  \\
  {\bf{Question:}} #3
  \end{minipage}

  }
\end{center}
  \vspace{1mm}
}

\newcommand{\FPT}{\ensuremath{\mathsf{FPT}}\xspace}
\newcommand{\WOH}{\ensuremath{\mathsf{W}[1]}-hard\xspace}

\newcommand{\NP}{\ensuremath{\mathsf{NP}}\xspace}

\newcommand{\NPC}{\ensuremath{\mathsf{NP}}-complete\xspace}

\newcommand{\OHA}{\textsc{Optimal House Allocation}}
\newcommand{\UHA}{\textsc{Utilitarian House Allocation}}
\newcommand{\EHA}{\textsc{Egalitarian House Allocation}}

\newcommand{\card}[1]{\ensuremath{{\vert {#1} \vert }}}
\newcommand{\cO}{\mathcal{O}}
\newcommand{\fn}[3]{\ensuremath{{{#1} : {#2} \rightarrow {#3}}}}

\newcommand{\set}[1]{\ensuremath{\left\{ {#1} \right\}}}

\newcommand{\poly}{\ensuremath{\text{poly}}}
\newcommand{\opt}{\ensuremath{\mathsf{opt}}}

\newcommand{\ilpoha}{\ensuremath{\texttt{P1}(\mathcal{I})}}
\newcommand{\ilpeha}{\ensuremath{\texttt{P2}(\mathcal{I})}}

\begin{document}

\maketitle


\abstract{We study almost envy-freeness in house allocation, where $m$ houses are to be allocated among $n$ agents so that every agent receives exactly one house. An envy-free allocation need not exist, and therefore we may have to settle for relaxations. We study different aggregate measures of envy as markers of fairness. In particular, we define the amount of envy experienced by an agent $a$ w.r.t. an allocation to be the number of agents that agent $a$ envies under that allocation. We \emph{quantify} the envy generated by an allocation using three different metrics: 1) the number of agents who are envious; 2) the maximum amount of envy experienced by any agent; and 3) the total amount of envy experienced by all agents, and look for allocations that minimize one of the three metrics. We prove a host of algorithmic and hardness results. We also suggest practical approaches for these problems via integer linear program (ILP) formulations and report the findings of our experimental evaluation of ILPs. Finally, we study the price of fairness, which quantifies the loss of welfare we must suffer due to the fairness requirements, and present tight bounds as well as algorithms that simultaneously optimize both welfare and fairness.}

\maketitle

\textbf{Funding.} Jayakrishnan Madathil was supported by the Engineering and Physical Sciences Research Council [EP/V032305/1]. For the purpose of open access, the author has applied a Creative Commons Attribution (CC BY) licence to any Author Accepted Manuscript version arising.

\section{Introduction}
\label{sec:intro}

% \AS{We state all the results as Lemma. Should we call some of major ones as Theorems?}
% \JM{Yes.} \AS{Done}

% \AS{One review was regarding the choice of names -- Egalitarian and Utilitarian for envy minimization -- but they are usually used in the context of welfare maximization. Not sure if the terminology should be changed now when the extended abstract is already out. In the price of fairness section -- I am right now refraining to use Utilitarian social welfare -- just calling it social welfare.}
% \JM{It's probably better to change the names. But this should be a low priority. If we don't change the names, then we should add a sentence (after defining these problems) that we are concerned with fairness, and these terms are not to be confused with corresponding terms for welfare. I don't think having the extended abstract out should be a concern for changing or not changing the names.}

The \emph{house allocation problem} consists of $n$ agents and $m$ houses, where the agents have preferences over the houses, and we have to allocate the houses to the agents so that each agent receives exactly one house and each house is allocated to at most one agent.\footnote{Notice that these requirements immediately imply that we must have $m \geq n$, and under any allocation, exactly $n$ houses are allocated and the remaining $m - n$ houses remain unallocated.}
%The agents' preferences may be cardinal or ordinal. 
The problem captures scenarios such as assigning clients to servers, employees to offices, families to government housing, and so on. %In this paper, we focus on the special case when the preferences are binary valuations; that is, for every agent $a$ and every house $h$, either $a$ likes or dislikes $h$. 
We may think of the house allocation problem as a one-sided matching problem. Several variants of house allocation have been studied in the matching-under-preferences literature~[\cite{DBLP:books/ws/Manlove13,HZ1979,Z1990}], where the typical objectives have been economic efficiency requirements such as Pareto-optimality~[\cite{HZ1979,abraham2004pareto}], rank maximality [\cite{10.1145/1198513.1198520}] or strategyproofness~[\cite{DBLP:journals/algorithmica/KrystaMRZ19}]. Another useful and desirable objective in any resource allocation setting is \emph{fairness}, and we may equally well think of the house allocation problem as a special case of the fair division of indivisible goods setting, with the additional requirement that each agent be allocated exactly one good.  Given an allocation, say $\Phi: A \rightarrow H$, an agent $a$ envies $a^\prime$ if she values $\Phi(a^\prime)$ more than $\Phi(a)$. Finding ``envy-free'' allocations, i.e., ones where no agent envies another, is one of the main goals in fair division. Fairness in house allocation, in particular, has been a topic of interest in recent years~[\cite{GSV2019,KMS2021,HPSVV23,10.5555/3635637.3662936}]. 

% In our work, minimizing aggregate envy is the objective of consideration, in the setting of house allocation problems.

% Achieving only efficiency in the assignment may not be `fair' and an agent may wish to swap her house with the one who she envies. A prominent fairness notion considered in this regard is Envy-Freeness. An assignment of houses is \textsc{envy-free(EF)} if every agent values her house at least as much as she values any other assigned house.

In the context of the house allocation problem, note that if $m=n$, an envy-free allocation exists if and only if there is a perfect matching in the following bipartite graph: introduce a vertex for every agent and every house, and let the vertex corresponding to an agent be adjacent to all the houses that she values no less than any other house. Since every house must be assigned when $n = m$, when an agent is assigned anything short of her best option, she will be envious. Therefore, the existence of an envy-free allocation can be determined efficiently in this situation using standard algorithms for checking if a perfect matching exists.

The question is less obvious when $n<m$, i.e., when there are more houses than agents. Indeed, one could work with the same bipartite graph, but it is possible for the house allocation instance to admit an envy-free allocation even though the bipartite graph does not have a perfect matching. Consider a situation with three houses and two agents, where both agents value one house above all else, and the other two equally. While the graph only captures the contention on the highly valued house, it does not lead us directly to the envy-free allocation that can be obtained by giving both agents the houses that they value relatively less (but equally). 
It turns out that the question of whether an envy-free allocation exists can be determined in polynomial time even when $n<m$, by an algorithm of [\cite{GSV2019}] that involves iteratively removing subsets of contentious houses.

% They also show that an envy-free assignment exists with high probability if the number of houses exceeds the number of agents by a logarithmic factor. 

When an envy-free allocation does not exist at all, the natural objective is to resort to a relaxation of the fairness objective. However, house allocation differs from the typical fair division setting, with an additional constraint that every agent receives exactly one item. This constraint renders the well-studied relaxed notions of fairness like envy-freeness up to `some' good (where an agent chooses to hypothetically ignore one good from the envied bundle) futile. Hence we resort to a different kind of relaxation of envy-freeness: We \emph{quantify} the envy involved in an allocation using different aggregate measures such as the number of envious agents and the maximum number of agents any agent envies, and look for allocations that minimize these ``measures of envy.'' We note that~\cite{10.1007/978-3-642-41575-3_21} and~\cite{shams:hal-03336026} previously studied minimizing aggregate measures of envy in resource allocation problems. In the context of house allocation,~\cite{KMS2021} studied the problem of minimizing the number of envious agents. They showed that it is \NPC{} to find allocations that minimize the number of envious agents, even for binary utilities, and this quantity is hard to approximate for general utilities. In this paper, we explore envy minimization in house allocation from a broader perspective and prove algorithmic results not only for minimizing the number of envious agents but for two other measures of envy as well---minimizing the amount of maximum envy experienced by any agent and minimizing the amount of total envy experienced by all the agents put together. We say that the \emph{amount of envy} experienced by an agent $a$ is the number of agents she is envious of.

% In resource allocation problems, \cite{10.1007/978-3-642-41575-3_21} and \cite{shams:hal-03336026} looked at minimizing the envy by aggregating it in various ways. In house allocations,~\cite{KMS2021} initiated the problem of finding allocations that minimize the number of agents who experience envy and showed that it is \NPC{} to find allocations that minimize the number of envious agents, even for binary utilities, and this quantity is hard to approximate for general utilities. In this work, we explore the realm of envy minimization in house allocation from a broader perspective and give hardness and algorithmic results not only for minimizing the number of envious agents but also for minimizing the maximum envy of any agent and minimizing the total envy experienced by all the agents put together.

We remark here that minimizing the number of envious agents may lead to a sub-optimal allocation in terms of maximum envy and vice-versa. For instance, consider an instance with $4$ agents $a, b, c, d$ and $4$ houses $h_1, h_2, h_3, h_4$ with the following rankings: 

\begin{center}
\begin{tabular}{ll}
$a:$ & $\mathbf{h_1} \succ h_2 \succ h_3 \succ h_4$ \\
$b:$ & $h_1 \succ h_2 \succ h_3 \succ \mathbf{h_4}$ \\
$c:$ & $\mathbf{h_2} \succ h_3 \succ h_4 \succ h_1$ \\
$d:$ & $\mathbf{h_3} \succ h_4 \succ h_1 \succ h_2$
\end{tabular}
\end{center}

Consider the highlighted allocation, denoted as $\Phi$. Only agent $b$ experiences envy under $\Phi$. This allocation effectively minimizes the number of envious agents. However, the envy experienced by the sole envious agent is substantial, as she envies all the other agents. Alternatively, the envy of agent $b$ could have been reduced to just $1$ under the allocation $\Phi'$, where agent $b$ receives house $h_2$ and agents $c$ and $d$ are allocated houses $h_3$ and $h_4$ respectively. While both $\Phi$ and $\Phi'$ are optimal in terms of minimizing overall total envy, $\Phi$ falls short when it comes to addressing maximum envy, whereas $\Phi'$ is not ideal for minimizing the number of agents who experience envy. This suggests that the three notions of envy minimization are not directly comparable in general and therefore demand individual scrutiny and analysis.

% \AS{Add example comparing OHA EHA UHA}

When our focus is on minimizing the envy, there can be a trade-off with regards to \textit{social welfare}, which is essentially the collective measure of individual agent utilities within any allocation. The method for aggregating these utilities can vary, including options such as the geometric mean (known as Nash), the summation of utilities, or the minimum utility of any agent, among others. We restrict our attention to the sum of the individual agent utilities as our measure of social welfare for this work. For this measure, the trade-off between welfare and envy-minimization is illustrated as follows. Consider an instance with $2n-1$ houses such that each of the $n$ agents like only the first $n-1$ houses. Then the only envy-free allocation is to allocate the last $n$ houses to everyone, resulting in zero social welfare. On the contrast, if we allow for envious agents, the above instance can achieve a social welfare of at least $n-1$. This potential loss in welfare due to fairness guarantees is captured by \textit{price of fairness} \citep{BFT11price}, which is the worst-case ratio of the maximum social welfare in any allocation to that in a \text{fair} allocation. In this work, we give tight bounds for the price of fairness. Our investigation into PoF inspired us to look at simultaneously minimizing all three envy-minimization objectives while maximizing welfare. We show that we can indeed do this for $m = n$ and binary valuations. In this setting, we establish that there is an allocation that simultaneously maximizes welfare and minimizes the number of envious agents, maximum envy, and total envy. 

% We consider the price that an envy-minimizing allocation has to pay in terms of social welfare (the sum of the individual agent utilities).

\paragraph*{Our Contributions.} We propose to study the issue of ``minimizing envy'' from a broader perspective, and to this end we consider three natural measures of the ``amount of envy'' created by an allocation: a) the total number of agents who experience envy (discussed above), b) the envy experienced by the \emph{most} envious agent, where the amount of envy experienced by an agent is simply the number of agents that she is envious of, and c) the total amount of envy experienced by all agents. We refer to the questions of finding allocations that minimize these three measures of envy the~\OHA{} (OHA),~\EHA{} (EHA), and~\UHA{} (UHA) problems, respectively. A summary of our main results is in \Cref{tab:my-table} and \Cref{tab:pof_table}. 

\textbf{Hardness Results.}
We show the (parameterized) hardness of OHA and EHA even under highly restricted input settings. We show that OHA is \NPC{} even on instances where every agent values at most two houses. Further, it is \WOH{} when parameterized by $k$, the number of agents who are allowed to be envious (which implies that it is unlikely to admit a $f(k) (n + m)^{O(1)}$ time algorithm). As for EHA, we show that it is \NPC{}, even on instances where every agent values at most two houses \emph{and} every house is approved by a constant number of agents. In fact, we achieve this hardness even when the maximum allowed envy is just \emph{one}, establishing that the problem is para-\NP-hard when parameterized by $k$, the maximum envy (which implies that it is NP-hard even for a constant value of the parameter). The (parameterized) complexity of UHA, however, remains open.

\textbf{Algorithmic and Experimental Results.} Despite the hardness results even under the restricted input settings mentioned above, we explore tractable scenarios and prove a number of positive results. Observe that in a given instance, the number of houses$(m)$ could be much larger than the number of agents($n$). But we show that all three problems admit \emph{polynomial-time pre-processing algorithms} that reduce the number of houses; in particular, after this pre-processing, we will have the guarantee that $m \leq 2(n - 1)$. This result, in parameterized complexity parlance, means that all three problems admit \emph{polynomial kernels} when parameterized by the number of agents. To prove this, we use a popular tool called the \emph{expansion lemma}. While the kernels are interesting in their own right, we also leverage them to design polynomial-time algorithms for all three problems on binary \emph{extremal} instances. 
An instance of OHA/EHA/UHA is extremal if the houses can be ordered in such a way that every agent values either the first few houses or the last few houses in the ordering; that is, there is an ordering $(h_1, h_2,\ldots, h_m)$ of the houses such that for every agent $a$, there is an index $i(a)$ with $0 \leq i(a) \leq m$ such that $a$ either values the houses $\{h_1, h_2,\ldots, h_{i(a)}\}$ or $\{h_m, h_{m - 1},\ldots, h_{m - i(a)}\}$. We note here that extremal instances, although restrictive, form a non-trivial subclass for demonstrating tractability and have been studied in the literature~[\cite{10.5555/2832415.2832529}]. The hardness of the optimization problem even in the binary setting motivates to look at the structured binary preferences in the quest of tractability. In the context of house allocations, extremal instances appear where agents approvals are, for example, influenced by the distance of a house to either a hospital or a school, and the preferences decrease as the distance increases. In fact, despite the relatively simple structure of the preference profile, the obvious greedy approaches do not work and the three problems require three different lines of arguments.

Finally, we show that both OHA and EHA are fixed-parameter tractable (\FPT)\footnote{FPT w.r.t a parameter $\ell$ means the instance can be decided in time $f(\ell) \cdot (n+m)^{\mathcal{O}(1)}$, where $f$ is an arbitrary computable function.} when parameterized by the total number of house types or agent types; two agents are of the same type if they both like the same set of houses, and two houses are of the same type if they are both liked by the same set of agents. Notice that the number of (house or agent) types could potentially be much smaller than $m$ or $n$. The \FPT\ algorithms are obtained using ILP formulations with a bounded number of variables and constraints (for UHA, we obtain an integer quadratic program). The ILPs may also be of independent practical value. We implemented our  ILPs for OHA and EHA over synthetic datasets of house allocation, generated uniformly at random. For a fixed number of houses and agents, the results show that the number of envious agents and the maximum envy decreases as the number of agent types (the maximum
number of agents with distinct valuations) increases. Instances with identical valuations seem to admit more envy, attributed to the contention on the specific subset of houses. Also, when we increase the number of houses, for a constant number of agents and agent-types, the envy decreases, which is as expected, because of the increase in the number of choices and the fact that some houses (the more contentious ones) remain unallocated.

% \JM{Again, start this paragraph with a heading (something like ``Algorithmic results'' or ``Tractability results''. Also, open the paragraph with a prefatory sentence? Something like, ``Given the hardness results even under restricted input settings mentioned above, we look for islands of tractability, and prove a number of positive results.'')} \AS{Done}

\textbf{Price of Fairness.} Along with different measures of envy, we also focus on the social welfare of an allocation, as captured by the sum of the individual agent utilities. Minimizing the measures of envy can lead to economically inefficient allocations with poor social welfare. For example, an allocation that only allocates the ``dummy houses'' (i.e., houses that no agent values) is trivially envy-free, but has zero social welfare. Quantifying this welfare loss, incurred as the cost of minimizing envy is, therefore, an imperative consideration. We quantify this trade-off between welfare and envy minimization using a metric called the \emph{price of fairness (PoF)}; each measure of envy leads to its corresponding PoF, and we defer a formal definition of PoF to~\Cref{sec:pof}. We prove several tight bounds for PoF for binary valuations. In particular, we show that when $m>n$, PoF can be as large as the number of agents and that the bound is tight, for all the three envy measures of envy. We also identify the instances where no welfare has to be sacrificed in order to minimize envy, and hence no price has to be paid. In particular, we show that the price of fairness is $1$ for $m=n$ and binary valuations and also for $m>n$ and binary doubly normalized valuations. We show in particular that when $m=n$, there is an allocation that simultaneously minimizes the number of envious agents, the maximum envy, and total envy while maximizing social welfare. Moreover, we can compute such an allocation in polynomial time.

\paragraph*{Related Work.} 
\cite {SHAPLEY197423} studied the house allocation model with existing tenants, which holds crucial applications in domains like kidney exchanges. Scenarios encompassing entirely new applicants, as well as mixed scenarios with a few existing tenants have also been studied in the literature  [\cite{HZ1979}, \cite{ABDULKADIROGLU1999233}], and their practical implementations span diverse areas such as public housing and college dormitory assignments, among others.

The notion of fairness in house allocation setting was initiated by \cite{BCGHLMW2019} who studied a local variant for an equal number of agents and houses, where an agent can envy only those who are connected to her in a given social network. \cite{GSV2019} studied envy-freeness when the number of houses can be more than that of agents and gave an efficient algorithm that returns an envy-free solution if it exists. When such solutions do not exist,~\cite{KMS2021} initiated the study of finding allocations that minimize the number of agents who experience envy and showed that it is \NPC{} to find allocations that minimize the number of envious agents, even for binary utilities, and this quantity is hard to approximate for general utilities. Further, \cite{2021} studied the relaxed variant of assigning at most one house to every agent and give an $O(m \sqrt{n})$ algorithm for finding an envy-free matching of maximum cardinality in the setting of binary utilities. \cite{Shende2020StrategyproofAE} studied envy-freeness in conjunction with strategy-proofness. In more recent work, \cite{HPSVV23, 10.5555/3635637.3662936} have considered minimizing the sum of all pairwise envy values
over all edges in a social network. They proved structural and computational
results for various classes of underlying graphs on agents. \cite{DBLP:journals/corr/abs-2407-04664} looked at the degree of fairness while maximizing the social welfare and the size of an envy-free allocation. \cite{CHOO2024107103} discussed house allocations in the context of subsidies and showed that finding envy-free allocations with minimum subsidy is hard in general but tractable if agents have identical utilities or $m$ differs from $n$ by an additive constant.

The price of fairness was first proposed by \cite{BFT11price}, following which there has been substantial progress towards finding the bounds for the price for various combinations of fairness and welfare notions, specifically in resource allocation setting [\cite{CKK+12efficiency}, \cite{BLM+21price}, \cite{BBS20optimal}, \cite{SCD23equitability}, \cite{10.1007/978-3-031-43254-5_16}]. \cite{DBLP:journals/corr/abs-2407-04664} recently looked at the degree of fairness while maximizing the social welfare and the size of an envy-free allocation.

\paragraph*{Organization of the paper.} We discuss the results for OHA, EHA, and UHA in \Cref{sec:oha}, \Cref{sec:eha}, and \Cref{sec:uha} respectively. We discuss the experiments in \Cref{sec:exp}. Finally, we discuss the price of fairness in \Cref{sec:pof}, which is largely independent of all the other sections.

% \JM{In general, I prefer stating our results as early on in the introduction as possible. So, if we can rearrange the paragraphs to move up the ``Our Contribution'' section, that would be nice. But this is also a low priority.}

% \JM{Also, we probably should add another sentence about the work of Hosseini et al., since it's directly about UHA, if I'm not wrong.}

% \JM{Another low priority comment: While discussing Our contributions, it might not be a bad idea to separate rankings and binary utilities altogether. Say, focus on binary utilities, state our hardness and algorithmic results. And then mention that we prove a host of hardness results in the rankings settings.} \AS{Can do that, but there is a certain overlap in the way results are stated -- so separating them will make it a little repetitive. In addition, we only have two results for rankings, one of which is just a little tweak in the Kamiyama's result}

% \JM{Top priority comment: Mention our experimental results in the introduction, a couple of sentences about our general experimental finding? Also, let's do a spell check and grammar check.} \AS{Added the experimental results in the intro.}

% \usepackage{multirow}
% \usepackage{graphicx}
\bgroup
\def\arraystretch{1.5}% 

% \usepackage{multirow}
% \usepackage{graphicx}
\begin{table}
\centering
\resizebox{\textwidth}{!}{%
\begin{tabular}{|c|cccccc|}
\hline
{\color[HTML]{343434} }                                            & \multicolumn{5}{c|}{{\color[HTML]{343434} \textbf{Cardinal}}}                                                                                                                                                                                                                                                                                                                                                                                                                                                                                                                                                                                                                                                                                                                                                                                 & {\color[HTML]{343434} }                                                                                                                                                               \\ \hline
{\color[HTML]{343434} }                                            & \multicolumn{1}{c|}{\multirow{-1}{*}{{\color[HTML]{343434} \textbf{General}}}}                                                                                                & \multicolumn{4}{c|}{{\color[HTML]{343434} \textbf{Binary}}}                                                                                                                                                                                                                                                                                                                                                                                                                                                                                                                                                                                                                                      & {\color[HTML]{343434} }                                                                                                                                                               \\ \hline
\multirow{-3}{*}{{\color[HTML]{343434} \textbf{}}}                 &           \multicolumn{1}{c|}{{\color[HTML]{343434} }}                                             & \multicolumn{1}{c|}{{\color[HTML]{343434} \textbf{General}}}                                                                                                                                                     & \multicolumn{1}{c|}{{\color[HTML]{343434} \textbf{Extremal Intervals}}}                                                              & \multicolumn{1}{c|}{{\color[HTML]{343434} \textbf{d = 1}}}                                                                                & \multicolumn{1}{c|}{{\color[HTML]{343434} \textbf{d = 2}}}                                                                                                                                 & \multirow{-3}{*}{{\color[HTML]{343434} \textbf{Rankings}}}                                                                                                                            \\ \hline
\multicolumn{1}{|c|}{{\color[HTML]{343434} \textbf{\textsc{OHA}}}} & \multicolumn{1}{c|}{\cellcolor[HTML]{EFEFEF}{\color[HTML]{343434} \begin{tabular}[c]{@{}c@{}}NP-Complete\\ (by implication)\end{tabular}}} & \multicolumn{1}{c|}{\cellcolor[HTML]{ECCECE}{\color[HTML]{343434} \begin{tabular}[c]{@{}c@{}}NP-Complete ($\dagger$)\\ from \textsc{Clique}\\ \\ (\Cref{lem:oha-clique})\end{tabular}}} & \multicolumn{1}{c|}{\cellcolor[HTML]{B3E7B4}{\color[HTML]{343434} \begin{tabular}[c]{@{}c@{}}P\\ (\Cref{lem:oha-ext})\end{tabular}}} & \multicolumn{1}{c|}{\cellcolor[HTML]{B3E7B4}{\color[HTML]{343434} \begin{tabular}[c]{@{}c@{}}P\\ (\Cref{lem:oha-onehouse})\end{tabular}}} & \multicolumn{1}{c|}{\cellcolor[HTML]{ECCECE}{\color[HTML]{343434} \begin{tabular}[c]{@{}c@{}}NP-Complete\\ from \textsc{Clique}\\ \\ (\Cref{lem:oha-is})\end{tabular}}}           & \cellcolor[HTML]{ECCECE}{\color[HTML]{343434} \begin{tabular}[c]{@{}c@{}} NP-Complete\\ (\Cref{lem:oha-clique})\end{tabular}}                        \\ \hline
\multicolumn{1}{|c|}{{\color[HTML]{343434} \textbf{\textsc{EHA}}}} & \multicolumn{2}{c|}{\cellcolor[HTML]{EFEFEF}\begin{tabular}[c]{@{}c@{}}NP-Complete\\ (by implication)\end{tabular}}                                                                                                                                                                                                                                           & \multicolumn{1}{c|}{\cellcolor[HTML]{B3E7B4}{\color[HTML]{343434} \begin{tabular}[c]{@{}c@{}}P\\ (\Cref{lem:eha-ext})\end{tabular}}} & \multicolumn{1}{c|}{\cellcolor[HTML]{B3E7B4}{\color[HTML]{343434} \begin{tabular}[c]{@{}c@{}}P\\ (\Cref{lem:eha-onehouse})\end{tabular}}} & \multicolumn{1}{c|}{\cellcolor[HTML]{ECCECE}{\color[HTML]{343434} \begin{tabular}[c]{@{}c@{}}NP-Complete ($\star$)\\ from \textsc{Independent Set}\\ \\ (\Cref{lem:eha-is})\end{tabular}}} & \cellcolor[HTML]{ECCECE}{\color[HTML]{343434} \begin{tabular}[c]{@{}c@{}}NP-Complete ($\star$)\\ from \textsc{Multi-Colored Independent Set}\\ \\ (\Cref{lem:eha-mcis})\end{tabular}} \\ \hline
\multicolumn{1}{|c|}{{\color[HTML]{343434} \textbf{\textsc{UHA}}}} & \multicolumn{2}{c|}{\cellcolor[HTML]{FFFFC7}{\color[HTML]{343434} \textbf{?}}}                                                                                                                                                                                                                                                                                & \multicolumn{1}{c|}{\cellcolor[HTML]{B3E7B4}{\color[HTML]{343434} \begin{tabular}[c]{@{}c@{}}P\\ (\Cref{lem:uha-ext})\end{tabular}}} & \multicolumn{1}{c|}{\cellcolor[HTML]{B3E7B4}{\color[HTML]{343434} \begin{tabular}[c]{@{}c@{}}P\\ (\Cref{lem:uha-onehouse})\end{tabular}}} & \multicolumn{2}{c|}{\cellcolor[HTML]{FFFFC7}{\color[HTML]{343434} \textbf{?}}}                                                                                                                                                                                                                                                                                                     \\ \hline
\end{tabular}%
}
\caption{A partial summary of our results. Here, $d$ denotes the maximum number of houses approved by any agent. The results marked with a $\star$ refer to reductions that imply hardness even when the standard parameter is a constant, while the result marked with a $\dagger$ is a FPT reduction and also implies \textsf{W[1]}-hardness in the standard parameter.}
\label{tab:my-table}
\end{table}
\egroup

\section{Preliminaries}
\label{sec:prelims}
Let $[k]$ denote the set $\{1,2, \ldots, k\}$ for any positive integer $k$. 

An instance $\mathcal{I}$ of the \textsc{House Allocation} problem (HA) comprises a set $A=$ $\left\{a_{1}, a_{2}, \ldots, a_{n}\right\}$ of agents, a set $H=\left\{h_{1}, h_{2}, \ldots, h_{m}\right\}$ of houses and a preference profile (rankings or cardinal utilities) that capture the preference of all agents over the houses. An assignment or house allocation is an injection $\Phi: A \rightarrow H$. Throughout this section, let $\Phi$ be an arbitrary but fixed allocation. While we make all our notation explicit with respect to $\Phi$, during future discussions, the subscript $\Phi$ may be dropped if the allocation is clear from the context. There are a few different ways in which agents may express their preferences over houses, and we focus here on both linear orders as well as cardinal utilities. 

\paragraph*{Rankings} In this setting, each agent $a \in A$ has a linear order $\succ_a$ over the set of houses $H$. We will typically use $\succ_i$ to denote\footnote{In some of the reductions, the indices of the agents in $A$ are different from $[n]$, and we continue to adopt the convention that $\succ_\circ$ is used to describe the preferences of the agent $a_\circ$.} the preferences of agent $a_i$. The \emph{rank} of a house $h$ in the order~$\succ_a$ is one plus the number of houses $h^\prime$ such that $h^\prime \succ_a h$. For example, if $m = 3$ and $\succ_a$ is given by $h_2 \succ h_3 \succ h_1$, then the houses $h_2$, $h_3$ and $h_1$ have ranks $1, 2,$ and $3$ respectively. We denote the rank of a house $h$ in an order $\succ$ by $rk(h,\succ)$. A ranking is said to have ties ($rk(h,\succeq)$) if there is an agent who ranks some pair of houses equally. 

An agent $a \in A$ \emph{envies} an agent $b \in A$ under the allocation $\Phi$ if $\Phi(a) \prec_a \Phi(b)$, which is to say that $a$ perceives $\Phi$ to have allocated a house to $b$ that she ranks more than the one allocated to her. We use $\mathcal{E}_\Phi(a)$ to denote the set of agents $b$ such that $a$ envies $b$ in the allocation $\Phi$.  

An agent $a$ is \emph{envy-free} with respect to $\Phi$ if there is no agent $b$ such that $a$ envies $b$. In other words, $\mathcal{E}_\Phi(a) = \emptyset$. An allocation $\Phi$ is said to be \emph{envy-free} if all agents $a$ are envy-free with respect to $\Phi$. The \emph{amount of envy} experienced by an agent $a$ is the number of agents she is envious of, that is, $|\mathcal{E}_\Phi(a)|$ and is denoted by $\kappa_\Phi(a)$. Note that if an agent is envy-free, then the amount of envy experienced by her is zero.

\paragraph*{Binary Preferences} The utility that an agent $a$ derives from a house $h$ is denoted by $u_a(h)$. Preferences are said to be \emph{binary} if $u_{a}(h) \in\{0,1\}$ for all $a \in A$ and $h \in H$. We note here that binary utilities are a crucial subclass with simple elicitation and several works in fair division and voting literature have paid special attention to this case  [\cite{halpern2020fair,10.5555/3237383.3237392,lackner2023multi}]. A house $h$ is called a \emph{dummy} house if the utility of every agent for it is zero, that is, $u_a(h) = 0$ for all $a \in A$. The set of dummy houses is denoted by $D$. An agent $a$ is called a \emph{dummy} agent if it values every house at zero, that is, $u_a(h) = 0$ for all $h \in H$. The set of dummy agents is denoted by $D'$.

As previously, an agent $a \in A$ \emph{envies} an agent $b \in A$ under the allocation $\Phi$ if $u_a(\Phi(a)) < u_a(\Phi(b))$, which is to say that $a$ perceives $\Phi$ to have allocated a house to $b$ that she values more than the one allocated to her. That is, $u_a(\Phi(a)) = 0$ but $u_a(\Phi(b)) = 1$. The definition of $\mathcal{E}_\Phi(a)$ and the notion of \emph{envy-freeness} is the same as before. The amount of envy experienced by an agent is $|\mathcal{E}_\Phi(a)|$ and is denoted by $\kappa_\Phi(a)$. Just as with rankings, if an agent is envy-free, then the amount of envy experienced by her is zero. 

% \AS{Flow? The following definitions may require a sub-heading?}

Let $\mathcal{P}$ be a profile of binary utilities of agents $A$ over houses $H$. For an agent $a \in A$, use $\mathcal{P}(a)$ to denote the set of houses $h$ for which $u_a(h) = 1$. We say that these are houses that are valued by the agent $a$. For a subset $S \subseteq A$, we use $\mathcal{P}(S)$ to denote $\cup_{a \in S} \mathcal{P}(a)$. Similarly, for a house $h$, we use $\mathcal{T}(h)$ to refer to the set of agents who value $h$, and for a subset $S \subseteq H$, we use $\mathcal{T}(S)$ to denote $\cup_{h \in S} \mathcal{T}(h)$.

Two agents $a_p$ and $a_q$ are said to be of the same type if $\mathcal{P}(a_p) = \mathcal{P}(a_q)$ and two houses $h_p$ and $h_q$ are said to be of the same type if $\mathcal{T}(h_p) = \mathcal{T}(h_q)$. For an instance with $n$ agents and $m$ houses, we use $n^\star$ and $m^\star$ to denote the number of distinct types of agents and houses, respectively. 

%We define the \emph{utility matrix} associated with $\mathcal{P}$ as the $n \times m$ matrix $U_{\mathcal{P}}$ where $U_{\mathcal{P}}[i][j] = u_{a_i}(h_j)$. 
The \emph{preference graph} $G$ based on $\mathcal{P}$ is a bipartite graph defined as follows: the vertex set of $G$ consists of one vertex $v_a$ corresponding to every agent $a \in A$ and one vertex $v_h$ corresponding to every house $h \in H$; and $(v_a,v_h)$ is an edge in $G$ if and only if $a$ values $h$. 

\textbf{Extremal Interval Structure.} We say that $\mathcal{P}$ has an \emph{extremal interval structure with respect to houses} if there exists an ordering $\sigma$ of the houses such that for every agent $a$, $\mathcal{P}(a)$ forms a prefix or suffix of $\sigma$ [\cite{10.5555/2832415.2832529}].  Further, we say that $\mathcal{P}$ has a \emph{\textbf{left} (respectively, \textbf{right}) extremal interval structure with respect to houses} if there exists an ordering $\sigma$ of the houses such that for every agent $a$, $\mathcal{P}(a)$ forms a prefix (respectively, suffix) of $\sigma$. \\
Analogously, $\mathcal{P}$ has an \emph{extremal interval structure with respect to agents} if there exists an ordering $\pi$ of the agents such that for every house $h$, the set of agents who value $h$ forms a prefix or suffix of $\pi$. The notions of left and right extremal interval structures here are also defined as before. In our discussions, whenever we speak of an \emph{extremal interval structure} without explicit qualification, it is with respect to houses unless mentioned otherwise.

% Note that the notions of \emph{total envy} and \emph{maximum envy} are the same as before.

\paragraph*{Optimization Objectives}

We focus on the following optimization objectives.

\begin{enumerate}
    \item The \emph{number of envious agents} in an allocation $\Phi$ is the number of agents $a \in A$ for which $\kappa_\Phi(a) \geq 1$ and will be denoted by $\kappa^\#(\Phi)$. Further, given an instance $\mathcal{I}$  of the house allocation problem, we use $\kappa^\#(\mathcal{I})$ to denote the number of envious agents in an optimal allocation, that is, $\kappa^\#(\mathcal{I}) := \min_{\Phi}(\kappa^\#(\Phi))$.
     
    \item The \emph{maximum envy} generated by $\Phi$ is $\max_{a \in A} \kappa_\Phi(a)$ and is denoted by $\kappa^\dagger(\Phi)$.  As before, given an instance $\mathcal{I}$  of the house allocation problem, we use $\kappa^\dagger(\mathcal{I})$ to denote the maximum envy in an optimal allocation, that is, $\kappa^\dagger(\mathcal{I}):= \min_{\Phi}(\kappa^\dagger(\Phi))$.

    \item The \emph{total envy} generated by $\Phi$ is $\sum_{a \in A} \kappa_\Phi(a)$ and will be denoted by $\kappa^\star(\Phi)$. Again, given an instance $\mathcal{I}$  of the house allocation problem, we use $\kappa^\star(\mathcal{I})$ to denote the total envy in an optimal allocation, that is, $\kappa^\star(\mathcal{I}):= \min_{\Phi}(\kappa^\star(\Phi))$. 
\end{enumerate}

\paragraph*{Computational Questions}

We now formulate the computational problems that we would like to address. 

\defproblem{Optimal House Allocation}{A set $A=$ $\left\{a_{1}, a_{2}, \ldots, a_{n}\right\}$ of agents and a set $H=\left\{h_{1}, h_{2}, \ldots, h_{m}\right\}$ of houses, a preference profile describing the preferences of all agents over houses, and a non-negative integer $k \in \mathbb{Z}^+$.}{Determine if there is an allocation $\Phi$ with number of envious agents at most $k$, i.e, $\kappa^\#(\Phi) \leq k$.}

\defproblem{Egalitarian House Allocation}{A set $A=$ $\left\{a_{1}, a_{2}, \ldots, a_{n}\right\}$ of agents and a set $H=\left\{h_{1}, h_{2}, \ldots, h_{m}\right\}$ of houses, a preference profile describing the preferences of all agents over houses, and a non-negative integer $k \in \mathbb{Z}^+$.}{Determine if there is an allocation $\Phi$ with  maximum envy at most $k$, i.e, $\kappa^\dagger(\Phi) \leq k$.}

\defproblem{Utilitarian House Allocation}{A set $A=$ $\left\{a_{1}, a_{2}, \ldots, a_{n}\right\}$ of agents and a set $H=\left\{h_{1}, h_{2}, \ldots, h_{m}\right\}$ of houses, a preference profile describing the preferences of all agents over houses, and a non-negative integer $k \in \mathbb{Z}^+$.}{Determine if there is an allocation $\Phi$ with total envy at most $k$, i.e, $\kappa^\star(\Phi) \leq k$.}

We use $[\succ]$-OHA, $[\succeq]$-OHA, and $[\nicefrac{0}{1}]$-OHA to denote the versions of the OHA problem when the preferences are given, respectively, by linear orders, rankings with ties, and binary utilities, respectively. We adopt this convention for EHA and UHA as well. 

We note that for all three problems, the question of finding an envy-free allocation, i.e., one for which the optimization objective attains the value zero, is a natural special case. This amounts to finding an allocation where no agent has any envy for another and is therefore resolved (for both binary valuations and rankings) by the algorithm of~\cite{GSV2019} which uses an approach based on iteratively eliminating subsets that violate Hall's condition in the preference graph. 

We also observe here that all three problems are tractable for the special case when $m = n$. We assume without loss of generality, that every agent values at least one house. Observe that since $n = m$, all valid allocations have no unallocated houses. 

\begin{proposition}[folklore]
\label{prop:mequalsn-oha} $[\nicefrac{0}{1}]$-\OHA{} and $[\succeq]$-\OHA{} can be solved in polynomial time if $m = n$.
\end{proposition}

\begin{proof}
Let $\mathcal{I} := (A,H,\mathcal{P};k)$ be an instance of $[\nicefrac{0}{1}]$-OHA and let $G = (A \cup H; E)$ be the associated preference graph. We obtain an optimal allocation in polynomial time, and we return an appropriate output based on how the value of the optimum compares with $k$.

Let $M$ be a maximum matching in $G$.  We claim that any optimal allocation for $\mathcal{I}$ has $|M|$ envy-free agents. To see that there exists an allocation that has at least $|M|$ envy-free agents, consider the allocation that gives the house $M(a)$ to every agent $a$ saturated by $M$, and allocates the remaining houses arbitrarily among the agents not saturated by $M$. It is easy to see that this allocation has at least $|M|$ envy-free agents, namely the ones corresponding to those saturated by the matching $M$. On the other hand, suppose there is an allocation $\Phi$ with $k$ envy-free agents, then the envy-free agents must have received houses that they value---indeed, consider any agent $a$, and let $h$ be any house that $a$ values. If $a$ does not value the house $\Phi(a)$, then $a$ envies the agent who received the house $h$. Therefore, the set:
\[M := \{ (a,\Phi(a)) ~|~ a \text{ is envy-free with respect to } \Phi \}  \]
corresponds to a matching with $k$ edges in $G$, and this concludes the argument. Notice that this argument extends to weak orders by the natural extension of the notion of a preference graph: we have that an agent $a$ is adjacent to all houses $h$ that she prefers over all other houses. Therefore, we have the claim $[\succeq]$-OHA as well. 
\end{proof}

\begin{proposition}
\label{prop:mequalsn-eha} $[\nicefrac{0}{1}]$-\EHA{} and $[\succ]$-\EHA{} can be solved in polynomial time if $m = n$.
\end{proposition}

\begin{proof}
Let $\mathcal{I} := (A,H,\mathcal{P};k)$ be an instance of $[\nicefrac{0}{1}]$-EHA. We show that this problem reduces to finding a perfect matching among high-degree agents in the preference graph, based on the following observations.

\begin{enumerate}
    \item Suppose $a$ is an agent who values at most $k$ houses. Then, the amount of envy experienced by $a$ is at most $k$ in \emph{any} allocation.
    \item Suppose $a$ is an agent who values at least $k + 1$ houses. Then, if $\Phi$ is a valid solution, then $a$ must value the house $\Phi(a)$. 
\end{enumerate}
% Notice that any agent who values at most $\kappa^\dagger$ houses has at most $\kappa^\dagger$. On the other hand, suppose $a$ is an agent who values at least $\kappa^\dagger + 1$ houses. Then, if $\Phi$ is a valid solution, then $a$ must value the house $\Phi(a)$. 

It follows that $\mathcal{I}$ is a~\textsc{Yes}-instance if and only if the projection of the preference graph $G$ on $(A^\star \cup H)$ admits a perfect matching, where $A^\star$ is the subset of agents whose degree in $G$ is at least $k + 1$.

Now, let $\mathcal{I} := (A,H,\succ;k)$ be an instance of $[\succ]$-EHA. Consider the bipartite graph $G = (A \cup H; E^\star)$, where $(a,h)$ is an edge if and only if the rank of $h$ is at most $k+1$ in $\succ_a$.  We claim that $\mathcal{I}$ is a \textsc{Yes}-instance if and only if $G$ has a perfect matching.

Indeed, observe that the amount of envy experienced by any agent $a$ with respect to an allocation $\Phi$ is exactly one less than the rank of $\Phi(a)$ in $\succ_a$. Therefore, if $\Phi$ is an allocation whose maximum envy is $k$, then the rank of $\Phi(a)$ in $\succ_a$ must be at most $k + 1$ for all agents $a$. It is easy to check that such allocations are in one-to-one correspondence with perfect matchings in the graph $G$.
\end{proof}

\begin{proposition}
\label{prop:mequalsn-uha} $[\nicefrac{0}{1}]$\UHA{} and $[\succ]$-\UHA{} can be solved in polynomial time if $m = n$.
\end{proposition}

\begin{proof}
Let $\mathcal{I} := (A,H,\mathcal{P};k)$ be an instance of $[\nicefrac{0}{1}]$-UHA. To begin with, let $d(a)$ denote the degree of $a$ in the preference graph $G$ of $\mathcal{I}$. Now consider a complete bipartite graph $G^\star$ with bi-partition $(A \uplus H)$ and a cost function $c$ on the edges defined as follows:

 \begin{equation*}
    c({\color{DarkSlateGray}(a,h)}) =
    \begin{cases}
      d(a) & \text{if } a \text{ does not value } h,\\
      0 & \text{otherwise.}
    \end{cases}
\end{equation*}
Let $M$ be a minimum cost perfect matching in $G$ with total cost $t$. We claim that $\mathcal{I}$ is a \textsc{Yes}-instance if and only if $t \leq k$. In the forward direction, if $\Phi$ is an allocation with total envy at most $\kappa^\star$, then consider the following perfect matching in $G$:
\[M := \{ (a,\Phi(a)) ~|~ a \in A \}  \]

Notice that the cost of $M$ corresponds exactly to $\kappa^\star(\Phi)$, the total amount of envy in $\Phi$. This shows that there is a perfect matching in $G^\star$ with cost at most $k$. 

On the other hand, let $M$ be a perfect matching in $G^\star$ with the cost at most $k$, and let $M(a)$ denote the house $h$ such that $(a,h) \in M$. Then consider the allocation $\Phi$ given by $\Phi(a) = M(a)$ for all agents $a$. Notice that every zero-cost edge in $M$ corresponds to an envy-free agent with respect to $\Phi$, and every other edge $e = (a,h)$ corresponds to an agent in $\Phi$ who was allocated a house she did not value. Observe that the amount of envy experienced by $a$ in $\Phi$ is the number of houses she values, in other words, $d(a)$; however, this is also exactly the cost of the edge $e$. Therefore, it follows that the amount of envy in $\Phi$ is exactly equal to the cost of the matching $M$, and this concludes the proof of our claim. 

Now, let $\mathcal{I} := (A,H,\succ;k)$ be an instance of $[\succ]$-UHA.  As before, consider a complete bipartite graph $G^\star$ with bipartition $(A \uplus H)$. This time, we have the cost function $c$ on the edges defined as $c({\color{DarkSlateGray}(a,h)}) = rk(h, \succ_a) - 1$. This cost reflects the envy experienced by the agent $a$ if she were to be allocated the house $h$. Using arguments similar to the setting of binary valuations, it is easily checked that $\mathcal{I}$ is a \textsc{Yes}-instance if and only if there is a perfect matching in $G^\star$ whose cost is at most $k$. 
\end{proof}

% We also make note of the following fact, which is useful when dealing with ``sparse'' binary valuations.

% \begin{proposition}
% \label{prop:alldummyhousesallocated}
% If $(A,H,\mathcal{P})$ is an instance of~\OHA{} where every agent approves at most two houses, then there is an optimal allocation where all dummy houses are allocated. The result also holds for instances of~\EHA{} and~\UHA{}.
% \end{proposition}

% \nmtodo{Introduce the notion of types, extremal preferences, sparse preferences.}

% \defproblem{Weighted Utilitarian House Allocation}{A set $A=$ $\left\{a_{1}, a_{2}, \ldots, a_{n}\right\}$ of agents and a set $H=\left\{h_{1}, h_{2}, \ldots, h_{m}\right\}$ of houses, and a preference profile describing the preferences of all agents as linear orders over the houses.}{Find an allocation $\Phi$ that minimizes the total degree of envy, i.e, $\Delta^\star(\Phi)$.}

\paragraph*{Parameterized Complexity} 
Parameterized algorithms or multi-variate analysis is a popular perspective in the context of ``coping with computational hardness''. The key idea here is to segregate the running time of our algorithms into two parts: one that is polynomially bounded in the size of the entire input so that it is efficient on a quantity that is expected to be large in practice and the other, a computable function of a carefully chosen \emph{parameter}---and this component of the running time remains feasible in practice because the parameter is expected to be small. The parameterized perspective also allows us to formalize ideas about efficient preprocessing, and this is now an active subfield in its own right. We refer the readers to the books~\cite{DBLP:books/sp/CyganFKLMPPS15} and \cite{DBLP:series/txcs/DowneyF13} for additional background on this algorithmic paradigm, while recalling here only the key definitions relevant to our discussions. 

Formally, a parameterized problem $L$ is a subset of $\Sigma^{*} \times \mathbb{N}$ for some finite alphabet $\Sigma$. An instance of a parameterized problem consists of $(x, k)$, where $k$ is called the parameter. A central notion in parameterized complexity is fixed-parameter tractability (FPT), which means for a given instance $(x, k)$ solvability in time $f(k) \cdot p(|x|)$, where $f$ is an arbitrary function of $k$ and $p$ is a polynomial in the input size. The notion of kernelization is defined as follows.

\begin{definition} A kernelization algorithm, or in short, a kernel for a parameterized problem $Q \subseteq \Sigma^{*} \times \mathbb{N}$ is an algorithm that, given $(x, k) \in \Sigma^{*} \times \mathbb{N}$, outputs in time polynomial in $|x|+k$ a pair $\left(x^{\prime}, k^{\prime}\right) \in \Sigma^{*} \times \mathbb{N}$ such that (a) $(x, k) \in Q$ if and only if $\left(x^{\prime}, k^{\prime}\right) \in Q$ and (b) $\left|x^{\prime}\right|+k^{\prime} \leq g(k)$, where $g$ is an arbitrary computable function. The function $g$ is referred to as the size of the kernel. If $g$ is a polynomial function then we say that $Q$ admits a polynomial kernel.
\end{definition} 

On the other hand, we also have a well-developed theory of parameterized hardness. We call a problem para-NP-hard if it is NP-hard even for a constant value of the parameter. Further, we have the notion of \emph{parameterized reductions}, defined as follows.

\begin{definition}[Parameterized reduction] Let $A, B \subseteq \Sigma^{*} \times \mathbb{N}$ be two parameterized problems, A parameterized reduction from $A$ to $B$ is an algorithm that, given an instance $(x, k)$ of $A$, outputs an instance $\left(x^{\prime}, k^{\prime}\right)$ of $B$ such that:

1. $(x, k)$ is a yes-instance of $A$ if and only if $\left(x^{\prime}, k^{\prime}\right)$ is a yes-instance of $B$.

2. $k^{\prime} \leq g(k)$ for some computable function $g$, and

3. the running time is $f(k) \cdot|x|^{O(1)}$ for some computable function $f$.
\end{definition}

A parameterized reduction from a problem known to be \WOH{} is considered to be strong evidence that the target problem is not \FPT{}. The formal definition of W[1]-hardness is beyond the scope of this discussion, but we state in this section the known W[1]-hard problems from which we will perform parameterized reductions to obtain our results. 

We also note that an instance of integer linear programming is \FPT{} in the number of variables.

\begin{theorem}[\cite{DBLP:journals/mor/Lenstra83}]
\label{thm:lenstra}
An integer linear programming instance of size $L$ with $p$ variables can be solved using $\mathcal{O}\left(p^{2.5 p+o(p)} \cdot\left(L+\log M_{x}\right) \log \left(M_{x} M_{c}\right)\right)$ arithmetic operations and space polynomial in $L+\log M_{x}$, where $M_{x}$ is an upper bound on the absolute value a variable can take in a solution, and $M_{c}$ is the largest absolute value of a coefficient in the vector $c$.
\end{theorem}

Finally, we state the problems that we will use in the reductions. We note that all the problems below are \WOH{} when parameterized by $k$.

\begin{framed}
    \textsc{Clique} (respectively, \textsc{Independent Set})\\
    \textbf{Input:} A graph $G$ and an integer $k$.\\
    \textbf{Question:} Does there exist a subset $S \subseteq V(G)$ such that $G[S]$ is a clique (respectively, independent set) and $|S| \geq k$? 
\end{framed}
\begin{framed}
    \textsc{Maximum Balanced Biclique}\\
    \textbf{Input:} A graph $G = (L \cup R, E)$ and an integer $k$.\\
    \textbf{Question:} Does there exist a subset $S \subseteq L$ and $T \subseteq R$ such that $G[S \cup T]$ is a biclique and $|S| = |T| = k$? 
\end{framed}
% \begin{framed}
%     \textsc{Independent Set}\\
%     \textbf{Input:} A graph $G$ and an integer $k$.\\
%     \textbf{Question:} Does there exist a subset $S \subseteq V(G)$ such that $G[S]$ is an independent set and $|S| \geq k$? 
% \end{framed}
\begin{framed}
    \textsc{Multi-Colored Independent Set}\\
    \textbf{Input:} A graph $G = (V_1 \uplus \cdots \uplus V_k, E)$.\\
    \textbf{Question:} Does there exist a subset $S \subseteq V(G)$ such that $G[S]$ is an independent set and $|V_i \cap S| = 1$ for all $i \in [k]$? 
\end{framed}

% \nmtodo{Add W-hardness, FPT reduction, definition of Clique, IS, MCIS, and ILP}

% Can remove the forced page breaks later :)

\section{Pre-processing using Expansion Lemma}

In this section, we introduce the expansion lemma, a powerful and popular tool for kernelization, which we will use later to design our algorithms. Let $G$ be a bipartite graph with vertex bi-partitions $(A, B)$. A set of edges $M \subseteq E(G)$ is called an \emph{expansion} of $A$ into $B$ if:

\begin{itemize}
\item every vertex of $A$ is incident to exactly one edge of $M$;
\item $M$ saturates exactly $|A|$ vertices in $B$.
\end{itemize}

Note that an expansion saturates all vertices of $A$. 

\begin{lemma}[\textbf{Expansion lemma} \citep{DBLP:books/sp/CyganFKLMPPS15}]
\label{lem:expansion} Let $G$ be a bipartite graph with vertex bi-partitions $(A, B)$ such that

\begin{enumerate}
\item $|B| \geq |A|$, and
\item there are no isolated vertices in $B$.
\end{enumerate} 

Then there exist non-empty vertex sets $X \subseteq A$ and $Y \subseteq B$ such that
\begin{itemize}
\item there is an expansion of $X$ into $Y$, and
\item no vertex in $Y$ has a neighbor outside $X$, that is, $N(Y) \subseteq X$.
\end{itemize} 

Furthermore, the sets $X$ and $Y$ can be found in time polynomial in the size of $G$ (\cite{10.1145/1721837.1721848}).
\end{lemma}

We now consider an instance $\mathcal{I} = (A,H,\mathcal{P})$ of \textsc{HA} with binary valuations parameterized by $k$, where $k$ is one of $\kappa^\#$, $\kappa^\dagger$ or $\kappa^\star$. We introduce the following reduction rules here whose implementation is parameter-agnostic. 

Let $G = (A \cup H; E)$ denote the preference graph of $\mathcal{I}$. Note that we may assume that we have at most $(n-1)$ dummy houses in the instance~$\mathcal{I}$, since instances with at least $n$ dummy houses admit trivial envy-free allocations, where every agent is given a dummy house. We make this explicit in the following reduction rule:

\begin{reductionrule}
\label{rr1}
If $\mathcal{I}$ has at least as many dummy houses as agents, then return a trivial~\textsc{Yes}-instance. The parameter $k$ is unchanged.
\end{reductionrule}

Let $G^\star$ denote the preference graph induced by $(A \setminus D') \cup (H \setminus D)$, where $D$ denotes the vertices corresponding to dummy houses in $\mathcal{I}$ and $D'$ denotes the vertices corresponding to the dummy agents in $\mathcal{I}$. We now propose the following reduction rule based on the expansion lemma. The safety of the following reduction rule is argued separately for the three parameters in~\Cref{lem:oha-kernel-n,lem:eha-kernel-n,lem:uha-kernel-n}, respectively.

\begin{reductionrule}
\label{rr2}
In $\mathcal{I}$, if $|H \setminus D| \geq |A \setminus D'|$, then let $(X,Y)$ be as given by~\Cref{lem:expansion} applied to $G^\star$, and let $M \subseteq E(G^\star)$ be the associated expansion. Proceed by eliminating all agents and houses saturated by $M$.  The parameter $k$ is unchanged.
\end{reductionrule}

Note that once~Reduction Rules~\ref{rr1} and \ref{rr2} are applied exhaustively, we have that: \[|H| = |H \setminus D| + |D| \leq |A \setminus D'| + |D| \leq |A| + |A|-1 \leq 2\cdot(|A| - 1),\] where we are slightly abusing notation and using $H$ and $A$ to denote the houses and agents in the reduced instance. Thus, once the safety of these reduction rules is established, we conclude that all the three problems under consideration---\textsc{OHA}, \textsc{EHA} and \textsc{UHA}---admit polynomial kernels with $\cO(\card{A})$ houses when parameterized by the number of agents. 

After the above two reduction rules have been applied, we have $|H| \geq |A|$ and $|H \setminus D| \leq |A \setminus D'|$, we get $|H| - |H \setminus D| \geq |A|-|A \setminus D'|$ and hence $|D| \geq |D'|$. Therefore, the number of dummy houses are at least as much as the number of dummy agents. We then apply the following reduction rule.

\begin{reductionrule}
\label{rr5}
In a reduced instance $\mathcal{I}$ with respect to Reduction Rules \ref{rr1} and \ref{rr2}, proceed by allocating a dummy house to each of the dummy agents and eliminate $|D'|$ dummy agents and $|D'|$ dummy houses. The parameter $k$ is unchanged.
\end{reductionrule}

Towards establishing the safety of the above reduction rule, we argue that in the reduced instance, there is an optimal allocation where all the dummy agents receive a dummy house each. Indeed, if not, say a dummy agent $a_d$ receives a house $h \in H \setminus D$ in an optimal allocation. Since $|A \setminus D'|  > |H \setminus D|$, there is an agent $a \in A \setminus D'$ who received a dummy house $h_d$. We re-allocate $h_d$ to $a_d$ and allocate all the houses in $|H \setminus D|$ to agents in $|A \setminus D|$ by finding a maximum matching in the associated preference graph $G$, restricted to these houses and agents. This re-allocation either does not create any new envy ($a_d$ is indifferent and no one else if worse-off) or makes an additional agent in $A \setminus D$ envy-free (because of the allocation of $h$ to an agent who values it). Hence, either we get the desired optimal allocation under which all the dummy agents receive a dummy house or we contradict the fact that we started with an optimal allocation. So we can allocate dummy houses to the dummy agents and assume going forward that there are no dummy agents. That is, $|D'| =\phi$.

We now make a claim here that will be useful in the arguments that we make later about the safety of~\Cref{rr2}.

\begin{definition}
\label{def:nice}
Let $(A,H,\mathcal{P})$ and $X,Y$ and $M$ be as in the premise of Reduction Rule~\ref{rr2}. An allocation $\Phi$ is said to be \emph{good} if $\Phi(a) = M(a)$ for all $a \in X$, where $M(a)$ denotes the unique vertex $h \in Y$ such that $(a,h) \in M$. 
\end{definition}

\begin{claim}
\label{claim:niceallocalongexpansion}
Let $(A,H,\mathcal{P})$ and $X,Y$ and $M$ be as in the premise of~\Cref{rr2}. There is a good allocation that minimizes the number of envious agents.
\end{claim}

\begin{proof}
Let $\Phi$ be an allocation that minimizes the number of envious agents. If $\Phi$ is already good then there is nothing to prove. Otherwise, suppose $\Phi(a) \neq M(a)$ for some $a \in X$. If $\Phi$ does not assign the house corresponding to $M(a)$ to any agent, then consider the modified allocation $\Phi^\prime$ where we assign $M(a)$ to $a$ while letting $\Phi(a)$ become an unassigned house, that is:

 \begin{equation*}
    \Phi^\prime(c) =
    \begin{cases}
      M(a) & \text{if } c = a,\\
      \Phi(c) & \text{otherwise.}
    \end{cases}
\end{equation*}

On the other hand, suppose $b$ is such that $\Phi(b) = M(a)$. Then consider the modified allocation $\Phi^\prime$ where we swap the houses of $a$ and $b$, that is:

 \begin{equation*}
    \Phi^\prime(c) =
    \begin{cases}
      M(a) & \text{if } c = a,\\
      \Phi(a) & \text{if } c = b,\\
      \Phi(c) & \text{otherwise.}
    \end{cases}
\end{equation*}

We keep modifying the original allocation $\Phi$ in the manner described above until we arrive at a good allocation. Let $\Phi^\star$ denote this final allocation.  

% Observe that in every step, an agent  outside $X$ who was not envious with respect to $\Phi$ continues to be non-envious with respect to our final allocation, since in each step, they either did not value house that was added to the set of assigned houses ($M(a)$ in the first case), or the set of assigned houses remained unchanged from their perspective (if the agent was not involved in the exchange in the second case), or the house they had to relinquish in the exchange was a house that they did not value (if $b \notin X$ in the second case). 

Now consider an agent $c \notin X$. We claim that the amount of envy experienced by $c$ does not increase at any step of the process of morphing $\Phi$ to $\Phi^\star$. Consider the following cases that arise at any step, where $a \in X$ and by a slight abuse of notation, we use $\Phi$ to denote the allocation that is being modified:

\begin{enumerate}
    \item Suppose $M(a)$ is assigned to $a$ and $\Phi(a)$ is unassigned. We know that $c$ does not value $M(a)$ since $c \notin X$. If $c$ valued $\Phi(a)$, then her envy with respect to the new allocation will be one less than her envy with respect to the previous allocation. 
    \item Suppose $M(a)$ and $\Phi(a)$ are swapped between agents $a$ and $b$; and $c \neq b$. Then the amount of envy experienced by $c$ does not change.
    \item Suppose $M(a)$ and $\Phi(a)$ are swapped between agents $a$ and $b$; and $c = b$. If $c$ valued $\Phi(a)$, then her envy with respect to the new allocation will be less than her envy with respect to the previous allocation. On the other hand, if $c$ does not value $\Phi(a)$ then the amount of envy experienced by $c$ does not change.
\end{enumerate}

Also, in the final allocation $\Phi^\star$, all agents in $X$ are envy-free, since they are assigned houses that they value via the expansion $M$. Therefore, the total number of envious agents in the final allocation $\Phi^\star$ is the same as the number of envious agents in the original allocation $\Phi$---recall that $\Phi$ minimized the number of envious agents. 
\end{proof}

% \nmtodo{can omit the explanations below.}

The following claims can be shown by the same argument that was used for~\Cref{claim:niceallocalongexpansion}, since for any agent $a$, the amount of envy experienced by $a$ respect to in $\Phi^\star$ is at most the amount of envy experienced by $a$ with respect to $\Phi$.

\begin{claim}
\label{claim:niceallocmaxenvy}
Let $(A,H,\mathcal{P})$ and $X,Y$ and $M$ be as in the premise of~\Cref{rr2}. There is a good allocation that minimizes the maximum envy.
\end{claim}

% \begin{proof}
% Let $\Phi$ be an allocation that minimizes the maximum envy among agents. We proceed exactly as in the proof of the previous claim --- letting $\Phi^\star$ denote final allocation. From the previous argument, it follows that for any agent $a$, the amount of envy experienced by $a$ respect to in $\Phi^\star$ is at most the amount of envy experienced by $a$ with respect to $\Phi$. Therefore, the maximum envy experienced by any agent in $\Phi^\star$ is the same as the number of envious agents in the original allocation $\Phi$ --- recall that $\Phi$ minimized the maximum envy among agents. 
% \end{proof}

\begin{claim}
\label{claim:nicealloctotalenvy}
Let $(A,H,\mathcal{P})$ and $X,Y$ and $M$ be as in the premise of~\Cref{rr2}. There is a good allocation that minimizes total envy.
\end{claim}

% \begin{proof}
% Let $\Phi$ be an allocation that minimizes the total envy. We proceed exactly as in the proof of the previous claim --- letting $\Phi^\star$ denote final allocation. From the previous argument, it follows that for any agent $a$, the amount of envy experienced by $a$ respect to in $\Phi^\star$ is at most the amount of envy experienced by $a$ with respect to $\Phi$. Therefore, the total envy experienced by all agents in $\Phi^\star$ is the same as the number of envious agents in the original allocation $\Phi$ --- recall that $\Phi$ minimized the total envy among agents. 
% \end{proof}

\paragraph*{Applications to Extremal Instances}
We observe several properties of instances $\mathcal{I} = (A, H, \mathcal{P}; k)$ of $[\nicefrac{0}{1}]$-HA, where $\mathcal{P}$ has an extremal interval structure (with respect to the houses).  The properties are parameter-agnostic, and therefore, hold for instances of $[\nicefrac{0}{1}]$-OHA, $[\nicefrac{0}{1}]$-EHA and $[\nicefrac{0}{1}]$-UHA. As a shorthand, we say that $\mathcal{I} = (A, H, \mathcal{P}; k)$ is an extremal instance (or simply extremal) if $\mathcal{P}$ has the extremal interval structure.  

\begin{table}
\begin{tabular}{@{}c|c|ccc|cccc|ccc@{}}
\multicolumn{1}{l|}{}                                                           & \multicolumn{1}{l|}{} & \multicolumn{3}{c|}{Left-houses $(H_L)$} & \multicolumn{4}{c|}{Dummy houses $(D)$} & \multicolumn{3}{c}{Right-houses $(H_R)$} \\ \midrule
\multicolumn{1}{l|}{}                                                           &                       & $h_1$        & $h_2$       & $h_3$       & $h_4$    & $h_5$    & $h_6$   & $h_7$   & $h_8$       & $h_9$       & $h_{10}$       \\ \midrule
\multirow{5}{*}{\begin{tabular}[c]{@{}c@{}}Left-agents\\ $(A_L)$\end{tabular}}  & $a_1$                 & 1            & 0           & 0           & 0        & 0        & 0       & 0       & 0           & 0           & 0            \\
                                                                                & $a_2$                 & 1            & 1           & 0           & 0        & 0        & 0       & 0       & 0           & 0           & 0            \\
                                                                                & $a_3$                 & 1            & 1           & 0           & 0        & 0        & 0       & 0       & 0           & 0           & 0            \\
                                                                                & $a_4$                 & 1            & 1           & 1           & 0        & 0        & 0       & 0       & 0           & 0           & 0            \\
                                                                                & $a_5$                 & 1            & 1           & 1           & 0        & 0        & 0       & 0       & 0           & 0           & 0            \\ \midrule
\multirow{4}{*}{\begin{tabular}[c]{@{}c@{}}Right-agents\\ $(A_R)$\end{tabular}} & $a_6$                 & 0            & 0           & 0           & 0        & 0        & 0       & 0       & 1           & 1           & 1            \\
                                                                                & $a_7$                 & 0            & 0           & 0           & 0        & 0        & 0       & 0       & 1           & 1           & 1            \\
                                                                                & $a_8$                 & 0            & 0           & 0           & 0        & 0        & 0       & 0       & 0           & 1           & 1            \\
                                                                                & $a_9$                 & 0            & 0           & 0           & 0        & 0        & 0       & 0       & 0           & 0           & 1           
\end{tabular}
\caption{An example of an extremal instance $\mathcal{I} = (A, H, \mathcal{P}; k)$, where $D \subseteq H$ denotes the dummy houses. Once Reduction Rules~\ref{rr1}, \ref{rr2} and \ref{rr5} are no longer applicable, then there are no dummy agents, but dummy houses must necessarily exist, i.e., $D \neq \emptyset$; and we must have $\card{H_L} < \card{A_L}$ and $\card{H_R} < \card{A_R}$.}
\label{table:extremal}
\end{table}

Consider an extremal instance $\mathcal{I} = (A, H, \mathcal{P}; k)$ of $[\nicefrac{0}{1}]$-HA. That is, there is an ordering $\sigma$ of the houses, say, $\sigma = (h_1,\ldots,h_m)$ such that for every agent $a \in A$, either $\mathcal{P}(a) = \phi$ or there exists an index $i(a)$ such that $1 \leq i(a) \leq m$ and either $\mathcal{P}(a) = \set{h_1, h_2,\ldots, h_{i(a)}}$ or $\mathcal{P}(a) = \set{h_m, h_{m-1},\ldots, h_{{i(a)}}}$. (See~\Cref{table:extremal} for an example.) If $\mathcal{P}(a) = \phi ~\text{or}~ \set{h_1, h_2,\ldots, h_{i(a)}}$ for every $a \in A$, then we say that the instance $\mathcal{I}$ is left-extremal. If $\mathcal{P}(a) = \phi ~\text{or}~ \set{h_m, h_{m-1},\ldots, h_{i(a)}}$ for every $a \in A$, then we say that $\mathcal{I}$ is right-extremal. 

We can check in polynomial time whether a given instance is (left/right)-extremal, and if so, then find the ordering $\sigma$ on the houses. Also, note that removing a subset of houses and agents from an extremal instance does not destroy the extremal property. That is, if $\mathcal{I} = (A, H, \mathcal{P}, k)$ is extremal, then so is $\mathcal{I'} = (A', H', \mathcal{P}')$, where $A' = A \setminus X$ and $H' = H \setminus Y$ and $\mathcal{P}'$ is the restriction of $\mathcal{P}$ to $(A \cup H) \setminus (X \cup Y)$. So, in particular, we can safely apply Reduction Rules~\ref{rr2} and \ref{rr5} to extremal instances. 

Now, consider an extremal instance $\mathcal{I} = (A, H, \mathcal{P}; k)$, which is irreducible with respect to Reduction Rules~\ref{rr1},~\ref{rr2} and \ref{rr5}. Let $\sigma = (h_1, h_2,\ldots, h_m)$ be an extremal ordering on the houses. Let $D'$ be the set of dummy agents. Due to \Cref{rr5}, we have $|D'| = \phi$ in the $\mathcal{I}$.  Let $D
$ be the set of dummy houses in $\mathcal{I}$. Then, we have seen that $\card{H \setminus D} \leq n-1$. Therefore, $D \neq \emptyset$. Let $h_d, h_{d'} \in D$ be such that $h_d$ is the first dummy house and $h_{d'}$ is the last dummy house in the ordering $\sigma$. It may be the case that $d = d'$. Then, for every $i$, where $d \leq i \leq d'$, the house $h_{i}$ is a dummy house. Let $H_L = \set{h_1, h_2, \ldots, h_{d - 1}}$ and $H_R =\set{h_{d'+1}, h_{d'+2}, \ldots, h_{m}}$. We call the houses in $H_L$ the left-houses and the houses in $H_R$ the right-houses. Assume that the reduced instance contains both left and right houses, that is, $|H_L| > 0$ and $H_R > 0$. For every agent $a \in A$, either $\mathcal{P}(a) \subseteq H_L$, in which case we call the agent $a$ a left-agent, or $\mathcal{P}(a) \subseteq H_R$, in which case we call the agent $a$ a right-agent. Let $A_L$ and $A_R$ respectively denote the set of left and right agents. See~\Cref{table:extremal}. Thus, we have a partition of $A$ into $A_L$ and $A_R$ and a partition of $H$ into $H_L, H_R$ and $D$. Notice now that if $\card{H_L} \geq \card{A_L}$ or $\card{H_R} \geq \card{A_R}$, then~\Cref{rr2} would apply. We thus have $\card{A_L} > \card{H_L}$ and $\card{A_R} > \card{H_R}$. For an allocation $\fn{\Phi}{A}{H}$ and an agent $a \in A$, we say that $a$ is extremality-respecting under $\Phi$ if either (a) $a \in A_L$ and $\Phi(a) \in H_L \cup D$ or (b) $a \in A_R$ and $\Phi(a) \in H_R \cup D$. We say that $\Phi$ is extremality-respecting if every agent in $A_L \cup A_R$ is extremality-respecting under $\Phi$. We now claim that there exists an extremality-respecting optimal allocation, irrespective of whether we are dealing with an instance of $[\nicefrac{0}{1}]$-OHA, $[\nicefrac{0}{1}]$-EHA or $[\nicefrac{0}{1}]$-UHA. 

\begin{claim}
\label{claim:extremal-respect}
There exists an extremality-respecting optimal allocation for any instance irreducible with respect to the Reduction Rules \ref{rr1}, \ref{rr2} and \ref{rr5}. 
\end{claim}

\begin{proof}
Let $\fn{\Phi}{A}{H}$ be an optimal allocation that maximizes the number of extremality-respecting agents.  If $\Phi$ is extremality-respecting, then the claim trivially holds. So, assume not. Then, assume without loss of generality that there exists $a \in A_L$ with $\Phi(a) \in H_R$. (The case when $a \in A_R$ with $\Phi(a) \in H_L$ is symmetric.)

% We will construct an allocation $\Phi'$ where the number of extremality-respecting agents is strictly greater than that in $\Phi$, which will contradict the assumption. 
% If $|H_L| \geq |A_L|$ ($|H_R| \geq |A_R|$), then we apply the reduction rule \ref{rr2} to $|H_L| \geq |A_L|$ ($|H_R| \geq |A_R|$). This gives us $|H_L| < |A_L|$ and $|H_R| < |A_R|$.

Since $\card{H_R} < \card{A_R}$, there exists an agent $a' \in A_R$ such that $\Phi(a') \notin H_R$. Then, $\Phi(a') \in H_L \cup D$. Let $\Phi'$ be the allocation obtained by swapping the houses of $a$ and $a'$. This house-swapping between $a$ and $a'$ does not cause any increase in the number of envious agents or the envy experienced by any agent. So we have $\kappa^{\#}(\Phi') \leq \kappa^{\#}(\Phi)$, $\kappa^{\dagger}(\Phi') \leq \kappa^{\dagger}(\Phi)$ and $\kappa^{\star}(\Phi') \leq \kappa^{\star}(\Phi)$. Therefore, $\Phi'$ is optimal. In addition, note that $\Phi'(a) = \Phi(a') \in H_L \cup D$, and $\Phi'(a') = \Phi(a) \in H_R$, and hence, both $a$ and $a'$ are extremality-respecting under $\Phi$. That is, the number of extremality-respecting agents under $\Phi'$ is strictly greater than that in $\Phi$, a contradiction. 
\end{proof}

\begin{remark}
\label{remark:extremal}
\Cref{claim:extremal-respect} shows that whenever dealing with an instance $\mathcal{I} = (A, H, \mathcal{P}, k)$ of $[\nicefrac{0}{1}]$-HA, where $\mathcal{I}$ is extremal, we only need to look for an extremality-respecting optimal allocation, say $\Phi$. Note that there are no dummy agents in the reduced instance (\Cref{rr5}). Now suppose $n_L$ dummy houses get allocated to left-agents under $\Phi$, and $n_R$ dummy houses get allocated to right-agents under $\Phi$. Hence, we can guess the pair $(n_L, n_R)$ and split $\mathcal{I}$ into two instances, $\mathcal{I}_L$ and $\mathcal{I_R}$, where $\mathcal{I}_L$ consists of the left-agents, the left-houses and $n_L$ dummy houses, and $\mathcal{I}_R$ consists of the right-agents, the right-houses and $n_R$ dummy houses. Thus, $\mathcal{I}_L$ is left-extremal and $\mathcal{I}_R$ is right-extremal. We only need to solve the problem separately on $\mathcal{I}_L$ and $\mathcal{I}_R$. Notice that the number of guesses for the pair $(n_L, n_R)$ is at most $n^2$. By reversing the ordering on the houses in the instance $\mathcal{I}_R$, we can turn $\mathcal{I}_R$ into a left-extremal instance as well. So, it suffices to solve the problem for left-extremal instances. Hence, whenever dealing with an extremal instance $\mathcal{I}$, we assume without loss of generality that $\mathcal{I}$ is left-extremal. 
\end{remark}

\begin{remark}
\label{remark:left-extremal}
Consider a left-extremal instance $\mathcal{I} = (A, h, \mathcal{P}, k)$. Let $\sigma = (h_1, h_2,\ldots, h_m)$ be a left-extremal ordering on the houses. So for every agent $a \in A$, there exists $i(a) \in [n]$ such that $\mathcal{P}(a) = \set{h_1, h_2,\ldots,h_{i(a)}}$. Notice that $\sigma$ imposes an ordering on the agents, say $\sigma_A = (a_1, a_2,\ldots,a_n)$ so that $\mathcal{P}(a_1) \subseteq \mathcal{P}(a_2) \subseteq \cdots \subseteq \mathcal{P}(a_n)$. Notice also that we can check in polynomial time if a given instance is left-extremal, and if so, find a left-extremal ordering $\sigma$ on the houses and then the ordering $\sigma_A$ on the agents. So, whenever dealing with a left-extremal instance $\mathcal{I}$, we assume without loss of generality that $\sigma$ and $\sigma_A$ are given. Whenever we talk about, for example, the ``first/last house,'' we always mean the first/last house with respect to the ordering $\sigma$. Same with the ordering $\sigma_A$ and the ``first/last agent.''
\end{remark}

\section{Optimal House Allocation}
\label{sec:oha}
In this section, we deal with the \OHA\ problems, where the goal is to minimize the number of envious agents. We start by discussing the cases for which we have polynomial time algorithms for \OHA. 
\subsection{Polynomial Time Algorithms for OHA}
We first prove that $[\nicefrac{0}{1}]$-\OHA{} is polynomial-time solvable on instances with an extremal structure. 
\begin{theorem}
\label{lem:oha-ext} There is a polynomial-time algorithm for~$[\nicefrac{0}{1}]$-\OHA{} when the agent valuations have an extremal interval structure.
\end{theorem} 

\begin{proof} In light of~\Cref{remark:extremal}, it suffices to prove for the case when agent valuations are left extremal. Let $\mathcal{I}:= (A,H,\mathcal{P}; k)$ denote an instance of HA with left extremal valuations, which is irreducible with respect to~\Cref{rr1}, \ref{rr2} and \ref{rr5}. 

We first make the following claim:

\begin{claim}
Given an instance of~$[\nicefrac{0}{1}]$- \OHA{} when the agent valuations are left extremal, there exists an optimal allocation where the set of allocated houses form an interval.
\end{claim}

\begin{proof} Suppose we start with an optimal allocation $\Phi$ under which the set of allocated houses does not form an interval. Let $h_u$ be an unallocated house such that $h_l$ and $h_r$ are allocated, where $l<u<r$. We will show that we can allocate $h_u$ instead of $h_r$ without any increase in the envy, and hence iteratively convert $\Phi$ to another optimal allocation $\Phi'$ under which the set of allocated houses do form an interval. 

Suppose no agent is envious of the allocation of $h_r$. Then, we can allocate $h_u$ to the recipient of $h_r$ without increasing the envy. Indeed, everyone who values $h_r$ also values $h_u$ because of the interval structure, so the envy of the recipient does not increase.
Since no one envies the allocation of $h_r$,
no one should become newly envious of the allocation of $h_u$ as well, since it lies to the left of $h_r$. If any agent is envious in this re-allocation because of $h_u$, it must be already envious due to the allocation of $h_l$ which lies to the left of $h_u$. So, whoever gets $h_r$ can get $h_u$ without any increase in the envy.

On the other hand, suppose we have an agent $a$ who is envious due to the allocation of $h_r$. Since $a$ values $h_r$, due to the interval structure, she must value $h_u$ as $u<r$. If we swap $\Phi(a)$ with $h_u$, then $a$ ceases to be envious. Also, note that allocation of $h_u$ does not create any new envious agent---indeed if an agent $a^\prime$
was not envious before $h_u$ was allocated, it means either she got a house she likes or her interval ended before $h_l$. In either case, the allocation of $h_u$ can not be a cause of envy to her. This contradicts the fact that we started with an optimal allocation. 
Also, notice that all the dummy houses (if any), in the reduced instance, lie to the extreme right of all the houses that are valued. The allocation of a dummy house can always be done respecting the interval property, as agents do not distinguish between any two dummy houses.
\end{proof}

Based on the above claim, our algorithm $\textsc{Alg}$ works as follows. It enumerates over all possibilities of the first allocated house in the optimal allocation where the set of allocated houses forms an interval. There are at most $m-n$ such choices, in particular, the first $m-n$ houses. For each such $h_i$, $\textsc{Alg}$ chooses the next consecutive $n$ houses to be allocated. This reduces the instance to the one where $m=n$ and by \Cref{prop:mequalsn-oha}, this can be done in polynomial time.
 
To see the correctness, in one direction, if under any iteration $i$, the allocation $\Phi$ constructed by $\textsc{Alg}$ has at most $k$ envious agents, then it returns~\textsc{Yes}, and the allocation is the witness that $\mathcal{I}$ is a~\textsc{Yes}~instance. 
In the other direction, if $\mathcal{I}$ is a~\textsc{Yes}~instance, then there exists an optimal allocation~\textsc{Opt}~ such that $\kappa^\#(\textsc{Opt}) \leq k$ and it allocates the consecutive houses, say $[h_i, h_{i+n}]$. The algorithm captures this optimal allocation when it iterates over the house $h_i$, and hence returns~\textsc{Yes}. 
\end{proof}

We now present our algorithms for the cases when (a)every agent approves exactly one house and (b)every house is approved by almost two agents. To that end, we first present a reduction rule that will be helpful in the following two results.   
\begin{reductionrule}
\label{rr3}
For every house $h$ such that $d(h) = 1$, we allocate $h$ to $N(h$). 
\end{reductionrule}
The above reduction is safe.\footnote{This rule is subsumed by \Cref{rr2} when $|H \setminus D| \geq |A|$.} Indeed, if $d(h)=1$, then no one except exactly one agent $a$ values the house $h$, and so allocating $h$ to $a$ does not generate any envy.

\begin{theorem}
\label{lem:oha-onehouse} There is a polynomial-time algorithm for~$[\nicefrac{0}{1}]$-\OHA{} when every agent approves exactly one house.
\end{theorem} 

\begin{proof} 

Let $\mathcal{I}:= (A,H,\mathcal{P}; k)$ denote an instance of HA such that every agent approves exactly one house. Consider the associated preference graph $G = (A \cup H, E)$. We first apply the reduction rules  \ref{rr1}, \ref{rr2} and \ref{rr3}. In the reduced instance, we have $|A| \leq |H| \leq 2(|A|-1)$ and the degree of any house $h$ in $G$ is strictly greater than $1$, that is, $d(h)>1$.

We order the houses in the reduced instance in $H \setminus D$ as $\{h_1, h_2, \ldots h_t\}$ such that $d(h_1) \geq d(h_2) \geq \ldots d(h_{t})$, where $t$ denotes the number of houses in $H \setminus D$. (Note that $t \geq 1$, since $|D| \leq n-1$.) The last $|A|-|D|$ of these houses are then allocated to their neighbors (chosen arbitrarily). The remaining $|D|$ agents get the $|D|$ dummy houses. 

We argue the correctness of the above algorithm $\textsc{Alg}$. We show that the number of envious agents under the allocation $\Phi$ returned by $\textsc{Alg}$ is equal to the number of envious agents under some optimal allocation.
We first claim that in any optimal allocation~\textsc{Opt}, all the dummy houses must be allocated. Suppose $h^\star \in D$ is unallocated. Consider a house $h_i$ allocated to some agent $a$ who approves it. This causes $d(h_i)-1$ agents to have envy, and since every agent likes exactly one house, there is no way these $d(h_i)-1$ agents can become envy-free once $h_i$ is allocated. Since $h^\star$ is available, we can allocate it to $a$ and add $h_i$ to the set of unallocated houses. This decreases the number of envious agents by $d(h_i)-1$, without adding to the envy of anyone else. This contradicts the fact that we started with an optimal allocation. Therefore all the dummy houses must be allocated.  \\
Now, consider the set $S$ of houses in $H \setminus D$ that are allocated by~\textsc{Opt}. Let $T$ be the set of such houses allocated under $\Phi$. $S$ and $T$ both contain exactly $|A|-|D|$ many houses from $H \setminus D$. Consider any two house $h_i$ and $h_j$ in $H \setminus D$ such that $d(h_i)>d(h_j)$. Note that the allocation of $h_i$ creates $d(h_i)-1$ many envious agents, strictly greater than the number of envious agents created by the allocation of $h_j$, which is $d(h_j)-1$. Therefore, $S$ contains $|A|-|D|$ houses of the least degree from among the set of $H \setminus D$ houses. Since $T$ also contain the $|A|-|D|$ houses of the least degree, therefore $\kappa^\#(\Phi) = \kappa^\#(\textsc{Opt})$. 
\end{proof}

We now state the algorithm when the house degree is bounded, while the agent degree is not.

\begin{theorem}
\label{lem:oha-twoagents} There is a polynomial-time algorithm for~$[\nicefrac{0}{1}]$-\OHA{} when every house is approved by at most two agents.
\end{theorem} 

\begin{proof}
Consider the agent-house bipartite preference graph $G$. We first apply the reduction rules \ref{rr1}, \ref{rr2}, \ref{rr5} and \ref{rr3}. Then, we apply the following reduction rule.

\begin{reductionrule}
\label{rr4} 
For every cycle $C = (h_1, a_1, h_2, a_2, \ldots a_i, h_1)$ in $G$, allocate $h_i$ to the agent $a_i$. 
\end{reductionrule}
The above reduction is safe. Indeed, since $d(h) \leq 2 ~\forall~ h \in G$, a house that participates in a cycle $C$ is valued \emph{only} by the agents participating in the same cycle. Since $G$ is a bipartite graph, all cycles are of even length, so the number of agents in $C$($=|A_C|$), is equal to the number of houses in $C$ ($=|H_C|)$. This implies that every agent in $C$ can get a house she values from $C$. This does not make any agent outside $C$ envious. 

Consider the remaining graph $G$ after the application of reduction rules. Note that $G$ is a collection of trees, that is, $G = T_1, T_2, \ldots T_r$. Let $D$ be the set of dummy houses in $G$. We now describe the algorithm $\textsc{Alg}$.
First, sort the trees in increasing order of sizes, that is, $|T_1| \leq  |T_2|, \ldots \leq |T_r|$. For each tree $T_i$, we root $T_i$ at some leaf agent $a_j$ such that $d(a_j) = 1$. (After \Cref{rr3}, such a leaf agent in $T_i$ always exists, since every leaf house is allocated to the parent agent by the above reduction rule.) Let $n_1, n_2, \ldots n_r$ be the number of agents in the trees $T_1, T_2, \ldots T_r$ respectively. Let $j$ be the first index such that $(n_1+n_2 \ldots + n_j) + (r-j)> |D|$. Then, $(n_1+n_2 \ldots +n_{j-1})+(r-(j-1)) \leq |D|$. We allocate all the agents in $T_1, T_2, \ldots T_{j-1}$ a dummy house.
Then, the number of dummy houses that remain is $|D| - (n_1+n_2+\ldots+n_{j-1}) \geq  (r-(j-1))$. 
For the remaining trees, $T_j, \ldots T_r$, we match the non-root agents to their parent house and allocate the root agent a dummy house. There are $r-j+1$ root agents and there are at least so many dummy houses.

Note that only the root agents in the tree $T_j, T_{j+1}, \ldots T_r$ are the envious ones — indeed, such a root agent gets a dummy house but a house valued by her is allocated. Notice that every non-root agent in the above trees got a house she valued, and all the agents in the trees $T_1, T_2, \ldots T_{j-1}$ got a dummy house, and none of the houses they valued got allocated.

Therefore, under $\textsc{Alg}$, the number of envious agents = $r-j+1$ =  total amount of envy.

\begin{claim}$\textsc{Alg}$ returns an allocation that minimizes the number of envious agents.
\end{claim}

\begin{proof} Let $G$ be the reduced graph after the reduction rules. Note that $|H_{T_i}| = |A_{T_i}| - 1$ for any $T_i \in G$. Indeed, root $T_i$ at a house say $h_1$. Let $N(h_1) = a_1~\text{and}~a_2$. Since there is no leaf house, and $d(h) = 2 ~\forall~ h$, every $h \neq h_1$ is a parent to a unique agent $a \neq a_1, a_2$. This gives a bijection from $H_i \setminus h_1$ to $A_i \setminus \{a_1, a_2\}$. 

Let $\textsc{Opt}$ be the allocation that minimizes the number of envious agents in the reduced graph $G$. Let $l$ and $l^\prime$ be the number of envious agents under $\textsc{Opt}$ and the allocation returned by $\textsc{Alg}$ respectively.

If $l=l^\prime$, we are done. We will now show that $l \nless l^\prime$. Suppose, $l < l^\prime$. Notice that every tree $T_i$ has either no envious agents (in the case when every agent in $T_i$ gets a dummy house) or exactly one envious agent (in the case when houses from $T_i$ are allocated). Indeed, every tree is rooted at a house vertex which is of degree at most two. Also, if any of the leaf vertex in the tree $T_i$ is a house $h$, then by the structure of the tree, there is only one agent $a$ (namely, the parent of $h$ in the tree $T_i$) that values the house $h$. Therefore, $h$ can be safely assigned to $a$, without causing envy to anyone else. Therefore, we can assume that every leaf vertex in the tree is an agent. This implies that in any tree, there are more agents than houses. That is why, in case, any of the house from $T_i$ is allocated, then at least one agent will be envious. The only case where no agent from $T_i$ is envious is when none of the houses from $T_i$ are allocated and every agent in $T_i$ gets a dummy house. 

Now consider that there are at least two envious agents in a tree $T_i$ under OPT. Then, houses from $T_i$ must have been allocated. But if so, then consider the following re-allocation where every agent vertex receives its parent house vertex and one agent from the two that are incident to the root house $h$, say $a$, receives a dummy house. Note that $a$ is now the only envious agent in $T_i$. The agents, say in $T_j$, ( $\neq T_i$) who previously received houses from $T_i$ are re-allocated houses from $T_j$. If $T_j$ has more houses than agents, then everyone is allocated a house they value from $T_j$. Otherwise, agents receive their respective parent houses, and either one of the two agents incident to the root vertex ends up with a dummy house. 

The agents not in $T_i$ who might have previously received houses from $T_i$ under OPT are now allocated their respective parent houses, and either one of the two agents incident to the root vertex ends up with a dummy house. Note that such agents who are not in $T_i$ do not value the houses in $T_i$ and therefore do not differentiate between them and dummy houses. This implies that there is an allocation where at most one agent is envious in every tree. Therefore, if there are at least two envious agents in a tree $T_i$ in the OPT, then one of them can be made envy-free without an increase in the envy of any other agent. This would contradict the fact that the allocation was OPT to begin with.

So, under $\textsc{Opt}$, there are exactly $l$ trees that have exactly one envious agent. Since these $l$ agents did not get what they value, they must have got a dummy house because of the reduction rules \ref{rr3}, \ref{rr4}, and the fact that every tree has one agent more than the number of houses. Also, since there are no envious agents in the remaining $r-l$ trees (say, $T_{i_1}, T_{i_2}, \ldots T_{i_{r-l}})$, all the agents in these trees must have got dummy houses each. Therefore the number of dummy houses allocated under $\textsc{Opt}$ are: 
\begin{equation}
(n_{i_1} + n_{i_2}+ \ldots n_{i_{r-l}}) + l \geq (n_1+n_2+ \ldots + n_{r-l}) + l                       \end{equation}                     

Note that $l^\prime = r-j+1$ where $j$ is the first index such that $n_1+n_2 \ldots + n_j+ r-j> |D|$. Now since $l < l^\prime$, therefore, 
\begin{equation}
(r-l)>(r-l^\prime) \Rightarrow (r-l)\geq(r-l^\prime+1)=j \Rightarrow (r-l)\geq j
\end{equation}

This implies that there are at least $j$ envy-free trees under $\textsc{Opt}$. Suppose WLOG there are exactly $j$ envy-free trees under $\textsc{Opt}$. Then the number of dummy houses allocated under $\textsc{Opt}$: 

$(n_{i_1}+n_{i_2}+ \ldots n_{i_j})+ (r-j)\geq(n_1+n_2+ \ldots +n_j) + (r-j)> |D|$ which is a contradiction.
\end{proof}

This concludes the argument.
\end{proof}

% Notice that finding a maximum matching in $G$ and allocating the dummy houses to the remaining agents may not optimise the number of envy-free agents. Consider an instance three agents $\{a, b, c\}$ and four houses $\{1,2,3,4\}$, with $a$ and $b$ approving house $3$ and $c$ approving house $4$. Then maximum matching ends up with the allocation $\{(a, 3), (b, 1), (c, 4)\}$, with agent $b$ being the envious one, but there is an optimal allocation that has no envious agents, under which the contentious house $3$ remains unallocated. 

\subsection{Hardness Results for OHA}
In this section, we design three different reductions that will establish the hardness of \OHA\ (both $\nicefrac{0}{1}$-OHA and $[\succeq]$-OHA), including in restricted settings. The first is a parameterized reduction from \textsc{Clique} to $\nicefrac{0}{1}$-OHA. 

\begin{theorem}
\label{lem:oha-clique} $[\nicefrac{0}{1}]$-\OHA{} is \NPC{} when every house is approved by at most three agents.
\end{theorem}

\begin{proof} We sketch a reduction from \textsc{Clique}.
Let $\mathcal{I} := (G = (V,E);k)$ be an instance of \textsc{Clique}. Let $|V|=n'$ and $|E|=m'$. We first describe the construction of an instance of OHA based on $G$:

  \begin{itemize}
      \item We introduce a house $h_e$ for every edge $e \in E$. We call these the \emph{edge houses}.
      
      \item We also introduce $m' + n' - \binom{k}{2}$ dummy houses.
      
      \item We introduce an agent $a_v$ for every vertex $v \in V$ and an agent $a_e$ for every edge $e \in E$. We refer to these as the vertex and edge agents, respectively.
      
      \item Every edge agent $a_e$ values the house $h_e$.
      
      \item For every vertex $v \in V$, the vertex agent $a_v$ values the edge house $h_e$ if and only if $e$ is incident to $v$ in $G$.
  \end{itemize}

Note that in this instance of OHA, there are $m'+n'$ agents, $m'$ edge houses and $m'+n'-\binom{k}{2}$ dummy houses, and every house is approved by at most three agents. We let $k$ be the target number of envious agents. This completes the construction of the reduced instance. We now turn to a proof of equivalence. 

\paragraph*{The forward direction.} Let $S \subseteq V$ be a clique of size $k$. Then consider the allocation $\Phi$ that assigns $h_e$ to $a_e$ for all $e \in E(G[S])$ and dummy houses to all other agents. Note that no edge agent is envious in this allocation, and the only vertex agents who are envious are those that correspond to vertices of $S$. Since $|S| = k$, this establishes the claim in the forward direction. 

\paragraph*{The reverse direction.} Let $\Phi$ be an allocation that has at most $k$ envious agents. We say that $\Phi$ is nice if every edge house is either unallocated by $\Phi$ or allocated to an edge agent who values it. If $\Phi$ is not nice to begin with, notice that it can be converted to a nice allocation by a sequence of exchanges that does not increase the number of envious agents. In particular, suppose $h_e$ is allocated to an agent $a \neq a_e$. Then we obtain a new allocation by swapping the houses $h_e$ and $\Phi(a_e)$ between agents $a$ and $a_e$. This causes at least one envious agent to become envy-free (i.e., $a_e$) and at most one envy-free agent to become envious (i.e., $a$), and therefore the number of envious agents does not increase. Based on this, we assume without loss of generality, that $\Phi$ is a nice allocation.

Now note that any nice allocation $\Phi$ is compelled to assign $n$ dummy houses among the $n'$ vertex agents, and this leaves us with $m'-\binom{k}{2}$  dummy houses that can be allocated among $m'$ edge agents. Therefore, at least $\binom{k}{2}$ edge agents are assigned edge houses. 
% We may assume that exactly ${k \choose 2}$ edge agents are assigned edge houses by~\Cref{prop:alldummyhousesallocated}.

Let $F \subseteq E$ be the subset of edges corresponding to edge agents who were assigned edge houses by $\Phi$. Let $S \subseteq V$ be the set of vertices in the span of $F$, that is:

$$S := \bigcup_{e = (u,v) \in F} \{u,v\}.$$

Note that for all $v \in S$, $a_v$ is envious with respect to $\Phi$, since---by the definitions we have so far---$a_v$ valued an assigned house and was assigned a dummy house. Since $\Phi$ admits at most $k$ envious agents, we have that $|S| \leq k$. However, $S$ is also the span of at least $\binom{k}{2}$ distinct edges, so it is also true\footnote{Intuitively, a smaller set of vertices would not be able to accommodate as many edges; and specifically a subset of at most $k-1$ vertices can account for at most $\binom{k-1}{2} < \binom{k}{2}$ edges.} that $|S| \geq k$. Therefore, we conclude that $|S| = k$, and since every edge in $F$ belongs to $G[S]$ and $F$ has $\binom{k}{2}$ edges, it follows that $S$ is a clique of size $k$ in $G$. This concludes the argument in the reverse direction.
\end{proof}

\begin{theorem}
\label{lem:oha-is} $[\nicefrac{0}{1}]$-\OHA{} is \NPC{} even when every agent approves at most two houses.
\end{theorem} 

\begin{proof}
Let $\mathcal{I} := (G = (V,E);k)$ be an instance of \textsc{Clique} where $G$ is a $d$-regular graph. Let $|V|=n'$ and $|E|=m'$. (The problem of finding a clique restricted to regular graphs is also \NPC{} [\cite{MATHIESON2012179}].) We first describe the construction of an instance of OHA based on $G$:

\begin{itemize}
    \item We introduce a house $h_v$ for every vertex $v \in V$. We call these the \emph{vertex houses}.
    
    \item We also introduce $m' + n' - k$ dummy houses.
    
    \item We introduce an agent $a_v$ for every vertex $v \in V$ and an agent $a_e$ for every edge $e \in E$. We refer to these as the vertex and edge agents, respectively.
    
    \item Every vertex agent $a_v$ values the house $h_v$.
    
    \item For every edge $e = (u,v)$ in $E$, the edge agent $a_e$ values the houses $h_u$ and $h_v$.
\end{itemize}

Note that in this instance of OHA, there are $n'$ vertex houses, $m'+n'-k$ dummy houses, and every agent approves at most two houses. We let $kd - \binom{k}{2}$ be the target number of envious agents. This completes the construction of the reduced instance. We now turn to a proof of equivalence.

\paragraph*{The forward direction.} Let $S \subseteq V$ be a clique of size $k$. Then consider the allocation $\Phi$ that assigns $h_v$ to $a_v$ for all $v \in S$ and dummy houses to all other agents. Note that no vertex agent is envious in this allocation, and the only edge agents that are envious are those that correspond to edges in $G$ that have at least one of their endpoints in $S$. The total number of distinct edges incident on $S$ is at most $kd$, but since $G[S]$ induces a clique, the exact number of distinct edges incident on $S$ is $kd - \binom{k}{2}$, and this establishes the claim in the forward direction. 

\paragraph*{The reverse direction.} Let $\Phi$ be an allocation that has at most $kd -\binom{k}{2}$ envious agents. We say that $\Phi$ is nice if every vertex house is either unallocated by $\Phi$ or allocated to a vertex agent who values it. If $\Phi$ is not nice to begin with, notice that it can be converted to a nice allocation by a sequence of exchanges that does not increase the number of envious agents. In particular, suppose $h_v$ is allocated to an agent $a \neq a_v$. Then we obtain a new allocation by swapping the houses $h_v$ and $\Phi(a_v)$ between agents $a$ and $a_v$. This causes at least one envious agent to become envy-free (i.e., $a_v$) and at most one envy-free agent to become envious (i.e., $a$), and therefore the number of envious agents does not increase. Based on this, we assume without loss of generality, that $\Phi$ is a nice allocation.

Now note that any nice allocation $\Phi$ is compelled to assign $m'$ dummy houses among the $m'$ edge agents, and this leaves us with $n'-k$  dummy houses that can be allocated among $n'$ vertex agents. Therefore, at least $k$ vertex agents are assigned vertex houses. We may assume that exactly $k$ vertex agents are assigned vertex houses---indeed if more than $k$ vertex agents are assigned vertex houses, these houses can be swapped with dummy houses without increasing the number of envious agents, and we perform these swaps until we run out of dummy houses to swap with. Finally, observe that the set of $k$ vertex agents (say, $S$) who are assigned vertex houses induce a clique of size $k$ in $G$.  Indeed, if not: 

$$\text{\# of edges incident on } S = kd - {\color{IndianRed}|E(G[S])|} > kd - {\color{IndianRed}\binom{k}{2}}.$$

The claim follows from the fact that every edge incident on $S$ in $G$ corresponds to a distinct edge agent who experiences envy in the reduced instance, and this would contradict our assumption that the number of agents envious with respect to $\Phi$ was at most $kd -\binom{k}{2}$.
\end{proof}

We now exhibit the hardness of approximation of finding the maximum number of envy-free agents. Computing a maximum balanced biclique is known to be hard to approximate within a factor of $n^{(1-\gamma)}$ for any constant $\gamma>0$, where $n$ is the number of vertices \citep{a11010010}, assuming the Small Set Expansion Hypothesis \citep{10.1145/1806689.1806792}. The reduction below shows that any $f(n)$ approximation to the maximum number of envy-free agents (where $n$ is the number of agents) implies a $2(1+\epsilon) \cdot f(|L|)$ approximation to the maximum balanced biclique where $|L|$ is the size of the left bi-partition. The argument is similar to that of \cite{KMS2021}, but we produce here for the sake of completeness.

\begin{theorem}
\label{lem:oha-bbc} If the Small Set Expansion Hypothesis holds, then finding the maximum number of envy-free agents for instances with weak rankings is hard to
approximate within a factor of $n^{1-\gamma}$
for any constant $\gamma > 0$.
\end{theorem} 

\begin{proof} 

We will first show that there is a polynomial-time reduction that takes an instance $G = (L, R, E)$ of maximum balanced biclique and produces an allocation instance such that if there is a bi-clique of size at least $k$ in $G$, then there exists an allocation $\Phi$ in the reduced allocation instance such that $\Phi$ admits at least $k$ envy-free agents. On the other hand, given any allocation $\Phi$ with at least $k$ envy-free agents in the reduced instance, then there is a bi-clique of size $\nicefrac{k}{2}$ in $G$, which can be found in polynomial time. Additionally, the number of agents $n$ in the reduced instance is exactly $|L|$.

Consider an instance $G = (L, R, E)$ of maximum balanced biclique such that $L = \{b_1, \ldots, b_{n'}\}$ and $R = \{c_1, \ldots, c_{m'}\}$, we create an allocation instance as follows: an agent $a_i$ for each $b_i \in L$ and a house $h_j$ for each $c_j \in R$. We also have $n'$ additional starred houses $\{h_1^\star, \ldots, h_{n'}^\star\}$. This amounts to a total of $n'$ agents and $m'+n'$ houses. 
An agent $a_i$ ranks the houses as follows:
    
$$
       \succeq_{a_i} =
        \begin{cases}
        h_j \succ h_l, ~\text{if}~ j>l; (b_i, c_j) \notin E ~\text{and}~ (b_i, c_l) \notin E\\
        h_j \succ h_l, ~\text{if}~ (b_i, c_j) \notin E ~\text{and}~ (b_i, c_l) \in E\\
        h_j = h_l, ~\text{if}~ (b_i, c_j) \in E ~\text{and}~ (b_i, c_l) \in E\\
        h_j^\star \succ h_l^\star, ~\text{if}~ j>l\\
        h_i \succ h_j^\star, ~\forall~ i \in [m'], j \in [n']
        \end{cases}
$$
Essentially, every agent ranks his non-neighbor houses first in a fixed strict order, then ranks all his neighbors equally, and at last, ranks all the additional starred houses, again, in some fixed strict order.

We now establish the correctness of the above reduction. In the forward direction,  suppose there is a $k$-sized balanced biclique in $G$, consisting of vertices $\{b_{i_1}, b_{i_2}, \ldots b_{i_k}\}$ from $L$ and $\{c_{i_1^\prime}, \ldots c_{i_k^\prime}\}$ from $R$. Then consider the following assignment: 
$$
       \Phi(a_i) =
        \begin{cases}
        h_{i_l^\prime} ~\text{if}~ i=i_l ~\text{for some}~ l \in [k] \\
        h_{i^\star} ~\text{otherwise}
        \end{cases}
$$

 That is, the agents corresponding to the clique vertices get one of their adjacent houses, again corresponding to the clique vertices. Note that each of the $\{a_{i_1}, \ldots {a_{i_k}}\}$ rank the houses $\{h_{i_1^\prime}, \ldots h_{i_k^\prime}\}$ equally, so they do not envy each other. Also, they do not envy the remaining agents who get $h_j^\star$, as they all have the ranking $h_{i_l} \succ h_j^\star$ for all $i$ and $j$.

For the reverse direction, suppose there exists an assignment $\Phi$ of houses such that there are $k$ envy-free agents. We claim that there is a balanced biclique of size $\frac{k}{2}$ in $G$. Let $A_{EF}$ denote the set of envy-free agents with respect to $\Phi$. Notice that none of the agents in $A_{EF}$ owns a starred house under $\Phi$.
If not, suppose some $a \in A_{EF}$ gets $h_j^\star$. Consider another agent $a^\prime$ in $A_{EF}$ such that $a^\prime \neq a$. In case $a^\prime$ gets a house $h_j$ for some $j\in [m']$, then $a$ would be envious as she ranks all $h_j$ better than the starred houses. Else, if $a^\prime$ gets a starred house, say $h_l^\star$, then depending on whether $j > l$ or not, one of these two agents experiences envy, as they both rank all the starred houses in the same manner. Hence, $\Phi(A_{EF}) \subseteq \{h_1, \ldots, h_{m'}\}$. 

Now, let the $k$ houses under $\Phi(A_{EF})$ be $\{h_{j_1}, \ldots h_{j_k}\}$ such that $j_1<j_2< \ldots <j_k$. Let $a_{i_l}= \Phi^{-1}(h_{j_l})$. Then consider the set $S$, consisting of the vertices in $L$ corresponding to the first half agents in $\Phi^{-1}(A_{EF})$ and the set $T$, consisting of the vertices in $R$, corresponding to the remaining half house vertices in $\Phi(A_{EF})$. Precisely, $S = \{b_{i_1}, \ldots b_{i_\frac{k}{2}}\}$ and $T= \{c_{j_{\frac{k}{2}+1}}, \ldots, c_{j_k} \}$. We claim that $S$ and $T$ together induce a biclique of size $\frac{k}{2}$ in $G$. 
Suppose $(b_{i_l}, c_{j_{l'}}) \notin E$ for some $b_{i_l} \in S$ and some $c_{j_{l'}} \in T$. By the choice of $S$ and $T$, notice that $l < l'$. This means that $a_{i_l}$ ranks the house $h_{j_{l'}}$ strictly better that $h_{j_l}$, therefore envies $a_{i_{l'}}$, who is assigned the house $h_{j_{l'}}$. In order for $a_{i_l}$ to not envy the owner of $h_{j_{l'}}$, it must be the case that $(b_{i_l}, c_{j_{l'}}) \in E$. Therefore, $S$ and $T$ form a biclique of size $\frac{k}{2}$. This completes the correctness of the reduction.

With this reduction in hand, we will now argue that a polynomial time $f(n)$-approximation to finding the maximum number of envy-free agents gives a $2(1+\epsilon) f(|L|)$-approximation algorithm for maximum balanced biclique. But assuming the Small Set Expansion Hypothesis, this contradicts the inapproximability result for maximum balanced biclique by \cite{a11010010}. 

Consider an instance $G$ of maximum balanced biclique. We first construct the reduced allocation instance $\mathcal{I}$. Suppose there is a polynomial time $f(n)$-approximation to finding the maximum number of envy-free agents that outputs an allocation $\Phi$ for $\mathcal{I}$. We first find the biclique ($S, T$) corresponding to $\Phi$ in $G$. Let $\beta = 2(1+\frac{1}{\epsilon})$. We enumerate all subsets of size $2\beta$ in $G$ and consider the largest biclique ($S', T'$) (of size at most $2\beta$). We then output the largest of the two bicliques ($S, T$) and $(S', T')$. Let $Opt$ be the size of optimal biclique in $G$. If $Opt \leq \beta\cdot f(|L|)$, then we have that the brute force biclique $(S', T')$ has size at least $\frac{Opt}{f(|L|)}$ and we are done. Otherwise, suppose  $Opt > \beta\cdot f(|L|)$. By the reduction and $f(n)$-approximation, we have that the number of envy-free agents under $\Phi$ is at least $\frac{Opt}{f(n)} = \frac{Opt}{f(|L|)}$. Then we have the size of the biclique ($S, T$) as $$|S| = |T| = \lfloor \frac{Opt}{2f(|L|)} \rfloor > \frac{Opt}{2f(|L|)} -1 > \frac{Opt}{2f(|L|)} - \frac{Opt}{\beta f(|L|)} = \frac{(\beta-1)Opt}{\beta(2 f(|L|))} = \frac{Opt}{2f(|L|)(1+\epsilon)}$$ Therefore,  we get an approximation ratio of $2f(|L|)(1+\epsilon)$. This settles the claim.
\end{proof}

Note that $[\succeq]$-OHA remains NP-Complete from as a corollary to \Cref{lem:oha-clique}, which establishes the hardness for binary valuations, a specific case of rankings with ties. We mention here that the complexity of $[\succ]$-OHA remains open.

\subsection{Parameterized Results for OHA}
We now turn to the parameterized complexity of OHA and first present a linear kernel parameterized by the number of agents.

\begin{theorem}
\label{lem:oha-kernel-n} $[\nicefrac{0}{1}]$-\OHA{} admits a linear kernel parameterized by the number of agents. In particular, given an instance of $[\nicefrac{0}{1}]$-\OHA{}, there is a polynomial time algorithm that returns an equivalent instance of $[\nicefrac{0}{1}]$-\OHA{} with at most twice as many houses as agents. 
\end{theorem} 

\begin{proof}
It suffices to prove the safety of~\Cref{rr2}. Let $\mathcal{I} := (A,H,\mathcal{P}; k)$ denote an instance of HA with parameter $k$. Further, let  $\mathcal{I}^\prime = (H^\prime := H \setminus X, A^\prime := A \setminus Y, \mathcal{P}^\prime; k)$ denote the reduced instance corresponding to $\mathcal{I}$. Note that the parameter for the reduced instance is $k$ as well. 

If $\mathcal{I}$ is a~\textsc{Yes}-instance of OHA, then there is an allocation $\Phi: A \rightarrow H$ with at most $k$ envious agents. By \Cref{claim:niceallocalongexpansion}, we may assume that $\Phi$ is a good allocation. This implies that the projection of $\Phi$ on $H^\prime \cup A^\prime$ is well-defined, and it is easily checked that this gives an allocation with at most $k$ envious agents in $\mathcal{I}^\prime$. 

On the other hand, if $\mathcal{I}^\prime$ is a~\textsc{Yes}-instance of OHA, then there is an allocation $\Phi^\prime: A^\prime \rightarrow H^\prime$ with at most $k$ envious agents. We may extend this allocation to $\Phi: A \rightarrow H$ by allocating the houses in $Y$ to agents in $X$ along the expansion $M$, that is:
 \begin{equation*}
    \Phi(a) =
    \begin{cases}
      \Phi^\prime(a) & \text{if } a \notin X,\\
      M(a) & \text{if } a \in X.
    \end{cases}
\end{equation*}

Since all the newly allocated houses are not valued by any of the agents outside $X$ and all agents in $X$ are envy-free with respect to $\Phi$, it is easily checked that $\Phi$ also has at most $k$ envious agents. 
\end{proof}

The following results follow from using the algorithm described in~\Cref{prop:mequalsn-oha} after guessing the allocated houses, which adds a multiplicative overhead of $\binom{m}{n} \leq 2^m$ to the running time.

\begin{corollary}
\label{lem:oha-fpt-m} $[\nicefrac{0}{1}]$-\OHA{} is fixed-parameter tractable when parameterized either by the number of houses or the number of agents. In particular, $[\nicefrac{0}{1}]$-\OHA{} can be solved in time $O^\star(2^m)$.
\end{corollary} 

\begin{corollary}
\label{lem:oha-fpt-m-rankings} $[\succ]$-\OHA{} is fixed-parameter tractable when parameterized by the number of houses and can be solved in time $O^\star(2^m)$.
\end{corollary} 

% Guess unallocated houses. m \choose (m-n)
% Maximum matching?

% Unsure about agents though.
% For binary allocations houses <= f(agents)
% Since #house types was 2^n and we only needed n copies of each type
% So FPT in m => FPT in n
% But doesn't seem to be the case here?

% \begin{proof}
% In any optimal allocation, exactly $n$ houses are allocated. We first guess the $(m-n)$ unallocated houses, and remove them. For the remaining $n$ houses and $n$ agents, let $G$ be the preference graph with an edge between an agent vertex $a_i$ and a house vertex $h_j$ if $h_j$ is the top preferred house of $a_i$ among the houses that remain. Then, since each of the remaining $n$ houses needs to be allocated, we find the maximum matching $M$ in $G$, and allocate the matched houses to the corresponding agents, who are now envy-free. The unmatched agents will be envious anyway, so the remaining houses can be allocated among them arbitrarily. In at least one of the guesses of $n$ houses that were kept, $\kappa^\#(\Phi)$ is minimized and hence, that corresponds to the optimal solution. The number of total guesses are bounded by ${m \choose m-n} \leq 2^m$.
% \end{proof}
The next two results follow from \Cref{lem:oha-clique}.

\begin{corollary}
\label{cor:ohawhard} $[\nicefrac{0}{1}]$-\OHA{} is \textsf{W[1]}-hard when parameterized by the solution size, i.e, the number of envious agents, even when every house is approved by at most three agents. 
\end{corollary} 

\begin{corollary}
\label{cor:ohawhard-rankings} $[\succeq]$-\OHA{} is \textsf{W[1]}-hard when parameterized by the solution size, i.e., the number of envious agents. 
\end{corollary} 

We now show that $[\nicefrac{0}{1}]$-OHA is fixed-parameter tractable when parameterized by the number of types of houses or the number of types of agents. To that end, we formulate $[\nicefrac{0}{1}]$-OHA as an integer linear program and then invoke \Cref{thm:lenstra}.

\begin{theorem}
\label{lem:oha-fpt-mstar} $[\nicefrac{0}{1}]$-\OHA{} is fixed-parameter tractable when parameterized by either the number of houses types or the number of agents types.
\end{theorem} 
% I think there is a ILP solution here, need to discuss. 

% \begin{proof}
% \end{proof}
 
%\begin{proposition}[\cite{DBLP:journals/mor/Lenstra83}]\label{prop:lenstra}There exist a function $f$ and an algorithm that solves an ILP instance $\mathtt{P}$ with $\ell$ variables in time $f(\ell) L^{\cO(1)}$, where $L$ is the number of bits required to encode the ILP $\mathtt{P}$. \end{proposition}

Consider an instance $\mathcal{I} = (A, h, \mathcal{P}, k)$ of $[\nicefrac{0}{1}]$-OHA. Recall that we use $n^*$ to denote the number of types of agents in $\mathcal{I}$ and the $m^*$ to denote the number of types of houses in $\mathcal{I}$. With a slight abuse of notation, for $i \in [n^*]$, we use $\mathcal{P}(i) (\subseteq [m^*])$ to denote the set of types of houses that each agent of type $i$ values. Also, for $i \in [n^*], j \in [m^*]$ and an allocation $\fn{\Phi}{A}{H}$, let $A(\Phi, i, j) \subseteq A$ be the set of agents of type $i$ who receive a house of type $j$ under $\Phi$. 

\begin{observation}\label{obs:typebounds}
Consider an instance $\mathcal{I} = (A, h, \mathcal{P}, k)$ of $[\nicefrac{0}{1}]$-OHA. Then, (1) $n^* \leq 2^{m^*}$ and (2) $m^* \leq 2^{n^*}$. To see (1), for each $i \in [n^*]$, the agents of type $i$ are uniquely identified by the types of houses they prefer, and the number of distinct choices for the types of houses is at most $2^{m^*}$. Similarly, to see (2), observe that for each $j \in [m^*]$, the houses of type $j$ are uniquely identified by the types of agents who prefer houses of type $j$; and the number of distinct choices for the types of agents is at most $2^{n^*}$. 
\end{observation}

\begin{lemma}\label{lem:allornone}
Consider an instance $\mathcal{I} = (A, h, \mathcal{P}, k)$ of $[\nicefrac{0}{1}]$-OHA and an allocation $\fn{\Phi}{A}{H}$. Consider any fixed pair $(i, j)$, where $i \in [n^*]$ and $j \in [m^*]$. Then, for every $a \in A$ and for every $a', a'' \in A(\Phi, i, j)$, either both $a'$ and $a''$ envy $a$, or neither $a'$ nor $a''$ envies $a$. 
\end{lemma}
\begin{proof}
Consider $a \in A$ and $a', a'' \in A(\Phi, i, j)$. First, if $j \in \mathcal{P}(i)$, then neither $a'$ nor $a''$ envies any agent. So, assume that $j \notin \mathcal{P}(i)$. Let $\Phi(a)$ be of type $\ell$, for some $\ell \in [m^*]$. If $\ell \notin \mathcal{P}(i)$, then, neither $a'$ nor $a''$ envies $a$. If $\ell \in \mathcal{P}(i)$, then both $a'$ and $a''$ envy $a$. 
\end{proof}

We now move to formulate the $[\nicefrac{0}{1}]$-OHA problem as an integer linear program (ILP). The number of variables in our ILP will be $\cO(n^* \cdot m^*)$. By \Cref{obs:typebounds}, the number of variables will be bounded separately by both $\cO(n^* \cdot 2^{n^*})$ and $\cO(m^{*} \cdot 2^{m^*})$. That is, the number of variables will be bounded separately by both the number of house types and the number of agent types. The result will then follow from~\Cref{thm:lenstra}. 

Consider an instance $\mathcal{I} = (A, h, \mathcal{P}, k)$ of $[\nicefrac{0}{1}]$-\OHA. For each $i \in [n^*]$ and $j \in [m^*]$, let $n_i$ be the number of agents of type $i$ and $m_j$ the houses of type $j$. To define our ILP, we introduce the following variables. For each $i \in [n^*]$, $j \in [m^*]$, we introduce four variables: $x_{ij}, z_{ij}, d_{ij}$ and $d'_{ij}$. Here, $x_{ij}$ and $z_{ij}$ are integer variables, and $d_{ij}$ and $d'_{ij}$ are binary variables. The semantics of the variables are as follows. (1) We want $x_{ij}$ to be the number of agents of type $i$ who receive houses of type $j$. Equivalently, we want $x_{ij}$ to be the number of houses of type $j$ that are allocated to agents of type $i$. (2) And we want $z_{ij}$ to be the number of \emph{envious} agents of type $i$ who receive houses of type $j$. By~\Cref{lem:allornone}, either all type $i$ agents who receive type $j$ houses are envious, or none of them is envious. That is, we must have either $z_{ij} = x_{ij}$ or $z_{ij} = 0$. Notice that type $i$ agents who receive type $j$ houses are envious if and only if $j \notin \mathcal{P}(i)$, and for some $j' \in \mathcal{P}(i)$, at least one house of type $j'$ has been allocated (to, say, an agent of type $i'$ for some $i' \in [n^*]$). That is, we must have $z_{ij} > 0$ if and only if $j \notin \mathcal{P}(i)$, $x_{ij} > 0$ and $x_{i',j'} > 0$ for some $i' \in [n^*]$ and $j' \in \mathcal{P}(i)$. Hence, for $j \notin \mathcal{P}(i)$, either $x_{ij} = 0$ or $z_{ij} > 0$ if $\sum_{i' \in [n]} \sum_{j' \in \mathcal{P}(i)}x_{i'j'} > 0$. (3) The variables $d_{ij}, d'_{ij}$ are only dummy variables that we use to enforce the ``either or'' constraints. %(3) The variable $d_{ij}$ is only a dummy varaible that we use to enforce the constraint $z_{ij} = 0$ or $z_{ij} = x_{ij}$. 

We now formally describe our ILP. Minimize $\sum_{i \in [n^*]} \sum_{j \in [m^*]} z_{ij}$ subject to the constraints in~\Cref{table:ilpoha}.  

% \usepackage{multirow}
\begin{table}
\centering
\begin{tabular}{lll}
(C1.$i$).           & $\sum_{j \in [m^*]} x_{ij} = n_i$                   & $ \forall i \in [n^*]$                                                         \\
                    &                                                     &                                                                                 \\
(C2.$j$).           & $\sum_{i \in [n^*]} x_{ij} \leq m_j$                & $\forall j \in [m^*]$                                                         \\
                    &                                                     &                                                                                 \\
(C3.a.$i.j$).       & $x_{ij} \leq nd'_{ij}$                           & \multirow{3}{*}{$\forall i \in [n^*]$, $j \in [m^*]\setminus \mathcal{P}(i)$}                         \\
(C3.b.$i.j$).       & $\sum_{i' \in [n^*]} \sum_{j' \in \mathcal{P}(i)} x_{i'j'} \leq nmz_{ij} + nm(1-d'_{ij})$                       &                                                                                 \\
(C3.c.$i.j$).       & $z_{ij} \leq n_i \sum_{i' \in [n^*]} \sum_{j' \in \mathcal{P}(i)} x_{i'j'}$ & \\
                    &                                                     &                                                                                 \\

(C4.a.$i.j$).       & $z_{ij} \leq n_id_{ij}$                             & \multirow{3}{*}{$\forall i \in [n^*]$, $j \in [m^*]$}                         \\
(C4.b.$i.j$).       & $x_{ij} - z_{ij} \leq n_i(1-d_{ij})$                &                                                                                 \\
(C4.c.$i.j$).       & $z_{ij} \leq x_{ij}$                                &                                                                                 \\
                    &                                                     &                                                                                 \\
(C5.$i.j$).         & $z_{ij} = 0$                                        & $\forall i \in [n^*], j \in \mathcal{P}(i)$                                         \\
                    &                                                     &                                                                                 \\
(C6.a.$i.j$).       & $x_{ij} \geq 0$                                     & \multirow{4}{*}{$\forall i \in [n^*]$, $j \in [m^*]$}                         \\
(C6.b.$i.j$).       & $z_{ij} \geq 0$                                     &                                                                                 \\
(C6.c.$i.j$).       & $d_{ij} \in \{0, 1\}$                               &                                                                                 \\
(C6.d.$i.j$).       & $d'_{ij} \in \{0, 1\}$                               &   
\end{tabular}
\caption{The constraints of the ILP \ilpoha.}
\label{table:ilpoha}
\end{table}

For convenience, we name this ILP \ilpoha, and denote the set of variables of $\ilpoha$ by $Var(\ilpoha)$ and the optimum value of \ilpoha\ by $opt(\ilpoha)$.  %By a solution for $\ilpoha$, we mean a function $\fn{f}{Var{\ilpoha}}{\mathbb{Z}}$.  
Constraint C1.$i$ ensures that for each $i \in [n^*]$, the number of houses allocated to agents of type $i$ is exactly $n_i$. In other words, all agents of type $i$ receive houses. Constraint C2.$j$ ensures that for each $j \in [m^*]$, the number of houses of type $j$ that are allocated does not exceed $m_j$. 
For $i \in [n^*]$ and $j \in [m] \setminus \mathcal{P}(i)$, constraints C3.a.$i.j$ and C3.b.$i.j$ together ensure that depending on whether $d'_{ij} = 0$ or $d'_{ij} =1$, we have either $x_{ij} = 0$ or $z_{ij} > 0$ if $\sum_{i' \in [n^*]} \sum_{j' \in [\mathcal{P}(i)]} x_{i'j'} > 0$. Constraint C3.c.$i.j$ ensures that if $x_{i'j'} = 0$ for every $i' \in [n^*], j' \in [m] \setminus \mathcal{P}(i)$, then $z_{ij} = 0$. Constraints C4.a.$i.j$-C4.c.$i.j$ together ensure that depending on $d_{ij} = 0$ or $d_{ij} =1$, we have either $z_{ij} = 0$ or $z_{ij} = x_{ij}$. 

To establish the correctness of $\ilpoha$, we prove the following two claims. 

\begin{claim}\label{claim:ilp-oha-1}
For any allocation $\fn{\Phi}{A}{H}$, there exists a feasible solution $\fn{f_{\Phi}}{Var(\ilpoha)}{\mathbb{Z}}$ for \ilpoha\ such that $\kappa^{\#}(\Phi) = \sum_{i \in [n^*]} \sum_{j \in [m^*]} f_{\Phi}(z_{ij})$. 
\end{claim}

\begin{claim}\label{claim:ilp-oha-2}
For every optimal solution $\fn{f}{Var(\ilpoha)}{\mathbb{Z}}$ for \ilpoha, there exists an allocation $\fn{\Phi_f}{A}{H}$ such that $\kappa^{\#}(\Phi_f) = \sum_{i \in [n^*]}\sum_{j \in [m^*]}f(z_{ij})$.  
\end{claim}

\begin{remark}
Notice that~\Cref{claim:ilp-oha-1}, in fact, proves that ILP \ilpoha\ is always feasible as there is always an allocation. Moreover, for an allocation $\Phi$, since $0 \leq \kappa^{\#}(\Phi) \leq n$, and since $\kappa^{\#}(\Phi) = \sum_{i \in [n^*]}\sum_{j \in [m^*]} f_{\Phi}(z_{ij})$, we can conclude that \ilpoha\ has a bounded solution. Therefore, $\opt(\ilpoha)$ is well-defined.  %Moreover, in Claim~\ref{claim:ilp-oha-1}, if $\Phi$ is an optimal allocation, then since $\kappa^{\#}(\Phi) = \sum_{i \in [n^*]} \sum_{j \in [m^*]} f_{\Phi}(z_{ij})$, we can conclude that  
\end{remark}
Assuming Claims~\ref{claim:ilp-oha-1} and \ref{claim:ilp-oha-2} hold, we now prove the following claim. 
\begin{claim}\label{claim:ilp-oha-main}
We have $\kappa^{\#}(\mathcal{I}) = \opt(\ilpoha)$. 
\end{claim}
\begin{proof}
To prove the claim, we will prove that (1) $\kappa^{\#}(\mathcal{I}) \geq  \opt(\ilpoha)$ and (2) $\opt(\ilpoha) \geq \kappa^{\#}(\mathcal{I})$. 

To prove (1), consider an optimal allocation $\fn{\Phi}{A}{H}$. That is, $\kappa^{\#}(\mathcal{I}) =  \kappa^{\#}(\Phi)$. By~\Cref{claim:ilp-oha-1}, we have $\kappa^{\#}(\Phi) = \sum_{i \in [n^*]} \sum_{j \in [m^*]} f_{\Phi}(z_{ij}) \geq \opt(\ilpoha)$, where $f_{\Phi}$ is as defined in~\Cref{claim:ilp-oha-1}. We thus have $\kappa^{\#}(\mathcal{I}) \geq \opt(\ilpoha)$.  

Now, to prove (2), consider an optimal solution $f$ for \ilpoha. That is,

$\opt(\ilpoha) = \sum_{i \in [n^*]} \sum_{j \in [m^*]} f(z_{ij})$. By~\Cref{claim:ilp-oha-2}, we have 

$\opt(\ilpoha) = \sum_{i \in [n^*]} \sum_{j \in [m^*]} f(z_{ij}) = \kappa^{\#}(\Phi_f) \geq \kappa^{\#}(\mathcal{I})$, where $\Phi_f$ is as defined in~\Cref{claim:ilp-oha-2}. We thus have $\opt(\ilpoha) \geq \kappa^{\#}(\mathcal{I})$. 
\end{proof}

We are now ready to prove~\Cref{lem:oha-fpt-mstar}. 
\begin{proof}[Proof of~\Cref{lem:oha-fpt-mstar}]
Given an instance $\mathcal{I}$ of $[\nicefrac{0}{1}]$-OHA, observe that we can construct the ILP \ilpoha in polynomial time. The number of variables in \ilpoha\ is bounded by $4n^* m^*$. The number of constraints in \ilpoha\ is also bounded by $\cO(n^* m^*)$. The maximum value of any coefficient or constant term in \ilpoha\ is bounded by $nm$. So, \ilpoha\ can be encoded using $\poly(n^*, m^*) \cdot \cO(\log(nm))$ bits. The result then follows from~\Cref{thm:lenstra} and \Cref{claim:ilp-oha-main}. 
\end{proof}
We now only have to prove Claims~\ref{claim:ilp-oha-1} and \ref{claim:ilp-oha-2}. 
\begin{proof}[Proof of \Cref{claim:ilp-oha-1}]
Consider any allocation $\fn{\Phi}{A}{H}$. Recall that $A(\Phi, i, j)$ is the set of agents of type $i$ who receive houses to type $j$ under $\Phi$. 

We define a solution $\fn{f_{\Phi}}{Var(\ilpoha)}{\mathbb{Z}}$ for \ilpoha\ as follows. For each $i \in [n^*], j \in [m^*]$, we set (1) $f_{\Phi}(x_{ij}) = \card{A(\Phi, i, j)}$;  (2) $f_{\Phi}(z_{ij}) = f_{\Phi}(x_{ij})$ if there exists an envious agent $a \in A(\Phi, i, j)$ and $f_{\Phi}(z_{ij}) = 0$ otherwise; (3) $f_{\Phi}(d_{ij}) = 0$ if  $f_{\Phi}(z_{ij}) = 0$ and $f_{\Phi}(d_{ij}) = 1$ otherwise; and (4) $f_{\Phi}(d'_{ij}) = 0$ if $f_{\Phi}(x_{ij}) = 0$ and $f_{\Phi}(d'_{ij}) =1$ otherwise. 

%We can verify that $f_{\Phi}$ satisfies all the constraints of \ilpoha, and hence is a feasible solution for \ilpoha. 
To see that $f_{\Phi}$ satisfies all the constraints of \ilpoha, observe first that $\card{A(\Phi, i, j)}$ is the number of agents of type $i$ who receive houses of type $j$ under $\Phi$; equivalently, $\card{A(\Phi, i, j)}$ is the number of houses of type $j$ that have been allocated to agents of type $i$ under $\Phi$. Therefore, $\sum_{j \in [m^*]} \card{A(\Phi, i, j)} = n_i$ and $\sum_{i \in [n^*]} \card{A(\Phi, i, j)} \leq m_j$. Hence, (1) $\sum_{j \in [m^*]} f_{\Phi}(x_{ij}) = \sum_{j \in [m^*]} \card{A(\Phi, i, j)} = n_i$ and $\sum_{i \in [n^*]} f_{\Phi}(x_{ij}) = \sum_{i \in [n^*]} \card{A(\Phi, i, j)} \leq m_j$. Thus $f_{\Phi}$ satisfies constraints C1.$i$ and C2.$j$ for every $i \in [n^*]$ and $j \in [m^*]$. Also, note that  \\ $\sum_{i' \in [n^*]} \sum_{j' \in \mathcal{P}(i)} f_{\Phi}(x_{i' j'}) \leq \sum_{i' \in [n^*]} \sum_{j' \in [m^*]} \card{A(\Phi, i, j)} \leq n$. 

Now, consider $i \in [n^*]$, $j \in [m^*]$. Suppose first that $f_{\Phi}(x_{ij}) = 0$. Then, by the definition of $f_{\Phi}$, we have $\card{A(\Phi, i, j)} = 0$, and hence $f_{\Phi}(z_{ij}) = 0$, which implies $f_{\Phi}(d_{ij}) = 0$; and $f_{\Phi}(x_{ij}) = 0$ implies that $f_{\Phi}(d'_{ij}) = 0$. Note that in this case, $f_{\Phi}$ satisfies all the constraints. In particular, constraint C3.b.$i.j$ is satisfied because $f_{\Phi}(z_{ij}) = f_{\Phi}(d'_{ij}) = 0$ implies that the right side of constraint C3.b.$i.j$ is exactly equal to $nm$, and the left side of the constraint is at most $n$. Since $f_{\Phi}(x_{ij}) = f_{\Phi}(z_{ij}) = 0$, all the other constraints corresponding to the pair $(i, j)$ are also satisfied. 

Suppose now that $f_{\Phi}(x_{ij}) > 0$. Then by the definition of $f_{\Phi}$, we have $f_{\Phi}(d'_{ij}) = 1$. There are two possibilities: (1) $f_{\Phi}(z_{ij}) = 0$ and (2) $f_{\Phi}(z_{ij}) > 0$. 

Assume first that $f_{\Phi}(z_{ij}) = 0$. Again, by the definition of $f_{\Phi}$, we have $f_{\Phi}(d_{ij}) = 0$. Notice that this choice of values satisfies all the constraints corresponding to the pair $(i, j)$, except possibly C3.b.$i.j$. To see that C3.c.$i.j$ is also satisfied, assume that $j \in [m^*] \setminus \mathcal{P}(i)$. Since  $f_{\Phi}(z_{ij}) = 0$, the definition of $f_{\Phi}$ implies that the agents of type $i$ who receive houses of type $j$ are not envious. Hence we can conclude that none of the houses of type $j'$ have been allocated under $\Phi$, for any $j' \in \mathcal{P}(i)$. That is, $A(\Phi, i', j') = \emptyset$ for every $i' \in [n^*]$ and $[j'] \in \mathcal{P}(i)$. Thus the left side of constraint C3.b.$i.j$ is $0$; and the right side is $0$ as well, as $f_{\Phi}(z_{ij}) = 1$ and $f_{\Phi}(d'_{ij}) = 1$. 

Finally, assume that $f_{\Phi}(z_{ij}) > 0$. Then, by the definition of $f_{\Phi}$, we have $f_{\Phi}(z_{ij}) = f_{\Phi}(x_{ij}) = \card{A(\Phi, i, j)} \leq n_i$ and $f_{\Phi}(d_{ij}) = 1$. Notice that this choice of values satisfies constraints C4.a.$i.j$-C4.c.$i.j$. From the definition of $f_{\Phi}$, we can also conclude that the agents of type $i$ who receive houses of type $j$ under $\Phi$ are envious, which implies that $j \notin \mathcal{P}(i)$ and $\card{A(\Phi, i', j')} = f_{\Phi}(x_{i' j'}) > 0$ for some $i' \in [n^*]$, $j' \in \mathcal{P}(i)$. Thus the right side of constraint C3.c.$i.j$ is strictly positive; and since $f_{\Phi}(z_{ij}) \leq n_i$, constraint C3.c.$i.j$ is satisfied. Since $f_{\Phi}(d'_{ij}) = 1$ and $f_{\Phi}(x_{ij}) \leq n$, constraint C3.a.$i.j$ is satisfied. Finally, constraint C3.b.$i.j$ is satisfied because its left side is at most $n$, and the right side is at least $nm$ as $f_{\Phi}(z_{ij}) > 0$. Notice also that in this case \ilpoha\ does not contain constraint C5.$i.j$ as $j \notin \mathcal{P}(i)$. 

We have thus shown that $f_{\Phi}$ satisfies all the constraints of \ilpoha. 

Consider $i \in [n^*], j \in [m^*]$. Suppose that $A(\Phi, i, j) \neq \emptyset$. By \Cref{lem:allornone}, either all agents in $A(\Phi, i, j)$ are envious or none of them are. By the definition of $f_{\Phi}$, we have $f_{\Phi}(z_{ij}) = f_{\Phi}(x_{ij}) = \card{A(\Phi, i, j)}$ if and only if the agents in $A(\Phi, i, j)$ are envious; and $f_{\Phi}(z_{ij}) = 0$ otherwise. Therefore, the number of envious agents under $\Phi$, $\kappa^{\#}(\Phi) = \sum_{i \in [n^*]}\sum_{j \in [m^*]} f_{\Phi}(z_{ij})$. 
\end{proof}

To prove~\Cref{claim:ilp-oha-2}, we first prove two preparatory claims below. In both these claims, $\fn{f}{Var(\ilpoha)}{\mathbb{Z}}$ is a feasible solution for \ilpoha. 

\begin{claim}\label{claim:zij}
For every $i \in [n^*], j \in [m^*]$, either $f(z_{ij}) = 0$ or $f(z_{ij}) = f(x_{ij})$. 
\end{claim}
\begin{proof}
Fix $i \in [n^*], j \in [m^*]$. Since constraint C6.c.$i.j$ is satisfied, we have $f(d_{ij}) \in \{0,1\}$. If $f(d_{ij}) = 0$, then constraints C4.a.$i.j$ implies that $f(z_{ij}) \leq 0$ and then constraint C6.b.$i.j$ implies that $z_{ij} =0$. Instead, if $f(d_{ij}) = 1$, then, constraint C4.b.$i.j$ implies that $f(x_{ij}) - f(z_{ij}) \leq 0$, which implies that $f(x_{ij}) \leq f(z_{ij})$. But then constraint C4.c.$i.j$ implies that $f(z_{ij}) = f(x_{ij})$. 
\end{proof}

\begin{claim}\label{claim:envy}
For every $i \in [n^*], j \in [m^*]$, $f(z_{ij}) > 0$ if and only if $j \notin \mathcal{P}(i)$, $f(x_{ij})> 0$ and $f(x_{i'j'}) > 0$ for some $i' \in [n^*]$ and $j' \in \mathcal{P}(i)$. 
\end{claim}
\begin{proof}
Fix $i \in [n^*], j \in [m^*]$. Assume first that $f(z_{ij}) > 0$. Then, constraint C5.$i.j$ implies that $j \notin \mathcal{P}(i)$. Constraint C4.c.$i.j$ implies that $f(x_{ij}) \geq f(z_{ij}) > 0$. Now, if $f(x_{i'j'}) = 0$ for every $i' \in [n], j' \in \mathcal{P}(i)$, then constraint C3.c.i.j would imply that $f(z_{ij}) \leq 0$, which is not possible. Hence, we have $f(x_{i'j'}) > 0$ for some $i' \in [n], j' \in \mathcal{P}(i)$. 

Assume now that $j \notin \mathcal{P}(i)$, $f(x_{ij}) > 0$, and $f(x_{i'j'}) > 0$ for some $i' \in [n^*]$ and $j' \in \mathcal{P}(i)$.  Since $f(x_{ij}) > 0$, constraint C3.a.$i.j$ implies that $f(d'_{ij}) > 0$. By constraint C6.d.$i.j$, we then have $f(d'_{ij}) = 1$. Then, constraint C3.b.$i.j$ ,along with the fact that $f(x_{i'j'}) > 0$ for some $i' \in [n^*]$ and $j' \in \mathcal{P}(i)$, implies that $0 < \sum_{i' \in [n^*]} \sum_{j' \in \mathcal{P}(i)} f(x_{i'j'}) \leq nmf(z_{ij})$, which, then implies that $f(z_{ij}) > 0$.  
\end{proof}

\begin{proof}[Proof of \Cref{claim:ilp-oha-2}]
Given $\fn{f}{Var(\ilpoha)}{\mathbb{Z}}$, we define $\fn{\Phi_f}{A}{H}$ as follows. For each $i \in [n^*], j \in [m^*]$, we allocate $f(x_{ij})$ houses of type $j$ to agents of type $i$ (one house per agent). Thus, we have $\card{A(\Phi_f, i, j)} = f(x_{ij})$. Notice that as $f$ satisfies constraints C1.$i$ and C2.$j$, the allocation $\Phi_f$ is valid. 

To complete the proof of the claim, we only need to prove that for every $i \in [n^*], j \in [m^*]$, the number of envious agents of type $i$ who received houses of type $j$ under $\Phi_f$ is exactly equal to $f(z_{ij})$. Fix $i \in [n^*], j \in [m^*]$. 
%If $\card{A(\Phi_f, i, j)} = f(x_{ij}) = 0$, then constraint C6.c.$i.j$ implies that $f(z_{ij}) \leq 0$; and then constraint C8.b.$i.j$ implies that $f(z_{ij}) = 0$. So, suppose that $\card{A(\Phi_f, i, j)} = f(x_{ij}) > 0$. 
By \Cref{lem:allornone}, either all agents in $A(\Phi_f, i, j)$ are envious or none of them is envious. Notice that the agents in $A(\Phi_f, i, j)$ (if they exist) are envious if and only if $\card{A(\Phi_f, i, j)} > 0$, $j \notin \mathcal{P}(i)$ and $A(\Phi_f, i', j') \neq \emptyset$ for some $i' \in [n^*], j' \in [m^*]\setminus \mathcal{P}(i)$. That is, the agents in $A(\Phi_f, i, j)$ are envious if and only if $f(x_{ij}) = \card{A(\Phi_f, i, j)} > 0$,  $j \notin \mathcal{P}(i)$ and $f(x_{i'j'}) = \card{A(\Phi_f, i', j')} > 0$ for some $i' \in [n^*], j' \in [m^*]\setminus \mathcal{P}(i)$. On the other hand, by \Cref{claim:envy}, $f(z_{ij}) > 0$ if and only if $f(x_{ij}) > 0$, $j \notin \mathcal{P}(i)$ and $f(x_{i'j'}) > 0$ for some $i' \in [n^*], j' \in [m^*] \setminus \mathcal{P}(i)$. We can thus conclude that the agents in $A(\Phi_f, i, j)$ are envious if and only if $f(z_{ij}) > 0$. By \Cref{claim:zij}, we also have $f(z_{ij}) = 0$ or $f(z_{ij}) = f(x_{ij})$. This implies that the agents in $A(\Phi_f, i, j)$ are envious if and only if $f(z_{ij}) = f(x_{ij}) >0$. By \Cref{lem:allornone}, if the agents in $A(\Phi_f, i, j)$ are envious, then the number of envious agents in $A(\Phi_f, i, j)$ is exactly $\card{A(\Phi_f, i, j)} = f(x_{ij}) = f(z_{ij})$. We thus have $\kappa^{\#}(\Phi_f) = \sum_{i \in [n^*]} \sum_{j \in [m^*]} f(z_{ij})$. 
\end{proof}

% Can remove the forced page breaks later :)
% \newpage

\section{Egalitarian House Allocation}
\label{sec:eha}

In this section, we deal with the EHA problems, where the goal is to minimize the maximum envy experienced by any agent. We first discuss the polynomial time algorithms for EHA. 

\subsection{Polynomial Time Algorithms for EHA}

\begin{theorem}
\label{lem:eha-ext} There is a polynomial-time algorithm for~$[\nicefrac{0}{1}]$-\EHA{} when the agent valuations have an extremal interval structure.
\end{theorem} 

\begin{proof} In light of  \Cref{remark:extremal}, it suffices to prove for the case when agent valuations are left-extremal. Let $\mathcal{I}:= (A, H,\mathcal{P}; k)$ be an instance of EHA with left-extremal valuations. Consider an envious agent $a_l$ in the allocation, who approves the interval $[h_1, h_i]$. Note that if $\mathcal{I}$ is a~\textsc{Yes}~instance, no more than $k$ houses can be allocated from $[h_1, h_i]$, else the envy experienced by $a_l$ will be more than $k$. 

Based on this observation, the algorithm works as follows. We order the agents in the increasing order of the length of their intervals, that is, $a$ appears before $a^\prime$ if $\mathcal{P}(a) \subseteq \mathcal{P}(a^\prime)$. We guess the last envious agent $a_l$. There are at most $n$ such guesses. 
Suppose the interval that $a_l$ approves is $[h_1, h_j]$ for some $j \in [m]$. Since any valid allocation is bound to allocate at most $k$ houses from $[h_1, h_j]$, we iterate over the number of houses $i$ allocated from $[h_1, h_j]$. (Note that $i \leq k$.) For each such $i$, we construct the associated preference graph $G_i = (A_{\geq l+i+1} \cup H_{\geq j+1}; E)$ where $A_{\geq l+i+1} = \{a_{l+i+1}, a_{l+i+2}, \ldots a_{n}\}$ and $H_{\geq j+1} = \{h_{j+1}, \ldots h_m\}$. We then find a matching in $G_i$ that saturates $A_{\geq l+i+1}$, and if it does not exist, then the iteration is discarded. The algorithm then constructs the allocation $\Phi_i$ as follows. The matched houses under $M_i$ are allocated to the matched agents. The first $i$ houses from $[h_1, h_j]$ are allocated to the first $i$ agents whose interval ends after $h_j$, particularly, to $\{a_{l+1}, a_{l+2}, \ldots, a_{l+i}\}$. The remaining unmatched houses are assigned arbitrarily to the unmatched agents. If the number of remaining houses is less than the remaining agents, then the iteration is discarded and the algorithm moves to the next iteration. The algorithm returns~\textsc{Yes} if for some iteration $i$, $\Phi_i$ is a complete allocation, that is, every agent gets a house under $\Phi_i$. Else, it returns~\textsc{No}.

To argue the correctness of the above algorithm, we show that if $\mathcal{I}$ is a~\textsc{Yes}~instance, if and only if for some iteration $i$, there exists a complete allocation $\Phi_i$.
In the forward direction, suppose $\mathcal{I}$ is a~\textsc{Yes}~instance. There must exist some optimal allocation~\textsc{Opt}~such that the maximum envy under~\textsc{Opt}~is at most $k$. Among all the envious agents under~\textsc{Opt}, consider the agent $a$ whose interval $[h_1, h]$ is longest, that is, for all envious agents $a^\prime$, $\mathcal{P}(a^\prime) \subseteq \mathcal{P}(a)$. Suppose~\textsc{Opt}~allocates $k^\prime$ houses from $[h_1, h]$. Then consider the iteration in $A$ that iterates over the agent $a_l = a$ and $i = k^\prime$. (This implies $[h_1, h_j] = [h_1, h]$.)  Note that under~\textsc{Opt}, all the agents whose interval ends after $h$ must get a house that they like. This implies that for all the agents $\{a_{l+1}, a_{l+2}, \ldots a_{n}\}$, there exists a house that they like and that can be allocated to them. As ~\textsc{Opt}~ allocates exactly $k^\prime (=i)$ houses from $[h_1, h_j]$, at most $i$ of the agents among $\{a_{l+1}, a_{l+2}, \ldots a_{n}\}$ can get a house they like from $[h_1, h_j]$. For the remaining ones, there must exist at least one house $[h_{j+1}, h_m]$ that they like and can be allocated to them, which implies that there must exist a matching saturating the said agents under the said iteration $i$. Therefore all the agents $\{a_{l+1}, a_{l+2}, \ldots a_{n}\}$ get a house that they like under $\Phi_i$. Also, since~\textsc{OPT}~is a complete allocation, so there are enough remaining houses for the agents $\{a_1, \ldots a_l\}$ that can be allocated to them, once $\{a_{l+1}, a_{l+2}, \ldots a_{n}\}$ get what they value. This implies that there are enough houses remaining under $\Phi_i$ as well to be allocated to the remaining agents. (Note that even if all the agents $\{a_1, \ldots a_l\}$ are envious, their envy is bounded by at most $i \leq k$.) This implies that the allocation $\Phi_i$ constructed in the iteration $i$ is indeed complete, and the algorithm returns~\textsc{Yes}.

In the reverse direction, suppose there exists a complete allocation $\Phi_i$. Note that all the envious agents under $\Phi_i$ do not value any house outside $[h_1, h_j]$ and exactly $i (\leq k)$ houses are allocated from $\{h_1, h_j\}$. Therefore the maximum envy is at most $k$, and $\mathcal{I}$ is a~\textsc{Yes}~instance.
This concludes the proof.
\end{proof}

We now state the algorithm for the restricted setting when every agent approves exactly one house. 
\begin{theorem}
\label{lem:eha-onehouse} There is a polynomial-time algorithm for~$[\nicefrac{0}{1}]$-\EHA{} when every agent approves exactly one house.
\end{theorem}

\begin{proof} Notice that when every agent approves exactly one house, then in any allocation, the maximum envy $\kappa^\dagger(\Phi)$ is bounded by $1$. Given an instance $\mathcal{I}:= (A,H,\mathcal{P}; k)$ of HA, we invoke the Hall's Violators Algorithm by \cite{GSV2019}. The above algorithm checks whether there is an envy-free allocation and returns one if it exists. Therefore, if it returns an allocation, then $\kappa^\dagger(\Phi) = 0$, else it has to be at least $1$ under any allocation.
\end{proof}

\subsection{Hardness Results for EHA}

In this section, we establish the hardness of EHA for both binary and linear orders, even under restricted settings.

\begin{theorem}
\label{lem:eha-is} There is a polynomial-time reduction from~\textsc{Independent Set} to~$[\nicefrac{0}{1}]$-\EHA{}, where every agent approves at most two houses and every house is approved by at most ten agents. This reduction shows that~$[\nicefrac{0}{1}]$-\EHA{} is \NPC{} even when the target value of the maximum envy is one.
\end{theorem}

\begin{proof}
Let $\mathcal{I} := (G = (V,E);k)$ be an instance of \textsc{Independent Set}. Let $|V|=n'$ and $|E|=m'$. Since \textsc{Independent Set} is known to be \NPC{} even on subcubic graphs, we assume that $G$ is subcubic, i.e, that the degree of every vertex in $G$ is at most three. We construct an instance of~$[\nicefrac{0}{1}]$-EHA as follows.

\begin{itemize}
  \item We introduce a house $h_v$ for every vertex $v$ in $V$. We call these the \emph{vertex houses}. 
  \item We also introduce $(3m'+n')-k$ dummy houses.
  \item We introduce an agent $a_v$ for every vertex $v \in V$ and three agents $a_e^{1},a_e^{2},a_e^{3}$ for every edge $e \in E$. We refer to these as the vertex and edge agents, respectively. We refer to the three agents corresponding to a single edge as the \emph{cohort} around $e$.
  \item All three edge agents corresponding to the edge $e = (u,v)$ value the houses $h_u$ and $h_v$.
  \item A vertex agent $a_v$ values the house $h_v$. 
\end{itemize}

Note that in this instance of EHA, there are $3m'+n'$ agents, $n'$ vertex houses, and $3m'+n'-k$ dummy houses, and every agent approves at most two houses. We set the maximum allowed envy at one, that is, the reduced instance asks for an allocation where every agent envies at most one other agent. This completes the construction of the reduced instance. We now turn to a proof of equivalence. 

\paragraph*{The forward direction.}

Let $S \subseteq V$ be an independent set of size $k$. Consider the allocation $\Phi$ that assigns $h_v$ to $a_v$ for all $v \in S$ and dummy houses to all remaining agents. Note that all vertex agents are envy-free under this allocation. If an edge agent envies two other agents, it must be two agents who received vertex houses, since recipients of dummy houses are never a cause for envy. So suppose an edge agent $a_e^\circ$ who approves, say, the houses $h_u$ and $h_v$ is envious, then it implies that both $a_u$ and $a_v$ belong to $S$ while it is also true that $(u,v) \in E$; contradicting our assumption that $G[S]$ induces an independent set. This completes the argument in the forward direction.

\paragraph*{The reverse direction.}

Let $\Phi$ be an allocation with respect to which every agent envies at most one other agent. First, note that if $e = (u,v)$ is an edge in $G$, then note that $\Phi$ can allocate at most one of the houses $h_u$ and $h_v$. Suppose not. Then notice that at least one of the agents among the cohort of edge agents corresponding to $e$, i.e, $a_e^1, a_e^2$ and $a_e^3$ envy the two agents who were assigned the houses $h_u$ and $h_v$. 

We say that $\Phi$ is nice if every vertex house is either unallocated by $\Phi$ or allocated to a vertex agent who values it. If $\Phi$ is not nice to begin with, notice that it can be converted to a nice allocation by a sequence of exchanges that does not increase the maximum envy of the allocation without changing the set of allocated houses. In particular, suppose $h_v$ is allocated to an agent $a \neq a_v$. Then we obtain a new allocation by swapping the houses $h_v$ and $\Phi(a_v)$ between agents $a$ and $a_v$. This causes $a_v$ to become envy-free. If $a$ is a vertex agent, then she experiences the same envy as before. If $a$ is an edge agent, then we have two possible scenarios. Suppose $a$ did not value $h_v$: then the amount of envy experienced by $a$ is either the same or less than, the amount of envy she had in the original allocation. On the other hand, if $a$ did value $h_v$, then $a$ envies $a_v$ in the new allocation but nobody else, since $h_u$ is unallocated in $\Phi$. 

Based on this, we assume without loss of generality, that $\Phi$ is a nice allocation. Now note that any nice allocation $\Phi$ is compelled to assign $3m'$ dummy houses among the $3m'$ edge agents, and this leaves us with $n'-k$  dummy houses that can be allocated among $n'$ vertex agents. Therefore, at least $k$ vertex agents are assigned vertex houses. Finally, observe that the corresponding vertices induce an independent set in $G$ of size at least $k$. Indeed, this follows from a fact that we have already argued: any two houses corresponding to vertices that are endpoints of an edge cannot be both allocated by $\Phi$. This concludes the argument in the reverse direction. 
\end{proof}

\begin{theorem}
\label{lem:eha-mcis} There is a polynomial-time reduction from~\textsc{Multi-Colored Independent Set} to~$[\succ]$-\EHA{}. This reduction shows that~$[\succ]$-\EHA{} is \NPC{} even when the target value of the maximum envy is one.
\end{theorem}

\begin{proof}
  Let $\mathcal{I} := (G = (V,E);k)$ be an instance of \textsc{Multi-Colored Independent Set} with color classes $V_1, \ldots, V_k$. By adding dummy global vertices if required, we assume that $|V_i| = n \geq 3$ for all $i \in [k]$, and denote the vertices in $V_i$ by $\{u_1^i, \ldots, u_n^i\}$. The global dummy vertices added in any partition $V_i$ are adjacent to all the other vertices not in $V_i$.

  We now construct an instance of~$[\succ]$-EHA as follows.
  
  \begin{itemize}
    \item We introduce a house for every vertex $v$ in $V$. We call these the \emph{vertex houses} and they are denoted by: 
    \[H_V := \{h_1^{1}, \ldots, h_{n}^{1}\} \uplus \cdots \uplus \{h_1^{k}, \ldots, h_{n}^{k}\}.\]
    \item We introduce a house $h_e$ for every edge $e$ in $E$. We call these the \emph{edge houses} and denote this set of houses by $H_E$.
    \item For each $1 \leq i \leq k$, and for every $1 \leq p \neq q \leq n$, we introduce three houses $h_{[i;p,q]}^1$, $h_{[i;p,q]}^2$ and $h_{[i;p,q]}^3$. We call them \emph{special houses}. 
    \item We also introduce $k \cdot (n-1)$ additional houses, denoted by:
    \[H_D := \{d_1^{1}, \ldots, d_{n-1}^{1}\} \uplus \cdots \uplus \{d_1^{k}, \ldots, d_{n-1}^{k}\}.\]    
    \item We introduce an agent for every vertex $v \in V$, denoted by $a_{[1,i]}, \ldots, a_{[n,i]}$ for $1 \leq i \leq k$. We call them \emph{vertex agents}.
    \item We introduce an agent $a_e$ for every $e \in E$.  We call them \emph{edge agents}.
    \item For each $1 \leq i \leq k$, and for every $1 \leq p \neq q \leq n$, we introduce three agents $a_{[i;p,q,1]}$, $a_{[i;p,q,2]}$ and $a_{[i;p,q,3]}$. We call them \emph{guards}. 
  \end{itemize}

  For a set $X$, we use $\overline{X}$ to denote an arbitrary order on the set $X$. Also, for a fixed order, say $\sigma$, we use $[[\sigma]]_i$ to denote the order $\sigma$ rotated $i$ times. For example, $[[x \succ y \succ z]]_2 = z \succ x \succ y$. Note that $\sigma$ is an order over $n$ elements, then $[[\sigma]]_n = \sigma$. We also use $H$ to denote the set of houses in the reduced instance. We are now ready to describe the preferences of the agents.

  \begin{itemize}
    \item An edge agent corresponding to an edge $e = (u_p^i,u_q^j)$ that has endpoints in $V_i$ and $V_j$ (with $i < j$) ranks the houses as follows:

  $$\succ_e: h_p^i \succ h_q^j \succ h_e \succ \overline{H \setminus \{h_p^i, h_q^j, h_e\}}$$

  \item For each $1 \leq i \leq k$, for every $1 \leq p \neq q \leq n$, and $\ell \in \{1,2,3\}$ the guard agent $a_{[i;p,q,\ell]}$ has the following preference:
  
  $$\succ_{[i;p,q,\ell]}: h_p^i \succ h_q^i \succ h_{[i;p,q]}^\ell \succ \overline{H \setminus \{h_p^i,h_q^i,h_{[i;p,q]}^\ell\}}.$$

  \item For $1 \leq i \leq k$, the vertex agents $a_{[1,i]}, \ldots, a_{[n,i]}$ rank the houses as follows
  
  \begin{align*}
    {\color{SteelBlue}\succ_{[1,i]}}: ~~&  h^i_1 \succ [[d^i_1 \succ d^i_2 \succ d^i_3 \succ \cdots \succ d^i_{n-1}]]_0 & \succ  \overline{H \setminus \{h^i_1, d^i_1, d^i_2, d^i_3, \ldots, d^i_{n-1}\}}  \\ 
    {\color{SteelBlue}\succ_{[2,i]}}: ~~& h^i_2 \succ [[d^i_1 \succ d^i_2 \succ d^i_3 \succ \cdots \succ d^i_{n-1}]]_0 & \succ  \overline{H \setminus \{h^i_2, d^i_1, d^i_2, d^i_3, \ldots, d^i_{n-1}\}}  \\ 
     & ~~~~~~~~~~~~~~~~~~~~ \vdots  &  \\ 
    {\color{SteelBlue}\succ_{[j,i]}}: ~~&  h^i_j \succ [[d^i_1 \succ d^i_2 \succ d^i_3 \succ \cdots \succ d^i_{n-1}]]_{j-2} & \succ  \overline{H \setminus \{h^i_j, d^i_1, d^i_2, d^i_3, \ldots, d^i_{n-1}\}}  \\ 
     & ~~~~~~~~~~~~~~~~~~~~ \vdots  &  \\ 
    {\color{SteelBlue}\succ_{[n-1,i]}}: ~~& h^i_{n-1} \succ  [[d^i_1 \succ d^i_2 \succ d^i_3 \succ \cdots \succ d^i_{n-1}]]_{n-3} & \succ  \overline{H \setminus \{h^i_{n-1}, d^i_1, d^i_2, d^i_3, \ldots, d^i_{n-1}\}}  \\ 
    {\color{SteelBlue}\succ_{[n,i]}}: ~~&  h^i_n \succ [[d^i_1 \succ d^i_2 \succ d^i_3 \succ \cdots \succ d^i_{n-1}]]_{n-2} & \succ  \overline{H \setminus \{h^i_n, d^i_1, d^i_2, d^i_3, \ldots, d^i_{n-1}\}}  \\
  \end{align*}

  \end{itemize} 

  Note that in this instance of EHA, there are $k \cdot (n-1)$ extra houses. We set the maximum allowed envy at one, that is, the reduced instance asks for an allocation where every agent envies at most one other agent. This completes the construction of the reduced instance. We now turn to a proof of equivalence. 
  
  \paragraph*{The forward direction.} 
  
  Let $S \subseteq V$ be a multicolored independent set. Let $s: [k] \rightarrow [n]$ be such that:
  $$S = \{u^1_{s(1)}, u^2_{s(2)}, \ldots, u^k_{s(k)}\}.$$ 
  
  We now describe an allocation $\Phi$ based on $S$. First, we let $\Phi(a_e) = h_e$ for all $e \in E$. Also, for each $1 \leq i \leq k$, and for every $1 \leq p \neq q \leq n$, we have $\Phi(a_{[i;p,q,1]}) = h_{[i;p,q]}^1$, $\Phi(a_{[i;p,q,2]}) = h_{[i;p,q]}^2$ and  $\Phi(a_{[i;p,q,3]}) = h_{[i;p,q]}^3$.
  
  Now, for the vertex agents corresponding to the vertices of $V_i$, we have the following if $s(i) \geq 2$:

  \begin{equation*}
    \Phi(a_{[j,i]}) =
    \begin{cases}
      d_{j}^i & \text{if } 1 \leq j < s(i),\\
      h^i_j & \text{if } j = s(i),\\
      d_{j-1}^i & \text{if } s(i) < j \leq n,\\
    \end{cases}
  \end{equation*}

  and if $s(i) = 1$, then we proceed as follows instead:
  
  \begin{equation*}
    \Phi(a_{[j,i]}) =
    \begin{cases}
      h^i_1 & \text{if } j = 1,\\
      d_{j-1}^i & \text{if } j > 1,\\
    \end{cases}
  \end{equation*}

  Note that every house corresponding to a vertex not in $S$ remains unallocated in $\Phi$, implying, in particular, that exactly one vertex from each color class corresponds to an allocated vertex house in $\Phi$. We now argue that every agent envies at most one other agent with respect to this allocation. 
  
  First, consider an edge agent $a_e$ corresponding to an edge $e = (u_p^i,u_q^j)$ that has endpoints in $V_i$ and $V_j$ (with $i < j$). Recall that $a_e$ gets her third-ranked house with respect to $\Phi$. Since at most one of $h_p^i$ or $h_q^j$ is allocated with respect to $\Phi$, we have that $a_e$ envies at most one agent. 
  
  Similarly, every guard agent receives the special house that she ranks third, and at most one of the two top-ranked houses is allocated in $\Phi$, since both of the top-ranked houses belong to the same color class by construction. 
  
  Now we turn to the vertex agents. It is easily verified that every vertex agent gets a house that they rank first (if they correspond to a vertex from $S$), second, or third. Thus, vertex agents corresponding to vertices in $S$ are envy-free, and all other vertex agents envy at most one other agent. (If an agent $a_{[j, i]}$ receives its third-ranked house $d_j^i$, then notice that the first-ranked house $h_j^i$ remains unallocated. Therefore, $a_{[j, i]}$ envies at most one other agent who might have received $d_{j-1}^i$, its second-ranked house.) This concludes the proof in the forward direction. 

  \paragraph*{The reverse direction.}

  Let $\Phi$ be an allocation for the reduced instance where every agent envies at most one other agent. We make a series of claims about the allocation $\Phi$ that allows us to observe that $\Phi$ has the following properties: it allocates exactly one vertex house from the houses corresponding to vertices in a common color class, and further, it allocates such a house to a vertex agent. Such vertex houses are then easily seen to correspond to a multi-colored independent set in $G$: indeed, if not, then the pair of adjacent vertices would correspond to an edge agent who is envious of at least two vertex agents, contradicting our assumption about $\Phi$. 
  
  We first observe that we cannot allocate more than one house from among vertex houses corresponding to vertices from a common color class of $G$.

  \begin{claim}
  \label{claim:5}
    Let $i \in [k]$ be arbitrary but fixed. Among the vertex houses $\{h_1^{i}, \ldots, h_n^{i}\}$, $\Phi$ leaves at least $(n-1)$ houses unallocated; in other words, $\Phi$ allocates at most one house from among these houses.
  \end{claim}

  \begin{proof}
    Suppose not, and in particular, suppose $\Phi$ allocates the houses $h_p^{i}$ and $h_q^{i}$ for some $1 \leq p \neq q \leq n$. Then at least one of the three guard agents $a_{i;p,q,\ell}$ for $\ell \in \{1,2,3\}$ will envy the two agents who receive these two houses, which contradicts the assumption that every agent envies at most one other agent in the allocation $\Phi$. 
  \end{proof}
  
 The total number of houses are $nk$ vertex houses, $(n-1)k$ dummy houses, $m$ edge houses and $3k \binom{n}{2}$ special houses. The total number of agents are $nk$ vertex agents, $m$ edge agents, and $3k \binom{n}{2}$ guards. Now since at least $(n-1)k$ vertex houses remain unallocated by \Cref{claim:5}, to ensure that every agent gets a house, at least one vertex house from each of the $n$ vertex partitions must be allocated. This implies that exactly one house is allocated from one color class.

Since $(n-1)k$ vertex houses are unallocated, all the remaining houses must be allocated by $\Phi$. In particular, all additional houses are allocated, and we use this fact in our next claim. 

  \begin{claim}
    If a vertex house is allocated in $\Phi$, then it is assigned to a vertex agent. 
  \end{claim}

  \begin{proof}
    Suppose not. Let $i$ be such that the vertex house under consideration corresponds to a vertex from $V_i$. Note that the additional houses $\{d_1^{i}, \ldots, d_{n-1}^{i}\}$ can be allocated among at most $(n-1)$ of the agents corresponding to the vertices in $V_i$. Therefore, there is at least one agent $a$ among the agents $a_{[j,i]}$, $j \in [n]$ who does not receive any of the houses among her top $n$-ranked houses. Further, she was also not assigned her top-ranked house, by the assumption we made for the sake of contradiction. Since $\Phi$ allocates all the houses in $\{d_1^{i}, \ldots, d_{n-1}^{i}\}$, the agent $a$ envies at least $(n-1) \geq 2$ agents, and this is the desired contradiction. 
  \end{proof}

  The previous two claims imply the desired structure on $\Phi$, and as argued earlier, the subset of vertices corresponding to allocated vertex houses induces a multicolored independent set in $G$, and this concludes the argument in the reverse direction.    
\end{proof}

\subsection{Parameterized Results for EHA}
In this section, we present a linear kernel for $[\nicefrac{0}{1}]$-EHA and discuss the tractable and the hard cases in the parameterized setting. 

\begin{theorem}
\label{lem:eha-kernel-n} $[\nicefrac{0}{1}]$-\EHA{} admits a linear kernel parameterized by the number of agents. In particular, given an instance of $[\nicefrac{0}{1}]$-\EHA{}, there is a polynomial time algorithm that returns an equivalent instance of $[\nicefrac{0}{1}]$-\EHA{} with at most twice as many houses as agents. 
\end{theorem} 

\begin{proof}
It suffices to prove the safety of~\Cref{rr2}. Let $\mathcal{I} := (A,H,\mathcal{P};k)$ denote an instance of HA with parameter $k$. Further, let  $\mathcal{I}^\prime = (H^\prime := H \setminus X, A^\prime := A \setminus Y, \mathcal{P}^\prime; k)$ denote the reduced instance corresponding to $\mathcal{I}$. Recall that the parameter for the reduced instance is $k$ as well. 

If $\mathcal{I}$ is a~\textsc{Yes}-instance of EHA, then there is an allocation $\Phi: A \rightarrow H$ with maximum envy at most $k$. By \Cref{claim:niceallocmaxenvy}, we may assume that $\Phi$ is a good allocation. This implies that the projection of $\Phi$ on $H^\prime \cup A^\prime$ is well-defined, and it is easily checked that this gives an allocation with maximum envy at most $k$

On the other hand, if $\mathcal{I}^\prime$ is a~\textsc{Yes}-instance of EHA, then there is an allocation $\Phi^\prime: A^\prime \rightarrow H^\prime$ with maximum envy $k$. We may extend this allocation to $\Phi: A \rightarrow H$ by allocating the houses in $Y$ to agents in $X$ along the expansion $M$, that is:

 \begin{equation*}
    \Phi(a) =
    \begin{cases}
      \Phi^\prime(a) & \text{if } a \notin X,\\
      M(a) & \text{if } a \in X.
    \end{cases}
\end{equation*}

Since all the newly allocated houses are not valued by any of the agents outside $X$ and all agents in $X$ are envy-free with respect to $\Phi$, it is easily checked that $\Phi$ also has maximum envy $k$.
\end{proof}

The following results follow from using the algorithm described in~\Cref{prop:mequalsn-eha} after guessing the allocated houses, which adds a multiplicative overhead of $\binom{m}{n} \leq 2^m$ to the running time.

\begin{corollary}
\label{lem:eha-fpt-m} $[\nicefrac{0}{1}]$-\EHA{} is fixed-parameter tractable when parameterized either by the number of houses or the number of agents. In particular, $[\nicefrac{0}{1}]$-\EHA{} can be solved in time $O^\star(2^m)$.
\end{corollary} 

% guess the max envy: L
% focus only on agents whose degree is L+1 or more
% then find a perfect matching for these agents
% if it doesn't exist say NO
% otherwise match the rest arbitrarily

% \begin{proof}

% \end{proof}

\begin{corollary}
\label{lem:eha-fpt-m-rankings} \sloppypar $[\succ]$-\EHA{} is fixed-parameter tractable when parameterized by the number of houses and can be solved in time $O^\star(2^m)$. \endsloppypar
\end{corollary} 
% Guess unallocated houses. m \choose (m-n)
% Guess the threshold for max envy, say l.
% Perfect matching on top l choices among what remains.
% Again, unclear if this extends to agents for rankings.

% \begin{proof}
% Similar to \Cref{lem:oha-fpt-m-rankings}, we first guess the $(m-n)$ unallocated houses, and remove them. 
% For the remaining instance where $m = n$, we first guess the amount of maximum envy, say $l$. Notice that $l \leq n-1$, and every agent must get one of its top $l$ houses among the $n$ houses that remain. Indeed, if she does not, then, as all these $n$ houses are allocated, the envy experienced by her will be greater than $l$.

% Now, let $G$ be the preference graph with an edge between an agent vertex $a_i$ and a house vertex $h_j$ if and only if $h_j$ lies in the top $l$ houses in the preference order of $a_i$ over the $n$ houses that remain.
% We then find the perfect matching $M$ in $G$, and allocate the houses along the matched edges. In at least one of the guesses of $n$ houses that were kept, $\kappa^\dagger(\Phi)$ is minimized and hence correspond to the optimal solution. The number of total guesses are bounded by ${m \choose m-n} \leq 2^m$.

% \end{proof}
The next two results follow respectively from \Cref{lem:eha-is} and \Cref{lem:eha-mcis} respectively.

\begin{corollary}
\label{cor:ehawhard} $[\nicefrac{0}{1}]$-\EHA{} is para-\textsf{NP}-hard when parameterized by the solution size, i.e, the maximum envy, even when every agent approves at most two houses.
\end{corollary} 

\begin{corollary}
\label{cor:ehawhard-rankings} $[\succ]$-\EHA{} is para-\textsf{NP}-hard when parameterized by the solution size, i.e., the maximum envy.
\end{corollary} 

We now formulate EHA as an integer linear program, as in the case of $[\nicefrac{0}{1}]$-OHA and establish the fixed-parameter traceability parameterized by the number of house types or agent types. The number of variables in our ILP will be $\cO(n^* \cdot m^*)$, where $n^*$ is the number of types of agents and $m^*$ the number of types of houses. Again, by~\Cref{obs:typebounds}, the number of variables will then be bounded separately by $2^{\cO(m^*)}$ and $2^{\cO(n^*)}$.

\begin{theorem}\label{lem:eha-fpt-mstar}
$[\nicefrac{0}{1}]$-\EHA{} is fixed-parameter tractable when parameterized either by the number of house types or the number of agent types.
\end{theorem}

Given an instance $\mathcal{I} = (A, h, \mathcal{P}, k)$ of $[\nicefrac{0}{1}]$-\EHA, we define an ILP \ilpeha\ that encodes the instance $\mathcal{I}$. The ILP \ilpeha\ is very similar to \ilpoha\ with exactly two distinctions. (1) The ILP \ilpeha\ has all the variables of \ilpoha. In addition, \ilpeha\ has an integer variable $w$ that encodes the maximum envy experienced by an agent. (2) In \ilpeha, the variable $z_{ij}$ for $i \in [n^*], j \in [m^*]$ encodes the envy experienced by each agent of type $i$ who receives a house of type $j$. Note that the envy experienced by such an agent is always either $0$, or $\sum_{i' \in [n^*]} \sum_{j' \in \mathcal{P}(i)} x_{i'j'}$. Since $w$ is the maximum envy experienced by an agent, we must also have $z_{ij} \leq w$ for every $i \in [n^*], j \in [m^*]$. 

\begin{table}
\begin{tabular}{lll}
(C1.$i$).           & $\sum_{j \in [m^*]} x_{ij} = n_i$                   & for every $i \in [n^*]$                                                         \\
                    &                                                     &                                                                                 \\
(C2.$j$).           & $\sum_{i \in [n^*]} x_{ij} \leq m_j$                & for every $j \in [m^*]$                                                         \\
                    &                                                     &                                                                                 \\
(C3.a.$i.j$).       & $x_{ij} \leq nd'_{ij}$                           & \multirow{2}{*}{for every $i \in [n^*]$, $j \in [m^*]\setminus \mathcal{P}(i)$}                         \\
(C3.b.$i.j$).       & $\sum_{i' \in [n^*]} \sum_{j' \in [\mathcal{P}(i)]} x_{i'j'} \leq nmz_{ij} + nm(1-d'_{ij})$                       &                                                                                 \\

                    &                                                     &                                                                                 \\
(C4.a.$i.j$).       & $z_{ij} \leq nd_{ij}$                             & \multirow{3}{*}{for every $i \in [n^*]$, $j \in [m^*]$}                         \\
(C4.b.$i.j$).       & $\left(\sum_{i' \in [n^*]} \sum_{j' \in \mathcal{P}(i)}x_{i'j'}\right) - z_{ij} \leq n(1-d_{ij})$                &                                                                                 \\
(C4.c.$i.j$).       & $z_{ij} \leq n\sum_{i' \in [n^*]} \sum_{j' \in \mathcal{P}(i)} x_{i'j'}$ & \\
                    &                                                     &                                                                                 \\
(C5.$i.j$).         & $z_{ij} = 0$                                        & for every $i \in [n^*], j \in \mathcal{P}(i)$                                         \\
                    &                                                     &                                                                                 \\
(C6.a.$i.j$).       & $x_{ij} \geq 0$                                     & \multirow{4}{*}{for every $i \in [n^*]$, $j \in [m^*]$}                         \\
(C6.b.$i.j$).       & $z_{ij} \geq 0$                                     &                                                                                 \\
(C6.c.$i.j$).       & $d_{ij} \in \{0, 1\}$                          &                                                                                 \\
(C6.d.$i.j$).       & $d'_{ij} \in \{0, 1\}$                               &
                       \\
                    &                                                     &
                       \\
     
(C7).              & $w \geq 0$                                          &
                      \\
                    &                                                     &
                      \\
(C8.$i.j$).        & $z_{ij} \leq w$                                     & for every $i\in [n^*]$, $j \in [m^*]$
\end{tabular}
\caption{The constraints of the ILP \ilpeha.}
\label{table:ilpeha}
\end{table}
We now formally describe the ILP. Minimize $w$ subject to the constraints in \Cref{table:ilpeha}. 
\begin{proof}[Proof Outline of \Cref{lem:eha-fpt-mstar}]

Observe that \ilpeha differs from \ilpoha in constraints C4.a.$i.j$, C4.b.$i.j$ and C4.c.$i.j$. These three constraints together now ensure that for any feasible solution $f$ for \ilpeha, we either have $f(z_{ij}) = 0$ or $f(z_{ij}) = \sum_{i' \in [n^*]} \sum_{j' \in \mathcal{P}(i)}f(x_{i'j'})$. 
(Also, constraint C4.c.$i.j$ subsumes the constraint C3.c$i.j$ in \ilpoha.) The only other difference is the addition of constraints C7 and C8.$i.j$. We can show that appropriate counterparts of Claims \ref{claim:zij} and \ref{claim:envy} hold for \ilpeha.  So do appropriate counterparts of Claims~\ref{claim:ilp-oha-1}, \ref{claim:ilp-oha-2} and \ref{claim:ilp-oha-main}. In particular, we have $\kappa^{\dagger}(\mathcal{I}) = \opt(\ilpeha)$. \Cref{lem:eha-fpt-mstar} will then follow from~\Cref{thm:lenstra}. 
\end{proof}

\begin{remark}
By modifying \ilpeha, we can formulate an integer program for $[\nicefrac{0}{1}]$-UHA. We only need to remove the variable $w$ and the constraints C7 and C8.$i.j$ for $i \in [n^*], j \in [m^*]$ and replace the objective function with $\sum_{i \in [n^*]} \sum_{j \in [m^*]} x_{ij} z_{ij}$. Notice that while all the constraints in this integer program are linear, the objective function is quadratic. We thus have an integer quadratic program (IQP). The value of the largest coefficient in the constraints and the objective function is $nm$. It is known that IQP is fixed-parameter tractable when parameterized by the number of variables plus the value of the largest coefficient~[\cite{DBLP:journals/corr/Lokshtanov15}]. Fixed-parameter tractability results for IQP w.r.t. other parameters are also known~[\cite{DBLP:conf/aaai/EibenGKO19}].  
\end{remark}
% Can remove the forced page breaks later :)
% \newpage
\section{Utilitarian House Allocation}
\label{sec:uha}
We now deal with the UHA problems, where the goal is to minimize total envy. We first discuss the polynomial time algorithms for \UHA. 

\subsection{Polynomial Time Algorithms for UHA}

\begin{theorem}
\label{lem:uha-ext} There is a polynomial-time algorithm for~$[\nicefrac{0}{1}]$-\UHA{} when the agent valuations have an extremal interval structure.
\end{theorem}

\begin{proof}
Consider an instance $\mathcal{I} = (H, A, \mathcal{P}; k)$ of $[\nicefrac{0}{1}]$-UHA. In light of \Cref{remark:extremal}, assume that the valuations have a left-extremal structure.  Consider the ordering on the agents such that $i < j$ if $\mathcal{P}(a_i) \subseteq \mathcal{P}(a_j)$. Our algorithm relies on the existence of an optimal allocation with some desirable properties. To that end, consider an allocation $\fn{\Phi}{A}{H}$. We say that an ordered pair of agents $(a_i, a_j) \in A \times A$ is rogue under $\Phi$ if $i < j$, $\Phi(a_i) \in \mathcal{P}(a_i)$ and $\Phi(a_j) \notin \mathcal{P}(a_j)$. That is, for $i < j$, $(a_i, a_j)$ is a rogue pair if the $a_i$ values the house that she receives and $a_j$ does not value that she receives. Consider a rogue pair $(a_i, a_j).$ Recall that we are in the left-extremal setting, and hence $\mathcal{P}(a_i) \subseteq \mathcal{P}(a_j)$, which implies that $\Phi(a_i) \in \mathcal{P}(a_j)$. Thus $a_j$ is envious under $\Phi$.  

We say that $\Phi$ is rogue-free if there does not exist any rogue pair under $\Phi$.  Notice that if $\Phi$ is rogue-free, then there exists $t(\Phi) \in [n] \cup \set{0}$ such that for every $i \in [n]$ with $i > t(\Phi)$, we have $\Phi(a_i) \in \mathcal{P}(a_i)$, and hence the agent $a_i$ is envy-free. For $i \leq t(\Phi)$, the agent $a_i$ may or may not be envy-free. 

\begin{claim}
\label{claim:rogue}
There exists a rogue-free optimal allocation. 
\end{claim}
\begin{proof}
Consider an optimal allocation $\fn{\Phi}{A}{H}$ that minimizes the number of rogue pairs. By optimal, we mean that $\kappa^{\star}(\mathcal{I}) = \kappa^{\star}(\Phi)$. If $\Phi$ is rogue-free, then the claim trivially holds. So, assume that $\Phi$ is not rogue-free. Then there exists a rogue pair under $\Phi$. We fix a rogue pair $(a_i, a_j)$ as follows. Let $a_i$ be the first agent such that $(a_i, a_p)$ is a rogue pair for some $p \in [n]$. Then choose $j$ such that $a_j$ is the last agent such that $(a_i, a_j)$ is a rouge pair. 
Since $(a_i, a_j)$ is a rogue pair, we have $i < j$, $\Phi(a_i) \in \mathcal{P}(a_i)$ and $\Phi(a_j) \notin \mathcal{P}(a_j)$. Since $i < j$, we have $\mathcal{P}(a_i) \subseteq \mathcal{P}(a_j)$, which implies that $\Phi(A) \cap \mathcal{P}(a_i) \subseteq \Phi(A) \cap \mathcal{P}(a_j)$. Notice first that $a_j$ is envious as $\Phi(a_i) \in \mathcal{P}(a_i) \subseteq \mathcal{P}(a_j)$. Also, the number of agents that $a_j$ envies, $\mathcal{E}_{\Phi}(a_j) = \card{\Phi(A) \cap \mathcal{P}(a_j)}$. 

Let $\Phi'$ be the allocation obtained from $\Phi$ by swapping the houses of $a_i$ and $a_j$. That is, $\Phi'(a_i) = \Phi(a_j)$, $\Phi'(a_j) = \Phi(a_i)$ and $\Phi'(a_r) = \Phi(a_r)$ for every $r \in [n] \setminus \set{i, j}$. Then, $a_i$ is envious under $\Phi'$ as $\Phi'(a_i) = \Phi(a_j) \notin \mathcal{P}(a_j) \supseteq \mathcal{P}(a_i)$. The number of agents $a_i$ envies, $\card{\mathcal{E}_{\Phi'}(a_i)} = \card{\Phi'(A) \cap \mathcal{P}(a_i)}$. Now, $a_j$ is not envious under $\Phi'$ as $\Phi'(a_j) = \Phi(a_i) \in \mathcal{P}(a_i) \subseteq \mathcal{P}(a_j)$. 

Notice that $\Phi(A) = \Phi(A')$.
We thus have $\card{\mathcal{E}_{\Phi'}(a_i)} = \card{\Phi'(A) \cap \mathcal{P}(a_i)} \leq \card{\Phi(A) \cap \mathcal{P}(a_j)} = \card{\mathcal{E}_{\Phi}(a_j)}$. 
Therefore, $\kappa^{\star}(\Phi') = \kappa^{\star}(\Phi) - \card{\mathcal{E}_{\Phi}(a_j)} + \card{\mathcal{E}_{\Phi'}(a_i)} \leq \kappa^{\star}(\Phi)$. Since $\Phi$ is optimal, we can conclude that $\Phi'$ is optimal as well. 

Now, we claim that the number of rogue pairs under $\Phi'$ is strictly less than that under $\Phi$, which will contradict the definition of $\Phi$. Notice first that $(a_i, a_j)$ is a rogue-pair under $\Phi$ but not under $\Phi'$. Consider $p, q \in [n]$ such that $(a_p, a_q)$ is a rogue pair under $\Phi'$, but not under $\Phi$. Then, either $q = i$ or $p = j$. If $q = i$, then $p < q = i < j$ and $(a_p, a_j)$ is a rogue pair under $\Phi$, which contradicts our choice of $i$. If $p = j$, then $j = p < q$, then $(a_i, a_q)$ is a rogue pair, which contradicts our choice of $j$. Thus the number of rogue pairs under $\Phi'$ is strictly less than that under $\Phi$, a contradiction. Hence, we conclude that $\Phi$ is rogue-free. 
\end{proof}

For an allocation $\fn{\Phi}{A}{H}$, let $S_{\Phi} = \set{a \in A ~|~ \Phi(a) \notin \mathcal{P}(a)}$ and $T_{\Phi} = \{a \in A ~|~ \Phi(a) \in \mathcal{P}(a)\}$. Note that $\set{S_{\Phi}, T_{\Phi}}$ is a partition of $A$ (with one of the parts possibly being empty). We say that $\Phi$ is nice if no agent in $S_{\Phi}$ envies any other agent in $S_{\Phi}$. Equivalently, $\Phi$ is nice if $\Phi(S_{\Phi}) \cap \mathcal{P}(S_{\Phi}) = \emptyset$. 

\begin{claim}
\label{claim:rogue-nice}
There exists an optimal rogue-free allocation that is also nice. %That is, such that $\Phi(S_{\Phi}) \cap \mathcal{P}(S_{\Phi}) = \emptyset$. 
\end{claim}
\begin{proof}
Let $\Phi$ be an optimal (i.e., $\kappa^{\star}(\mathcal{I}) = \kappa^{\star}(\Phi)$) rogue-free allocation that minimizes $\card{S_{\Phi}}$. Then, there exists $t_{\Phi} \in [n] \cup \set{0}$ such that $S_{\Phi} = \set{a_1,\ldots, a_{t(\Phi)}}$ and $T_{\Phi} = \set{a_{t(\Phi) + 1},\ldots,a_n}$.  Note that $S_\Phi$ is indeed contiguous. If not, then there exists indices $i$ and $j$ such that $i < j-1$ and $a_i, a_j \in S_{\Phi}$ and $a_p \notin S_\Phi$ for every index $p$ with $i<p<j$. Then $(a_{j-1}, a_j)$ is a rogue-pair.

If $\card{S_{\Phi}} \leq 1$, then the claim trivially holds. So, assume that $\card{S_{\Phi}} \geq 2$. Suppose that there exist $a_i, a_j \in S_\Phi$ such that $a_j$ envies $a_i$. Then, as $\mathcal{P}(a_j) \subseteq \mathcal{P}(a_{t(\Phi)})$, $a_{t(\Phi)}$ envies $a_i$ as well. Let $\Phi'$ be the allocation obtained from $\Phi$ by swapping the houses of $a_i$ and $a_{t(\Phi)}$. Then we have $S_{\Phi'} = S_{\Phi} \setminus \set{a_{t(\Phi)}}$ and $T_{\Phi'} = T_{\Phi} \cup \set{a_{t(\Phi)}}$. Thus, $\card{S_{\Phi'}} < \card{S_{\Phi}}$. Note that we constructed $\Phi'$ from $\Phi$ without introducing any new rogue-pairs. Additionally, we converted an envious agent under $\Phi$ (in particular, $a_{t(\Phi)}$) to an envy-free agent under $\Phi'$. This contradicts the optimality of $\Phi$ and the fact that $\Phi$ minimizes $|S_\Phi|$. 
\end{proof}

\begin{claim}
\label{claim:rogue-nice-envy}
Let $\Phi$ be a nice rogue-free allocation. Consider a house $h \in \Phi(T_{\Phi})$. Then, (1) the number of agents who envy $\Phi^{-1}(h)$ is exactly $|\{a ' \in S_{\Phi} ~|~ h \in \mathcal{P}(a')\}|$, and (2) $\kappa^{\star}(\Phi) = \sum_{h \in T_{\Phi}} \card{ \{a ' \in S_{\Phi} ~|~ h \in \mathcal{P}(a')\}}$. 
\end{claim}
\begin{proof}
By the definition of $T_{\Phi}$, no agent in $T_{\Phi}$ envies $\Phi^{-1}(h)$. Hence the number of agents who envy $\Phi^{-1}(h)$ is exactly equal to the number of agents in $S_{\Phi}$ who value $h$. This is precisely what assertion (1) says. Now, to compute $\kappa^{\star}(\Phi)$, for all $h \in \Phi(T_{\Phi})$, we only need to sum the number of agents in $S'$ value $h$. This is precisely what assertion (2) says.  
\end{proof}

{\bf Informal description of our algorithm:} Based on Claims~\ref{claim:rogue}-\ref{claim:rogue-nice-envy}, we are now ready to describe our algorithm. Informally, our algorithm works as follows. We are given an instance $\mathcal{I} = (A, H, \mathcal{P}; k)$. Suppose that $\Phi$ is the optimal allocation that we are looking for. By \Cref{claim:rogue-nice}, we can assume that $\Phi$ is rogue-free and nice. We guess $t_{\Phi}$. There are at most $n+1$ guesses. For the correct guess, we correctly identify $S_{\Phi}$ and $T_{\Phi}$. Then, $\Phi$ must allocate to each $a \in T_{\Phi}$ a house that $a$ values. To each agent $a' \in S_{\Phi}$, $\Phi$ must allocate a house that no agent in $S_{\Phi}$ values. For a house $h \in H$, the envy generated by allocating $h$ is precisely the number of agents in $S_{\Phi}$ who value $h$. We can thus reduce the problem to a minimum cost maximum matching problem, where the cost of matching each $h \in H$ to (1) an agent $a \in T_{\Phi}$ who values $h$ is precisely $\card{ \{a ' \in S_{\Phi} ~|~ a' \text{ values } h\}}$; (2) an agent $a \in T_{\Phi}$ who does not value $h$ is prohibitively high; (3) an agent $a' \in S_{\Phi}$ is $0$ if no agent in $S_{\Phi}$ values $h$, and prohibitively high otherwise. We can compute a minimum cost maximum matching in polynomial time.  

{\bf Algorithm:} We are given an instance $\mathcal{I} = (A, H, \mathcal{P}; k)$ as input. For each fixed $t \in [n] \cup \set{0}$, we do as follows. We partition $A$ into two sets $S$ and $T$ as follows:   $S = \set{a_1,\ldots,a_t}$ and $T = \set{a_{t+1},\ldots, a_n}$. We construct a complete bipartite graph $G_t^{\star}$, with vertex bipartition $A \uplus H$ and a cost function $c_t$ on the edges defined as follows: 

 \begin{equation*}
    c_t({\color{DarkSlateGray}(a,h)}) =
    \begin{cases}
      \card{\{a ' \in S ~|~ h \in \mathcal{P}(a')\}} & \text{if } a \in T \text{ and } a \text{ values } h,\\
      0 & \text{if } a \in S \text{ and no agent } a' \in S \text{ values } h, \\ 
      k+1 & \text{otherwise.}
    \end{cases}
\end{equation*}

If $G_t^{\star}$ contains a matching of size $n$ and cost at most $k$ for any $t \in [n] \cup \set{0}$, then we return that $\mathcal{I}$ is a yes-instance of [\nicefrac{0}{1}]-UHA. If $G_t^{\star}$ does not contain such a matching for any choice of $t \in [n] \cup \set{0}$, then we return that $\mathcal{I}$ is a no-instance of [\nicefrac{0}{1}]-UHA. 

{\bf Correctness:} To see the correctness of our algorithm, assume first that there exists $t \in [n] \cup \set{0}$ for which $G_t^{\star}$ contains a matching, say $M$, of size $n$ and cost at most $k$. Then, since $\card{M} = n$, $M$ saturates $A$. Consider an allocation $\fn{\Phi_M}{A}{H}$ defined as follows: for each $(a, h) \in M$, $\Phi_M$ allocates $h$ to $a$. We claim that $\kappa^{\star}(\Phi_M) = c_t(M) \leq k$. First, each $a \in T$ values $\Phi_M(a)$, for otherwise, $c_t((a, \Phi_M(a))) = k+1$, which is not possible. So, $a \in T$ does not envy any agent. Similarly, each $a' \in S$ does not value $\Phi_M(a'')$ for any $a'' \in S$, for otherwise, $c_t((a'', \Phi_M(a''))) = k+1$, which is not possible. 
So, for $a', a'' \in S$, $a'$ does not envy $a''$.  Also, for $(a', h) \in M$ with $a' \in S$, we have $c_t((a', h)) = 0$. Now, an agent $a' \in S$ may envy an agent $a \in T$. 
But note that every $h \in \Phi_M(T)$ contributes exactly $\card{\{a' \in S ~|~ h \in \mathcal{P}(a')\}} = c_t((a, h))$ , where $\Phi_M(a) = h$, to $\kappa^{\star}(\Phi_M)$. Hence, $\kappa^{\star}(\Phi_M) = \sum_{\substack{h \in \Phi_M(T) \\ (a, h) \in M}} c_t((a, h)) = \sum_{\substack{(a, h) \in M \\ a \in T}} c_t((a, h)) \leq k$. 

Conversely, assume that $\mathcal{I} = (A, H, \mathcal{P}; k)$ is a yes-instance. Let $\fn{\Phi}{A}{H}$ be an optimal allocation. By \Cref{claim:rogue-nice}, we assume without loss of generality that $\Phi$ is rogue-free and nice. Consider the iteration of our algorithm for which $t = t_{\Phi}$. Consider the matching $M_{\Phi}$ in $G_t^{\star}$ defined as $M_{\Phi}= \set{(a, \Phi(a)) ~|~ a \in A}$. We claim that $c_t(M_{\Phi}) = \kappa^{\star}(\Phi) \leq k$. First, consider $a \in S$. Since $\Phi$ is nice, no $a' \in S$ values $\Phi(a)$, which implies that $c_t(a, \Phi(a)) = 0$. Now, consider $a \in T$. Since $\Phi$ is rogue-free, $a$ values $\Phi(a)$. By \Cref{claim:rogue-nice-envy}, $\Phi(a)$ contributes exactly $\card{\{a' \in S ~|~ a' \text{ values } \Phi(a)\}} = c_t(a, \Phi(a))$ to $\kappa^{\star}(\Phi)$. Thus, $c_t(M_{\Phi}) = \sum_{(a, \Phi(a)) \in M_{\Phi}} c_t(a, \Phi(a)) = \sum_{\substack{(a, \Phi(a)) \in M_{\Phi}\\ a \in T}} c_t((a, \Phi(a))) = \sum_{a \in T} \card{\{a' \in S ~|~ a' \text{ values } \Phi(a)\}} = \kappa^{\star}(\Phi) \leq k$. 
\end{proof}

We now turn to the restricted setting where every agent likes exactly one house. In this case, the total envy is equal to the number of envious agents, that is, $\kappa^\star(\Phi) = \kappa^\#(\Phi)$, so the following result follows from~\Cref{lem:oha-onehouse}.

\begin{corollary}
\label{lem:uha-onehouse} There is a polynomial-time algorithm for~$[\nicefrac{0}{1}]$-\UHA{} when every agent approves exactly one house.
\end{corollary}

\subsection{Parameterized Results}
In this section, we discuss the parameterized results for UHA.
First, we design a linear kernel for $[\nicefrac{0}{1}]$-UHA.

\begin{theorem}
\label{lem:uha-kernel-n} $[\nicefrac{0}{1}]$-\UHA{} admits a linear kernel parameterized by the number of agents. In particular, given an instance of $[\nicefrac{0}{1}]$-\UHA{}, there is a polynomial time algorithm that returns an equivalent instance of $[\nicefrac{0}{1}]$-\UHA{} with at most twice as many houses as agents. 
\end{theorem}

\begin{proof}
It suffices to prove the safety of~\Cref{rr2}. Let $\mathcal{I} := (A,H,\mathcal{P}; k)$ denote an instance of HA. Further, let  $\mathcal{I}^\prime = (H^\prime := H \setminus X, A^\prime := A \setminus Y, \mathcal{P}^\prime; k)$ denote the reduced instance corresponding to $\mathcal{I}$. Note that the parameter for the reduced instance is $k$ as well. 

If $\mathcal{I}$ is a~\textsc{Yes}-instance of UHA, then there is an allocation $\Phi: A \rightarrow H$ with total envy at most $k$. By \Cref{claim:nicealloctotalenvy}, we may assume that $\Phi$ is a good allocation. This implies that the projection of $\Phi$ on $H^\prime \cup A^\prime$ is well-defined, and it is easily checked that this gives an allocation with total envy at most $k$.

On the other hand, if $\mathcal{I}^\prime$ is a~\textsc{Yes}-instance of UHA, then there is an allocation $\Phi^\prime: A^\prime \rightarrow H^\prime$ with total envy at most $k$. We may extend this allocation to $\Phi: A \rightarrow H$ by allocating the houses in $Y$ to agents in $X$ along the expansion $M$, that is:
 \begin{equation*}
    \Phi(a) =
    \begin{cases}
      \Phi^\prime(a) & \text{if } a \notin X,\\
      M(a) & \text{if } a \in X.
    \end{cases}
\end{equation*}

Since all the newly allocated houses are not valued by any of the agents outside $X$ and all agents in $X$ are envy-free with respect to $\Phi$, it is easily checked that $\Phi$ also has total envy at most $k$. 
\end{proof}

The following results follow from the algorithm described in~\Cref{prop:mequalsn-uha} after guessing the allocated houses, which adds a multiplicative overhead of $\binom{m}{n}\leq 2^m$ to the running time.

\begin{corollary}
\label{lem:uha-fpt-m} $[\nicefrac{0}{1}]$-\UHA{} is fixed-parameter tractable when parameterized either by the number of houses or the number of agents. In particular, $[\nicefrac{0}{1}]$-\UHA{} can be solved in time $O^\star(2^m)$.
\end{corollary} 

% Guess unallocated houses. m \choose (m-n)
% Can do min-cost max-matching here

% \begin{proof}

% \end{proof}

\begin{corollary}
\label{lem:uha-fpt-m-rankings} \sloppypar $[\succ]$-\UHA{} is fixed-parameter tractable when parameterized by the number of houses and can be solved in time $O^\star(2^m)$.\endsloppypar
\end{corollary}

\begin{table}
\begin{tabular}{c|ccc|ccc|c}
\multicolumn{1}{l|}{($n, m, n^\star$)} &
  \multicolumn{3}{c|}{OHA} &
  \multicolumn{2}{c}{EHA} &
  \multicolumn{1}{l|}{} &
  \multicolumn{1}{l}{\begin{tabular}[c]{@{}l@{}}Time/Instance\\ OHA, EHA\\ (in sec.)\end{tabular}} \\ \hline
\multicolumn{1}{l|}{} &
  \begin{tabular}[c]{@{}c@{}}Env.\\  Agents\\ ($\kappa^\star(\Phi))$\end{tabular} &
  \begin{tabular}[c]{@{}c@{}}Max \\ Envy\end{tabular} &
  \begin{tabular}[c]{@{}c@{}}Total\\  Envy\end{tabular} &
  \multicolumn{1}{l}{\begin{tabular}[c]{@{}l@{}}Env.\\ Agents\end{tabular}} &
  \begin{tabular}[c]{@{}c@{}}Max \\ Envy\\ ($\kappa^\dagger(\Phi)$)\end{tabular} &
  \multicolumn{1}{l|}{\begin{tabular}[c]{@{}l@{}}Total\\ Envy\end{tabular}} &
  \multicolumn{1}{l}{} \\ \hline
(30, 30, 1)    & 15.11 & 14.89 & 216.71 & 15.11 & 14.89 & 216.71 & 0.003, 0.002 \\
(30, 30, 5)    & 0.95  & 8.78  & 12.33  & 9.07  & 7.56  & 11.29  & 0.19, 0.10   \\
(30, 30, 15)   & 0     & 0     & 0      & 0     & 0     & 0      & 0.20, 0.33   \\
(30, 40, 1)    & 10.18 & 19.82 & 191.76 & 20.18 & 9.82  & 188.16 & 0.01, 0.003  \\
(60, 60, 1)    & 30.36 & 29.64 & 888.08 & 30.36 & 29.64 & 888.08 & 0.006, 0.002 \\
(60, 60, 15)   & 0.01  & 0.31  & 0.31   & 0.21  & 0.21  & 0.21   & 0.57, 0.10   \\
(60, 60, 30)   & 0     & 0     & 0      & 0     & 0     & 0      & 2.00, 4.21   \\
(120, 120, 1)  & 59.45 & 60.55 & 3567.8 & 59.45 & 60.55 & 3567.8 & 0.01, 0.002  \\
(120, 120, 5)  & 3.83  & 57.07 & 218.79 & 66.62 & 51.07 & 200.63 & 0.11, 0.20   \\
(120, 120, 15) & 0     & 0     & 0      & 0     & 0     & 0      & 1.98, 4.26   \\
(120, 130, 5 ) & 0     & 0     & 0      & 0     & 0     & 0      & 0.11, 0.10   \\ \hline
\end{tabular}
\caption{A summary of the results, averaged over 100 instances of each type. The OHA column corresponds to the solution from OHA ILP and the max-envy and total envy in that column shows those values when the number of envious agents is minimized. Similarly for the EHA column.}
\label{table:ilp}
\end{table}

\section{Experiments}
\label{sec:exp}

We implemented the ILP for OHA and EHA over synthetic datasets of house allocation problems generated uniformly at random. We used Gurobi Optimizer version $9.5.1$\footnote{The code can be accessed at https://github.com/anonymous1203/House-Allocation}.
The average was taken over $100$ trials for each instance. A summary is recorded in \Cref{table:ilp}. For a fixed number of houses and agents, notice that as the number of agent types, $n^\star$ increases, the number of envious agents and the maximum envy decreases. Instances with identical valuations (where $n^\star = 1$) seem to admit more envy than the other extreme (where $n^\star = n$). This is due to the fact when valuations are identical, there is more contention on the specific subset of goods. On the contrary, for instances with $m^\star=1$, envy-free allocations always exist. Indeed, when $m^\star=1$, all houses are of the same type, which means that an agent either likes all the houses or dislikes all of them and in either case, she is envy-free no matter which house she gets.
Also note that, when we increase the number of houses, for a constant number of agents and agent types, the envy decreases, which is as expected, because of the increase in the number of choices and the fact that some houses (the more contentious ones) remain unallocated.

\section{Price of Fairness}
\label{sec:pof}
In this section, in addition to the envy-minimization, we will focus on the social welfare of an allocation, as captured by the sum of the individual agent utilities. An allocation is considered more efficient when it results in a higher level of social welfare. Minimizing the envy objectives can lead to inefficient allocations with poor social welfare. 
Indeed, our algorithms for OHA, EHA and UHA first check if there are enough (more than $n$) dummy houses and if so, allocate these dummy houses to everyone, potentially leading to an envy-free solution, but with no social welfare gain. Quantifying this welfare loss, incurred as the cost of minimizing envy is, therefore, an imperative consideration. In particular, we discuss the worst-case welfare loss under different scenarios when any of the envy objectives is supposed to be minimized, and give tight bounds for the same.

We first define the Price of Fairness in the house allocation setting as follows.
We use the notation $PoF_{OHA}$ to denote the fact that the fairness notion under consideration is the minimum number of envious agents. $PoF_{EHA}$ and $PoF_{UHA}$ are defined analogously. When the meaning is clear from the context, we drop the subscript and simply write $PoF$.

%\JM{NEW: We should perhaps specify in the above sentence that we're  defining PoF w.r.t. the number of envious agents. Maybe we should even use the notation $PoF_{OHA}$ or something like that. Just to emphasize and avoid confusion. And we can add a line saying that the other ones are defined analogously, and when the meaning is clear from the context, we drop the subscript and simply write $PoF$. \AS{Done.} Also, we don't need the $k$ in $\mathcal{I}:= (A,H,\mathcal{P}; k)$ in the definition below. And the supremum is over all instances with $m$ houses and $n$ agents, right? We should specify that also.} \AS{Added a line post the definition to clarify this.}

\begin{definition}
    For a house allocation instance $\mathcal{I}:= (A,H,\mathcal{P})$ with $n$ agents and $m$ houses, consider an allocation $\Phi^\star$ that maximizes the social welfare, denoted by $SW(\Phi^\star)$. Let $\Phi$ be the allocation that minimizes the number of envious agents and $SW(\Phi)$ be the social welfare of $\Phi$. Then, the price of fairness $PoF_{OHA}$ is defined as $$PoF_{OHA} = \sup_\mathcal{I} \frac{SW(\Phi^\star)}{SW(\Phi)} = \sup_\mathcal{I} \frac{\sum_{i \in A} u_i(\Phi^\star(i))}{\sum_{i \in A} u_i(\Phi(i))},$$
where the supremum is taken over all instances with $n$ agents and $m$ houses.
\end{definition}

 We say that an instance with binary valuations is normalized if every agent likes an equal number of houses. Moreover, if every house is also liked by an equal number of agents, then we say that the instance is doubly normalized. We now present the bounds for $PoF$. We show that if $m=n$, then $PoF=1$ for all the three envy minimization objectives. When $m>n$, if the instance is doubly normalized, then $PoF=1$ but it can be as large as $n$ if we drop the double normalization assumption. The results are summarized in \Cref{tab:pof_table}.

% \usepackage{multirow}

% \usepackage{multirow}
\begin{table}
\centering
\begin{tabular}{|c|c|c|c|}
\hline
                                                            & $\mathbf{m=n}$ & \multicolumn{2}{c|}{$\mathbf{m>n}$} \\ \hline
                                                            & \begin{tabular}[c]{@{}c@{}}\textbf{Non}\\ \textbf{Normalized}\end{tabular} & \begin{tabular}[c]{@{}c@{}}\textbf{Doubly}\\ \textbf{Normalized}\end{tabular} & \textbf{Normalized} \\ \hline
\begin{tabular}[c]{@{}c@{}} \textbf{\textsc{OHA / EHA /}} \\ \textbf{\textsc{UHA}}\end{tabular} & {\cellcolor[HTML]{B3E7B4}{\color[HTML]{343434}$1$}} & {\cellcolor[HTML]{B3E7B4}{\color[HTML]{343434}$1$}} & {\cellcolor[HTML]{ECCECE}{\color[HTML]{343434}$\frac{n}{2} \leq \text{PoF} \leq n$}} \\ \hline
\end{tabular}
\caption{Price of minimizing the number of envious agents, the maximum envy, and the total envy for binary valuations.}
\label{tab:pof_table}
\end{table}

Towards proving our results for the case when $m = n$, we first present the following lemma based on the characterization that a matching $M$ in a graph $G$ is maximum if and only if $G$ has no $M$-augmenting path. An $M$-augmenting path is a path in $G$ that starts and ends at unmatched vertices (vertices not included in $M$) and alternates between edges in the matching $M$ and edges that are not in $M$.

\begin{proposition}[folklore]
\label{lem:max-matching-saturating-vertices}
Let $G$ be a graph. For any $A' \subseteq V(G)$, if $G$ contains a matching that saturates $A'$, then $G$ contains a maximum matching that saturates $A'$. 
\end{proposition}

\begin{proof} Let $M'$ be the matching that saturates $A'$ in $G$. If $M'$ is itself a maximum matching, then we are done. Suppose not. Then there exists an $M'$-augmenting path $P$ in $G$ which starts and ends at a vertex in $A \setminus A'$. Replacing the edges of $M$ in $P$ by the other edges in $P$ gives a strictly larger matching $M$ (with exactly one more edge) which saturates $A'$ and in addition, saturates two vertices from $A \setminus A'$. For every such augmentation, the set of saturated vertices increases by two, keeping the original set of saturated vertices intact. When no augmenting path exists, then by the characterization of maximum matchings, the resulting matching $M$ must be a maximum matching that saturates $A'$. 
\end{proof}

We now consider the case when the number of houses is equal to the number of agents, and show that $PoF = 1$ in this case. Recall that we are under the assumption that every agent values at least one house. Hence, when $m = n$ (and the valuations are binary), in any allocation, an agent either receives a house she values or she is envious; in particular, if agent $a$ is envious, then she envies exactly $d(a)$ other agents, where $d(a)$ is the degree of $a$ in the associated preference graph. We will crucially rely on this fact to prove that $PoF = 1$. Our arguments will also use the correspondence between matchings in the preference graph and allocations: For a matching $M$ in the preference graph, we denote the allocation corresponding to $M$ by $\Phi_M$, which allocates the house $M(a)$ to the agent $a$ and allocates houses to the remaining unmatched agents arbitrarily. Notice that the allocation $\Phi_M$ need not be unique. 
% \JM{NEW: I added a few lines in the above paragraph. Please  read it over once.} 

% \JM{Maybe we should separate the following theorem into three results depending on the fairness objective. The current phrasing could be needlessly confusing. The first result should say, there is an allocation that simultaneously maximizes welfare and minimizes the number of agents. Second one about welfare and maximum envy and like that. We can have these three lemmas, and then one theorem that simply says PoF = 1 for all three fairness concepts.}
% \JM{Also, maybe we should avoid using the terms OHA-optimal, EHA-optimal, UHA-optimal. We can simply say allocation that minimizes the number of envious agents (or maximum envy or total envy). I mean, the term optimal could be confusing, and even if it is not, could invite a negative comment from reviewers. Let's not give them another reason to demand that we change the names of the problems.}
% \AS{Sure, done}

\begin{lemma}
\label{lem:welfareOHA} 
For an instance of house allocation with $m=n$ and binary valuations, there exists an allocation that simultaneously maximizes the social welfare and minimizes the number of envious agents.
\end{lemma}

\begin{proof}
By \Cref{prop:mequalsn-oha}, we know that any allocation that minimizes the number of envious agents has exactly $|M|$ envy-free agents, where $M$ is the maximum matching in the associated preference graph $G$. It is easy to see that the maximum social welfare of the instance is also exactly $|M|$, hence any allocation, say $\Phi_M$, corresponding to a maximum matching $M$ in $G$ maximizes the welfare and minimizes the number of envious agents.
\end{proof}

\begin{lemma}
\label{lem:welfareEHA} 
For an instance of house allocation with $m=n$ and binary valuations, there exists an allocation that simultaneously maximizes social welfare and minimizes the maximum envy.
\end{lemma}

\begin{proof} Let $M$ be a maximum matching in the associated preference graph $G$. Since $m=n$, all the agents that are unmatched under $M$ are envious under $\Phi_M$. This implies that at least $n - |M|$ agents are envious. In particular, if $|M| = n$, then under the allocation $\Phi_M$, all agents are envy-free, and $\Phi_M$ simultaneously maximizes welfare and minimizes the maximum envy, and thus the lemma trivially holds. So, assume that $|M| < n$. As $M$ is a maximum matching, we can conclude that there is no matching that saturates all the agents. 

Let us now order the agents in the non-increasing order of their degrees, that is, $a_1, a_2,\ldots, a_n$ such that  $d(a_1) \geq d(a_2) \geq \cdots \geq d(a_n)$. Let $p_1 \in [n]$ be the least index such that there is no matching that saturates all of $a_1, a_2,...,a_{p_1}$. This implies that there is a matching that saturates $a_1, a_2, \ldots a_{p_1-1}$ but none that saturates the agents $a_1, a_2,...,a_{p_1}$. That is, in any allocation, at least one agent among $a_1, \ldots a_{p_1}$ is envious. Thus, the maximum envy of the instance is at least $d(a_{p_1})$. But there is a matching that saturates $a_1, a_2,\ldots, a_{p_1-1}$, which can be extended to a corresponding allocation; and under such an allocation, $a_{p_1}$ would be the first envious agent (first in the ordering $a_1, a_2,\ldots, a_n$). Therefore, we can conclude that the maximum envy of the instance is precisely equal to $d(a_{p_1})$. Now, we only have to prove that there is indeed a \emph{maximum} matching that saturates $a_1, a_2,\ldots,a_{p_1-1}$; and \Cref{lem:max-matching-saturating-vertices} guarantees that $G$ does contain such a maximum matching (we simply need to apply \Cref{lem:max-matching-saturating-vertices} with $A' = \set{a_1, a_2,\ldots, a_{p_1-1}}$). This implies that there is a welfare-maximizing allocation that also minimizes the maximum envy. 
\end{proof}

\begin{lemma}
\label{lem:welfareUHA} 
For an instance of house allocation with $m=n$ and binary valuations, there exists an allocation that simultaneously maximizes social welfare and minimizes total envy.
\end{lemma}

\begin{proof}
Consider an allocation, say $\Phi$, that minimizes the total envy. We claim that the matching $M$ corresponding to $\Phi$ in the associated preference graph $G$ is a maximum matching. If $M$ is not a maximum matching, then we have an augmenting path, and we get a strictly larger matching $M^\prime$ such that the vertices saturated by $M$ are also saturated under $M^\prime$. Now consider the allocation corresponding to $M^\prime$, say $\Phi_{M^\prime}$. Since $|M|<|M^\prime|$, the number of envious agents under $\Phi_{M^\prime}$ is strictly less than those under $\Phi$; also all envy-free agents under $\Phi$ remain envy-free under $\Phi_{M^\prime}$. These arguments imply that the total envy under $\Phi_{M^\prime}$ is less than the total envy under $\Phi$. This is a contradiction to the fact that $\Phi$ minimizes the total envy. 
\end{proof}

The following result now follows from $\Cref{lem:welfareOHA}$, $\Cref{lem:welfareEHA}$ and $\Cref{lem:welfareUHA}$.
\begin{theorem}
\label{thm:PoF-m-equals-n}
     For an instance of house allocation with $m=n$ and binary valuations, $PoF = 1$ for all the three envy-minimization objectives.
\end{theorem}

\Cref{thm:PoF-m-equals-n} tells us that when $m = n$, there is an allocation that simultaneously maximizes welfare and minimizes any \emph{one} of the three measures of envy. This does raise the following question: Can we simultaneously maximize welfare and minimize \emph{all} three measures of envy? We show that we can indeed do this; that is, there is an allocation that simultaneously minimizes the number of envious agents, the maximum and total envy while maximizing social welfare. 

% \JM{NEW: I added a couple of sentences above. And added the second sentnce in the theorem statement below.}

\begin{theorem}
For an instance of house allocation with $m=n$ and binary valuations, there is an allocation that simultaneously minimizes the number of envious agents, the maximum and total envy, while maximizing social welfare. Moreover, we can compute such an allocation in polynomial time. 
\end{theorem}

\begin{proof} Let $M$ be a  maximum matching in the associated preference graph $G$. First, the corresponding allocation $\Phi_M$ maximizes social welfare. Since $m=n$, all the agents that remain unmatched under $M$ are envious under $\Phi_M$. This implies that at least $n - |M|$ agents are envious under any allocation. Thus, if $\card{M} = n$, then every agent is envy-free under $\Phi_M$, and thus, the theorem trivially holds. So, assume from now on that $\card{M} < n$. 

We first order the agents in the non-increasing order of their degrees, that is, $a_1, a_2,\ldots, a_n$ such that  $d(a_1) \geq d(a_2) \geq \cdots \geq d(a_n)$. Let $p_1 \in [n]$ be the least index such that $G$ does not contain a matching that saturates $\{a_1, a_2,...,a_{p_1}\}$. That is, there is a matching that saturates all of $a_1, a_2, \ldots a_{p_1-1}$ but none that saturates all of  $a_1, a_2,...,a_{p_1}$. Now, let $p_2 \in [n]$  be the  least index such that $p_2 > p_1$ and $G$ does not contain a matching that saturates $\{a_1, a_2, \ldots a_{p_2}\} \setminus \{a_{p_1} \}$ (if such an index $p_2$ exists); In general, having defined $p_1, p_2,\ldots, p_{i - 1}$, we define $p_i$ to be the least index in $[n]$ such that $p_i > p_{i - 1}$ and $G$ does not contain a matching that saturates $\{a_1, a_2, \ldots a_{p_i}\} \setminus \{a_{p_1}, a_{p_2},\ldots, a_{p_{i - 1}} \}$. Let $p_1, p_2,\ldots, p_k$ (for some $k \geq 1$), be the indices defined this way.  By their definition, $G$ contains a matching, say $M'$, that saturates $\{a_1, a_2, \ldots a_n\} \setminus \{a_{p_1}, a_{p_2}, \ldots a_{p_k}\}$. Consider the corresponding allocation $\Phi_{M^\prime}$; recall that $\Phi_{M^\prime}$ allocates houses along the edges of $M^\prime$ and the remaining agents, i.e., $a_{p_1} \ldots a_{p_k}$, receive the remaining houses in an arbitrary manner. %The total envy under $\Phi_{M^\prime}$ is $d(a_{p_1})+d(a_{p_2}) + \ldots d(a_{p_k})$.

Before we proceed further, let us observe that we can indeed compute $M^\prime$ and hence the corresponding allocation $\Phi_{M^\prime}$ in polynomial time. The arguments we have used so far are constructive. To compute $M^\prime$, we only need to identify the indices $p_1, p_2,\ldots, p_k$ and find a matching that saturates $\set{a_1, a_2,\ldots, a_{n}} \setminus \set{a_{p_1}, a_{p_2},\ldots, a_{p_k}}$. Notice that to find each $p_i$, we only need to check if $G$ contains a matching that saturates $\set{a_1, a_2,\ldots, a_j} \setminus \set{a_{p_1}, a_{p_2},\ldots, a_{p_{i - 1}}}$ for each $j > p_{i -1}$, which we can do in polynomial time. 

We will now show that the allocation $\Phi_{M^\prime}$ satisfies all the properties required by the statement of the theorem; that is, $\Phi_{M^\prime}$ maximizes welfare and minimizes the number of envious agents, the maximum envy and the total envy. To that end, observe first that $M^\prime$ does not saturate any of the agents $a_{p_1}, a_{p_2},\ldots, a_{p_k}$. Otherwise, let $p_i$ be the least index in $\set{p_1, p_2,\ldots, p_k}$ such that $M^\prime$ saturates $a_{p_i}$. In particular, $M^\prime$ is a matching that saturates $\set{a_1, a_2,\ldots, a_{p_i}} \setminus \set{a_{p_1}, a_{p_2},\ldots, a_{p_{i - 1}}}$, which contradicts the definition of $p_i$. By the same reasoning, we can also conclude that that $M^\prime$ is indeed a maximum matching. If not, then there is a larger matching, say $M''$, which also saturates at least one of the $a_{p_i}$ in addition to saturating $\{a_1, a_2, \ldots a_n\} \setminus \{a_{p_1}, a_{p_2}, \ldots a_{p_k}\}$ (by \Cref{lem:max-matching-saturating-vertices}), which again will lead to a contradiction.  %assume without loss of generality that $p_i$ is the least index among $\set{p_1, p_2,\ldots, p_k}$ such that $M''$ saturates $a_{p_i}$. In particular, $M''$ saturates $\set{a_1, a_2,\ldots, a_{p_i}} \setminus \set{a_{p_1}, a_{p_2},\ldots, a_{p_{i - 1}}}$, which contradicts the definition of $p_i$. \JM{NEW: Notation: Used $M''$ here because we already used $M$ in the beginning.}

Now, the fact that $M^\prime$ is a maximum matching immediately implies that the corresponding allocation $\Phi_{M^\prime}$  maximizes social welfare and minimizes the number of envious agents. We will now use the same arguments we used in the proof of~\Cref{lem:welfareEHA} to show that $\Phi_{M^\prime}$ minimizes the maximum envy. 
Notice that the set of agents that are \emph{not} saturated by $M'$ is precisely $\set{a_{p_1}, a_{p_2}, \ldots, a_{p_k}}$. This fact, along with the fact that $M^\prime$ is a maximum matching, implies that the set of envious agents under $\Phi_{M^\prime}$ is precisely $\set{a_{p_1}, a_{p_2}, \ldots, a_{p_k}}$. In particular, $a_{p_1}$ is the first envious agent under $\Phi_{M^\prime}$ (again, first in the ordering $a_1, a_2,\ldots, a_n$). Thus, the maximum envy of the allocation $\Phi_{M^\prime}$ is $d(a_{p_i})$. By the definition of $p_1$, in any allocation, at least one agent among $a_1, \ldots a_{p_1}$ is envious, and hence the maximum envy in any allocation is at least $d(a_{p_i})$ (because $d(a_1) \geq d(a_2) \geq \cdots \geq d(a_n)$). From these arguments, we can conclude that $\Phi_{M^\prime}$ minimizes the maximum envy. 

To complete the proof, now we only need to argue that $\Phi_{M^\prime}$ minimizes the total envy. To that end, we first prove the following claim. 
\begin{claim}
\label{claim:at-least-i-unmatched}
For every $i \in [k]$, every matching in $G$ does \emph{not} saturate at least $i$ agents from the set $\set{a_1, a_2,\ldots, a_{p_i}}$. Or equivalently, for every $i \in [k]$, at least $i$ agents from the set $\set{a_1, a_2,\ldots, a_{p_i}}$ are envious under every allocation. 
\end{claim}
\begin{proof}
Suppose for a contradiction that $G$ contains a matching $M'''$ such that the number of agents from $\set{a_1, a_2,\ldots, a_{p_i}}$ that are not saturated by $M'''$ is strictly less than $i$. Therefore, $M'''$ saturates at least one agent from the set $\set{a_{p_1}, a_{p_2},\ldots, a_{p_i}}$; 
let $p_j$ be the least such index such that $M'''$ saturates $a_{p_j}$. But then $M'''$ saturates $\set{a_1, a_2,\ldots, a_{p_j}} \setminus \set{a_{p_1}, a_{p_2},\ldots, a_{p_{j - 1}}}$, which contradicts the definition of $p_j$. 
Notice that the second sentence in the statement of the claim is merely a restatement of the first, because of the correspondence between matchings and allocations: Under any allocation, the envy-free agents and the houses they received form a matching in $G$.  
\end{proof}
% \JM{NEW: Please do double check the claim and its proof.}

Now, to complete the proof of the theorem, recall that the set of envious agents under $\Phi_{M^\prime}$ is precisely $\set{a_{p_1}, a_{p_2}, \ldots, a_{p_k}}$, and thus the total envy under $\Phi_{M^\prime}$ is precisely $\sum_{i = 1}^k d(a_{p_i})$. 
Also, as we have already argued, $\Phi_{M^\prime}$ minimizes the number of envious agents, and hence we can conclude that at least $k$ agents are envious under \emph{every} allocation. 
Assume now for a contradiction that $\Phi_{M^\prime}$ does not minimize the total envy. 
Let $\Phi$ be an allocation that minimizes the total envy. 
Again, at least $k$ agents are envious under $\Phi$; let $1 \leq q_1 < q_2 <\cdots < q_k \leq n$ be such that the agents $a_{q_1}, a_{q_2},\ldots, a_{q_k}$ are the first $k$ envious agents under $\Phi$ (first in the ordering $a_1, a_2,\ldots, a_n$). 
Thus the total envy under $\Phi$ is at least $\sum_{i = 1}^k d(a_{q_i})$. 
Now, by our assumption, $\Phi_{M'}$ does not minimize the total envy and $\Phi$ does, and thus $\sum_{i = 1}^k d(a_{p_i}) > \sum_{i = 1}^k d(a_{q_i})$; we will derive a contradiction from this. Since $\sum_{i = 1}^k d(a_{p_i}) > \sum_{i = 1}^k d(a_{q_i})$, there exists an index $i \in [k]$ such that $d(a_{p_i}) > d(a_{q_i})$; let $i \in [k]$ be the least index such that $d(a_{p_i}) > d(a_{q_i})$. Hence $p_i < q_i$, which implies that the set $\set{a_1, a_2,\ldots, a_{p_i}}$ contains at most $i - 1$ of the agents $a_{q_1}, a_{q_2},\ldots, a_{q_i}$. But by the definition of the $q_j$s, the agents $a_{q_1}, a_{q_2},\ldots, a_{q_k}$ are the first $k$ envious agents under $\Phi$. We can thus conclude that the set $\set{a_1, a_2,\ldots, a_{p_i}}$ contains at most $i - 1$ agents who are envious under $\Phi$, which contradicts Claim~\ref{claim:at-least-i-unmatched}. 

We have thus shown that $\Phi_{M'}$ minimizes the total envy, and this completes the proof of the theorem. 
\end{proof}

We now consider the case when $m>n$ and the setting of the doubly normalized valuations. The following result follows from the proof of Theorem $5$ in \cite{10.1007/978-3-031-43254-5_16}, where they show that for such structured valuations, every good can be assigned non-wastefully such that every agent derives the value of at least one. In our setting, we can find an envy-free allocation that is non-wasteful and hence, welfare maximizing as well, proving the following result.

\begin{corollary}
    For $m>n$ and doubly normalized binary valuations, $PoF = 1$ for all three envy-minimization objectives.
\end{corollary}

It is easy to see that when there are dummy houses, then there can be instances where the fair allocation can be highly inefficient. Consider the case when number of dummy houses is at least the number of agents. Then, irrespective of the individual valuations, every agent can get a dummy house and be envy-free, leading to no social welfare at all. But, the following result suggests that even when there aren't any dummy houses, there are instances where the price of fairness can be high.

\begin{theorem}
\label{lem:pofisn}
For $m>n$ and binary normalized valuations, $PoF  = \Theta(n)$, for all the three envy optimization objectives, even when there are no dummy houses to begin with.
\end{theorem}

\begin{proof} 
We first show that $PoF \geq \nicefrac{n}{2}$. We say that a house is allocated non-wastefully if the receiving agent values it at $1$. Consider an instance with $2n$ agents and $3n$ houses, such that $n>2$. Agent $1$ likes the first $n$ houses, $2$ likes the next $n$ houses, and all of the remaining $2n-2$ agents collectively like the last $n$ houses. The instance does not have any dummy houses. Note that any welfare-maximizing allocation $\Phi^\star$ allocates $n+2$ houses non-wastefully. So $SW(\Phi^\star) \geq n+2$. But in this instance, there exists an envy-free allocation $\Phi$ that allocates exactly $2$ houses non-wastefully to the first two agents. The last $n$ houses remain unallocated under $\Phi$, as they create Hall's violator for the set of last $2n-2$ agents. Therefore, the last $2n-2$ agents receive the houses they don't like, and the set of their liked houses remains unallocated. This implies that $SW(\Phi) \leq 2$. The PoF for this particular instance is, therefore, $\nicefrac{(n+2)}{2} \approx n$. Since the instance had $2n$ agents to begin with, we have that in general, $PoF \geq \nicefrac{n}{2}$.

We now show that $PoF \leq n$. Under any arbitrary instance with $n$ agents and binary valuations, the maximum possible social welfare under any allocation is $n$. Consider an envy-free allocation $\Phi$, if it exists. We claim that there is at least one house that is allocated non-wastefully under $\Phi$. Suppose not. Then every house is allocated wastefully, and hence, $SW(\Phi) = 0$. But since there are no dummy houses, there is at least one agent $a$ who likes (a subset of) the allocated houses. But since welfare is zero, $a$ is an envious agent, which contradicts the fact that we started with an envy-free allocation. Therefore, $SW(\Phi) \geq 1$. This implies that $PoF \leq n$.  

If an envy-free allocation does not exist, consider $\Phi$ to be an allocation that minimizes the number of envious agents. Suppose the number of envious agents under $\Phi$ is $k$ and let $a$ be one such envious agent. If $SW(\Phi) = 0$, then everyone, including $a$ receives a house that they do not like. Now $a$ is envious because one of his liked houses, say $h$ is allocated to some agent $\Phi^\prime$ and that too wastefully. Consider the re-allocation of the house $h$ to $a$ and the house $\Phi(a)$ to $a^\prime$. Then $a$ becomes envy-free, and no agent (including $a^\prime$) becomes newly envious of this re-allocation, since the set of allocated houses is exactly the same and $SW(\Phi) = 0$. This re-allocation has $k-1$ envious agents, which contradicts the fact that we started with an optimal allocation. Therefore, $SW(\Phi) \geq 1$. This implies that $PoF \leq n$. The argument for EHA and UHA is analogous. This settles our claim.
\end{proof}

In the proof of $\Cref{lem:pofisn}$,  although the number of agent types in the lower bound instance is $3$, the $PoF$ is linear in the number of agents. This implies that we do not expect better bounds for instances with few agent types.

\section{Conclusion}
\label{sec:con}

We studied three kinds of natural quantifications of envy to be minimized in the setting of house allocation: OHA, EHA, and UHA. Most of our results are summarized in~\Cref{tab:my-table}. We also show that both OHA and EHA are \FPT{} when parameterized by the total number of house types or agent types. We also look at the price of fairness in the context of house allocation and give tight bounds for the same. We leave several questions related to UHA open. The complexity of OHA for strict rankings also remains open. Also, finding interesting classes of structured input---beyond extremal intervals---for which these problems are tractable is an important direction for future work.

\newpage
\bibliography{main}
\end{document}